\documentclass{LMCS} 

\usepackage[english]{babel}
\usepackage[latin1]{inputenc}
\usepackage{amssymb}
\usepackage{latexsym}
\usepackage{enumerate}

\newtheorem{claim}{Claim}
\newtheorem{claimapp}{Claim}

\newcommand{\Iff}{\Leftrightarrow}
\newcommand{\id}{ \textrm{id}} 
\newcommand{\rank}{ \textrm{rank}}
\newcommand{\Eval}{ \textrm{Eval}}

\newcommand{\upset}[1]{\uparrow \!\!{#1}}

\newcommand{\nat}{\mathbb{N}}

\newcommand{\ck}{{\bf K }}

\newcommand{\cdomain}{{\bf Dom }}
\newcommand{\ccbdomain}{{\bf \omega Dom }}
\newcommand{\cdom}{{\bf Dom^{e} }}
\newcommand{\ccbdom}{{\bf \omega Dom^{e} }}

\newcommand{\cdomainwp}{{\bf PER(Dom) } }

\newcommand{\cclcdomwp}{{\bf clcDP} }

\newcommand{\cqcb}{{\bf QCB_{0} }}

\newcommand{\funcf}{{\bf F }}
\newcommand{\funcg}{{\bf G }}

\begin{document}

\def\doi{6 (3:9) 2010}
\lmcsheading%
{\doi}
{1--47}
{}
{}
{Oct.~22, 2009}
{Aug.~25, 2010}
{}   

\title[Strictly positive induction]{Domain representable spaces defined by strictly positive induction}

\author[P.~K.~Køber]{Petter Kristian Køber}
\address{University of Oslo, Department of Mathematics, PO Box 1053, Blindern, N-0316 Oslo, Norway}
\email{petterk@math.uio.no}




\begin{abstract}
Recursive domain equations have natural solutions.
In particular there are domains defined by strictly positive induction.
The class of countably based domains gives a computability theory for possibly non-countably based  topological spaces. 
A  $ qcb_{0} $ space is  a  topological space characterized by its strong representability over domains.

In this paper, we study strictly positive inductive definitions for $ qcb_{0} $ spaces by  means of domain representations, i.e.
we show that there exists a canonical fixed point of every strictly positive operation on $qcb_{0} $ spaces.
\end{abstract}

\maketitle

\section*{Introduction}

\noindent The domains we consider in this paper are consistently complete, algebraic cpos, so called Scott domains \cite{Sco_72}.
The initial motivation for these domains was to provide a denotational semantic for the $\lambda$-calculus.

An important aspect of domain theory is the existence of solutions of recursive domain equations, which are equalities
 between terms built from certain basic operations and a finite list of parameters.
Category theory is applied to solve  recursive domain equations, and the category $\cdom$ used has domains as objects
and embedding-projection pairs as morphisms. The canonical solution is a least fixed point of some functor over
$\cdom$ \cite{SP_82} and  occurs as the limit of an inductively defined  $\omega$-chain of domains.

Some recursive domain equations can be solved iteratively within set theory \cite{SLG_94}.
These solutions are referred to as positive inductive definitions. We will focus on 
 definitions  by {\em strictly} positive induction, where all function spaces involved  have fixed input domains.
This is a  natural  restriction from a computer science point of view, which is also technically beneficial.
A fundamental example is the domain $ D = A + [ B \rightarrow D ] $, with $ A $ and $ B$ some parameters.

A domain representation of a topological space $X$ is a triple $ (D, D^R, \delta) $, where $D$ is a domain, $D^R $ is a 
subspace of the domain and $\delta: D^R \to X $ is a continuous representation map. Countably based domains carry a natural notion of 
computability. Via domain representations we get a computability theory for a wide range of topological spaces \cite{Bla_00,ST_07}.

The topological $ T_{0} $ quotients of countably based spaces, the {\em $qcb_{0} $ spaces}, form an interesting class of topological 
spaces \cite{BSS_07,ELS_04}. The category $\cqcb$ of $ qcb_{0} $ spaces with continuous functions is Cartesian closed,
so it admits finite products as well as an exponential. 
$Qcb_{0} $  spaces have been characterised as the topological spaces with an admissible quotient TTE representation \cite{Sch_03}.
This result has been generalised to admissible quotient domain representations \cite{Ham_05}.

A quotient domain representation might as well be considered as  a domain with a partial equivalence relation.
The class of domains with partial equivalence relations is also of great interest in its own right \cite{BBS_04,MS_02}. It is  strongly
related to domains with totality \cite{Ber_93}.

In this paper we show that we can define $qcb_{0} $ spaces  by strictly positive induction. The fundamental example is
$ X = A \uplus [ B \Rightarrow^s X ] $, with $ A$ and $B$ some parameters, $\cdot \uplus \cdot $ the disjoint union and
$ [ \cdot \Rightarrow^s \cdot ] $ the exponential of $\cqcb$.

The category $\cqcb$ is known to have countable inductive limits \cite{ELS_04,BSS_07}. Still, our result is highly non-trivial,
 as it is apparent that  a transfinite  and  possibly uncountable inductive construction is required.

Topological domains are $qcb_0$ spaces with a domain-like structure. Many important results for domains have already been generalised to
topological domains, including solutions of recursive domain equations \cite{Bat_06,Bat_08}.
Our aim here, however, is to show that in the simple case of a strictly positive induction, such solutions exist for all $qcb_0$ spaces.
If we restricted ourselves to topological domains, we would also throw away most spaces of interest in computable analysis.
For this purpose, it is essential that we use some kind of representation of the $qcb_0$ spaces.
We will choose to work with domain representations.
Our result could be regarded as a further justification for the utility of domain representations.

In brief, we proceed as follows: 
We first define a category $\cclcdomwp$ of certain well-structured partial equivalence relations on domains. 
This category will be designed to fulfill the following requirements:
\begin{enumerate}[(1)]
\item It contains representations of all $qcb_0$ spaces.
\item Strictly positive operations are functorial.
\item It admits transfinite inductive limits.
\end{enumerate}
It is then possible to construct least fixed points of all strictly positive functors.
We then show that this least fixed point construction can be performed with dense partial equivalence relations on domains
and with a dense least fixed point as the outcome.
We also prove that  this dense least fixed point induces an admissible domain representation if all the parameters involved
are admissible, and this is the main technical difficulty of the paper.

On the other hand, if we have a strictly positive operation $\Gamma$ on $qcb_0$ spaces, we can represent it by a strictly
 positive endofunctor $\funcf$ over $\cclcdomwp$ in a standard way. The dense  least fixed point of $\funcf$ gives us a 
fixed point of $\Gamma$ which  is independent of the actual representing functor.
This is a   $qcb_{0} $ space defined by strictly positive induction.

In section~\ref{s_backg}, we give a short introduction to  domain theory, strictly positive induction,  $qcb_{0} $ spaces and 
admissible domain  representations.
In section~\ref{s_dwp}, we  study domains with partial equivalence relations, and in particular the category $\cclcdomwp$ and
its least fixed point construction. 
In section~\ref{s_qcb} we apply the results from  section~\ref{s_dwp} to prove our main result, theorem~\ref{t_fixedpoint}, 
that a strictly positive operation on  $qcb_{0} $ spaces has a canonical fixed point.

Some of the results have very long and technical proofs. 
For the sake of readability, the proof of all claims made in these proofs are moved to appendix~\ref{app.proofs}. 
An overview of the notation used in different proofs can be found in appendix~\ref{app.notation}.

\section{Background}  \label{s_backg}

\noindent We review some of the basic theory. Our only intention is to present the notation we will use.
For an introduction to domain theory, see  \cite{AJ_94,GHKLMS,SLG_94}, and in particular \cite{SLG_94}
for background on inductive definitions and recursive domain equations.
For more on $qcb_0$ spaces, see \cite{BSS_07,ELS_04,Sch_03}. For details on the theory of domain representations,
we refer to \cite{Bla_00,ST_07}. The listed results concerning admissible domain representations are from \cite{Ham_05}.

\subsection{Domain theory}

A {\em cpo} is a partial order with a least element, $ D = ( D; \sqsubseteq , \bot ) $, 
for which every directed subset $ \Delta \subseteq D $ has a least upper bound $ \bigsqcup \Delta \in D$.
A  $ p \in D $ is {\em compact} if whenever $ \Delta \subseteq D $ is directed and $ p \sqsubseteq \bigsqcup \Delta $, there exists $ d \in \Delta $ with $p \sqsubseteq d $. 
We  denote by $ D_{c} $ the set of compact elements. 
We let $ \textrm{approx} (x) := \{ p \in D_{c} : p \sqsubseteq x \} $,  the set of compact approximations of $ x \in D$.
A cpo is {\em algebraic} if, for every $ x \in  D$, the set  $ \textrm{approx} (x) $ is directed with $ x = \bigsqcup  \textrm{approx} (x) $.

A subset of a partial order is {\em consistent} if it has an upper bound. 
A cpo $ D $ is {\em consistently complete} if every consistent $ A \subseteq D $ has a least upper bound $ \bigsqcup A $.
For a consistent pair of elements $ x,y  \in D $, we  usually denote  the least upper bound by $ x \sqcup y $.
We let $ \upset{x} : = \{ y \in D : x \sqsubseteq y \} $.

A {\em domain} is a consistently complete, algebraic cpo.
We will consider a domain $D$ as a topological space with the {\em Scott topology}.
A base for this topology is given by $ \{ \upset{p} : p \in D_{c} \}  $ and a domain is {\em separable} or
{\em countably based} if $ D_{c} $ is countable.

A function $ f : D \rightarrow E $ is {\em continuous} if firstly it is monotone, 
i.e. $ x \sqsubseteq y \Rightarrow  f(x) \sqsubseteq f(y) $ for all $ x,y \in D$,
and secondly $ f  ( \bigsqcup \Delta ) = \bigsqcup f [ \Delta] $ whenever $ \Delta \subseteq D $ is directed.
Every monotone function $ f: D_{c} \rightarrow E $   has a unique extension to a continuous function $ f : D \rightarrow E $ with $ f(x) = \bigsqcup  f [  \textrm{approx} (x) ]  $ for every $ x \in D $.
In fact, every  continuous $ f : D \rightarrow E $ can be recovered from its restriction to $ D_{c}$ in this way.
We let $ \cdomain $ be the category of domains with continuous functions as morphisms, and 
we denote by $ \ccbdomain $ its full subcategory of countably based domains.

An {\em embedding-projection pair} $  (f, g ) : D \rightarrow E $ is a pair of continuous functions 
$ f : D \rightarrow E $, the embedding,  and $ g : E \rightarrow D  $, the projection, such that 
$  g (f (x)) = x  $ for every $  x \in D $ and $ f ( g  (y)) \sqsubseteq y $ for every $ y \in   E $.
Each one of these functions  is uniquely determined by the other.
Usually we will refer simply to the embedding $ f : D \rightarrow E $ and denote the associated projection by $ f^{-}$.
We let $ \cdom $ be the category of domains with embeddings as morphisms and $ \ccbdom $ its full subcategory of countably based domains.

More generally, an {\em adjunction pair} $ (f,g) : D \rightarrow E $ is a pair of continuous functions $ f: D \rightarrow E  $, the lower adjoint, 
and $ g: E \rightarrow D $, the upper adjoint, such that
$   x  \sqsubseteq  g (f (x)) $ for every $  x \in D $ and $ f ( g  (y)) \sqsubseteq y $ for every $ y \in   E $.

If $ D$ and $ E$ are cpos, we let {\em the disjoint sum} $ D + E $ be
\[ \{ (0, x) : x \in D \} \cup \{ ( 1, y) : y \in E \} \cup \{ \bot \} ,\]
with the separated partial order, i.e. $ (i,x) \sqsubseteq (j,y) $ if
and only if $ i = j $ and $ x \sqsubseteq y $, and with $ \bot \notin
D \cup E $ as a least element.  The disjoint sum of two domains is
itself a domain.  If $\{ D_i\}_{i \in I} $ is any finite, non-empty
set of domains, we define the disjoint sum $ \biguplus_{ i \in I} D_i
$ in the same way. Observe that the disjoint sum of just one domain
$D$ is the {\em lifting} $D_\bot$ of $D$, that is $ D$ with a new
least element added.  The {\em strict sum} $ D \oplus E $ of cpos (or
domains) $D$ and $E$ is the disjoint sum $D +E$ with $(0, \bot_D ) $
and $(1, \bot_E) $ removed.

If $D$ and $E$ are domains, we let the {\em Cartesian product} $ D
\times E $ be the domain obtained from the Cartesian product of sets
and the product order.  The {\em strict product} $ D \otimes E $ of
$D$ and $E $ is the Cartesian product $ D \times E $ with all pairs
$(x,y) $ such that either $x= \bot_D$ or $y= \bot_E $ removed.

If $ p \in D_{c} , q \in E_{c} $,  the {\em step function} $ [p ; q] :D \rightarrow E $  is the continuous function defined by
\[  [p;q] (x) =
\left\{
\begin{array}{ll}
q \mbox{ if } p \sqsubseteq x  \\
\bot \mbox{ otherwise } \\
\end{array}
\right. \]
The {\em function space} $ [D \rightarrow E] $, the set of continuous functions with the point-wise order,
is a domain with least upper bounds of finite consistent sets of step functions as compact elements. So, an
element of $[D\to E]_c$ is written as $ \bigsqcup_{ j \in J } [ p^{j} ; q^{j} ]$, with $J$ a finite set and 
$p^j\in D_c$ and $q^j\in E_c$ for every $j\in J$.
We will occasionally refer to $[D\to E]$ as the {\em exponentiation} of $E$ by $D$.

A continuous function $ f :D \to E $ is {\em strict} if $ f( \bot_D) = \bot_E$. The {\em strict function space} $ [D \to_\bot E] $ is the domain of strict continuous functions with point-wise order.

Each of the above-mentioned operations on domains give countably based domains from countably based domains $D$ and $E$.

If $ f : D \rightarrow D' $ and $ g: E \rightarrow E' $ are domain functions,
there are natural definitions of functions
$ (f + g) : ( D + E ) \rightarrow (D'+ E') $ and
$ (f \times  g) : ( D \times E ) \rightarrow (D' \times  E') $.
If $ f$ and $g $ are continuous (resp. embeddings), then $ f + g $ and $ f \times g $ are continuous (resp. embeddings) as well.
If $ f : D' \rightarrow D $ (note that $ D $ and $D'$ have changed positions) 
and $ g: E \rightarrow E' $ are continuous functions, then the function
$ ( f \rightarrow g) : [ D \rightarrow E ] \rightarrow [ D' \rightarrow E' ] $ defined by
$ ( f \rightarrow g)  (x) =  g \circ x \circ f $ is continuous.
If $ f : D \rightarrow D' $ and $ g: E \rightarrow E' $ are embeddings, then
$  (f^{-} \rightarrow g ) : [ D \rightarrow E ] \rightarrow [ D' \rightarrow E' ] $  is an embedding with $ f \rightarrow g^{-}   $ as associated projection.

\subsection{Strictly positive induction}

An operation $ \Gamma $ on domains is {\em strictly positive} if it is constructed from a finite list of fixed domains using 
the {\em basic} operations identity, disjoint sum, Cartesian product and exponentiation by a fixed domain.
We will refer to fixed domains occurring as exponents, i.e. on the left hand side of a function space, 
as the {\em non-positive parameters} of $ \Gamma $ and the remaining parameters in $ \Gamma $ as the {\em positive} ones.
In our fundamental example  $ \Gamma (X) = A + [ B \rightarrow X ] $, $A$ is the positive parameter and $B$ is the non-positive
parameter.

We have seen that the operations $ \cdot + \cdot $, $ \cdot \times  \cdot $ and $[ \cdot \rightarrow \cdot] $ have strict counterparts
$ \cdot \oplus \cdot $, $ \cdot \otimes \cdot  $ and $ [\cdot  \rightarrow_{\bot} \cdot ] $. 
It may seem natural to include these as well as the lifting operation $ \cdot_{ \bot  } $
as basic operations above.
However, our main concern here is the theory of domain representations and not domain theory itself, 
and a domain representation  $  (D, D^{R} , \delta ) $ can always be chosen such that $ \bot_{D} \notin D^{R} $.
Therefore, the lifting operation and the strict sum and product can safely be omitted from our discussion, since they differ 
from  the identity operation and the respective non-strict operations on $ D \setminus D^{R} $ only.
When we go from the function space to the strict function space, we throw away many total elements, since total continuous functions by no means have to be strict.
However, under the assumption that least elements are not total, the represented space remains unchanged. This explains why even the strict function space is irrelevant for us here and therefore ignored.

If $ \ck $ is a category, an operation $ \Gamma : \textrm{Obj} ( \ck ) \rightarrow \textrm{Obj} (\ck ) $ is 
{\em functorial} in $ \ck $ if there exists a functor  $ \funcf : \ck \rightarrow \ck $ extending $ \Gamma $, 
that is $ \funcf (X ) = \Gamma (X  )  $ for every $ X    \in  \textrm{Obj} (\ck ) $.
In particular, it is easily verified that strictly positive operations on domains are functorial in  $ \cdom $.
A functor $ \funcf : \cdom \rightarrow \cdom $ is  said to be {\em strictly positive } if it is the functorial 
extension of a strictly positive operation on domains.
This  generalises to multivariate operations and multifunctors.

\begin{defi}
Let $ \ck $ be a category and let $ \funcf : \ck \rightarrow \ck $ be a functor.
A {\em fixed point} of $ \funcf $ is an $ X \in  \textrm{Obj} (\ck ) $ which is isomorphic to $ \funcf (X) $ in $ \ck$.
An {\em $ \funcf$-algebra} is a pair $ ( X , f ) $, where $ X$ is an object of $ \ck $ and
$ f : \funcf  (X) \rightarrow X $  is a morphism.
If $ (X, f) $ and $ (Y,g) $ are $ \funcf$-algebras, an {\em $ \funcf$-morphism} from 
$ (X, f) $ to $ (Y,g) $ is a morphism $ h: X \rightarrow Y $ such that
$ h \circ f = g \circ \funcf (h) $.
An  $ \funcf$-algebra $ (X, f) $ is {\em initial} if for every other $ \funcf$-algebra $ (Y,g) $,
there exists a unique $ \funcf$-morphism $h$ 
from $ (X, f) $ to $ (Y,g) $.

Finally, $ X$ is a {\em least fixed point} of $ \funcf $ if there exists some $f$ such that  $ (X, f) $ is an  initial 
$ \funcf$-algebra.
\end{defi}

Note that initial $\funcf$-algebras correspond to initial objects in the category of $ \funcf$-algebras and $\funcf$-morphisms (for a fixed $ \funcf$), thus initial $\funcf$-algebras are unique up to isomorphism.
It can also be proved that if $ (X, f) $ is an initial $ \funcf$-algebra, 
then $ (\funcf  (X) , \funcf (f) ) $ is an initial $ \funcf$-algebra:
If $ (Y,g) $ is another $ \funcf$-algebra and $ h: X \rightarrow Y $ is the unique $ \funcf$-morphism 
from $ (X, f) $ to $ (Y,g) $,
then $ h \circ f $ is the unique $ \funcf$-morphism 
from $ ( \funcf (X) , \funcf ( f) ) $ to $ (Y,g) $.
As a consequence, $f$ is an isomorphism in $ \ck $ and the least fixed point  is indeed a fixed point.
Moreover, a least fixed point is, when it exists, unique up to isomorphism.

If $ \funcf $ is an endofunctor over $ \cdom $, a least fixed point of $ \funcf $ is a fixed point $D$ of $ \funcf $ with  a natural and unique embedding $ h: D \rightarrow E $ into every other fixed point $ E $ of $ \funcf $.
The categorical presentation using $ \funcf $-algebras makes it possible to generalise this concept to other categories, and at the same time it guarantees that a least fixed point is unique up to isomorphism.

\begin{defi}
Let $ \ck $ be a category and let $ (I, \leq ) $ be a directed partial order. 
A {\em directed system} over $I$  in $ \ck $ consists of a family $ \{ X_{i}  \}_{ i \in I } $  of objects from $ \ck $ and
a family of morphisms $ f_{i,j} : X_{i} \to X_{j} $ for all $i \leq j \in I $ satisfying
\begin{enumerate}[$\bullet$]
\item{} $ f_{i,i} = \id_{ X_{i}} $ for every $ i \in I $; and
\item{} $ f_{i,k} = f_{j,k} \circ f_{i,j} $ for all $ i,j,k \in I $ with $ i \leq j \leq k $.
\end{enumerate}
An {\em inductive limit} over this directed system consists of an $ X \in \textrm{Obj} ( \ck ) $ and morphisms
$ f_{i} : X_{i} \rightarrow X  $ for all $ i \in I  $ such that
$ f_{i}  = f_{j} \circ f_{i,j } $ whenever $i \leq j $.
It is universal in the sense that for every other such pair $ ( Y,  \{ g_{i}  \}_{ i \in I } ) $,
there exists a unique {\em mediating morphism} $ g_{I} : X \rightarrow Y $ 
such that $ g_{I} \circ f_{i} = g_{i} $ for every $ i \in I $.
In categorical terms, $ ( X,  \{ f_{i}  \}_{ i \in I } ) $  is a co-limiting cocone in $ \ck $.
\end{defi}

If  $ ( \{ D_{i}  \}_{ i \in I }, \{ f_{i,j} \}_{ i \leq j \in I })  $ is a directed system in $ \cdom $,
there exists an inductive limit $ ( D,  \{ f_{i}  \}_{ i \in I } ) $, defined as follows:
Let $ D = (D, \sqsubseteq ) $ be the domain with
\[ D = \{ x \in \prod_{i \in I} D_{i} : \forall i,j \in I  ( i\leq j \to f^{-}_{i,j} ( x_{j} ) = x_{i}) \} \]
and $ x \sqsubseteq y \Iff  \forall i \in I ( x_{i} \sqsubseteq_{D_{i}}  y_{i} ) $
Let $ f_{i} : D_{i} \to D $ be the embedding  such that $ f^{-}_{i} (x) = x_{i} $ for all $ x \in D $.
It is worth noting that $ f_{i} (x)_{j}  = f_{i,j} (x) $ whenever $ i \leq j $ and $ x \in D_{i} $.

A directed system over a limit ordinal $ \gamma $ is also called an $ \gamma $-chain and an endofunctor is {\em $ \gamma $-continuous} if it preserves inductive limits of $ \gamma $-chains.  
A classical result from domain theory says that every $\omega $-continuous functor $ \funcf : \cdom \rightarrow \cdom $  has a least fixed point, see
\cite{SLG_94}.
For a sketch of the proof,  
consider the $ \omega $-chain $ ( \{ D_{n} \}_{ n \in \omega } , \{ f_{m,n} \}_{ m \leq n \in \omega } ) $  defined inductively as follows:
\begin{enumerate}[$\bullet$]
\item
Let $ D_{0} = \{ \bot  \} $ and $ D_{n+1} = \funcf ( D_{n} ) $.
\item
Let $ f_{0,n} $ be the unique embedding from $ D_{0} $ into $ D_{n} $ and let
\[ f_{ m+1 , n+1} := \funcf ( f_{ m , n } ) : D_{m+1} \to D_{n+1} .\]
\end{enumerate}
Let $ ( D_\omega  , \{ f_{ n} \}_{ n \in \omega  }  ) $ be the inductive limit of this chain.
Then there is an isomorphism $ f : D_\omega \rightarrow \funcf (D_\omega) $, and $ D_\omega$ is a least fixed point of $ \funcf $.
Since this least fixed point is obtained as a countable limit, the result holds even for $\omega $-continuous endofunctors 
over $  \ccbdom  $.

All strictly positive functors over $\cdom$ are $ \omega $-continuous, see \cite{SLG_94}. This means that if $ \Gamma $ is strictly 
positive, we obtain a least solution to the recursive domain equation $ X = \Gamma (X ) $ by a least fixed point construction.
This is of course not set-theoretical equality, but equality of domains up to isomorphism.

\subsection{The category of $qcb_0$ spaces}

If $ X $ and $ Y $ are topological spaces, we let $ X \uplus Y $ be the disjoint union of $ X $ and $ Y$,
i.e. the set $ \{ (0,x ) : x \in X \} \cup \{ (1,y) : y \in Y \} $ provided with the finest topology which 
makes both inclusion maps continuous.

If $X$ is a topological space, then $ U \subseteq X $ is {\em sequentially open} in $X$ 
if for every sequence $ \{ x_{n} \}_{n} $ converging to $ x \in U $, there exists $ n_{0} $ such that
$ \{ x_{n} : n \geq n_{0} \} \subseteq U $.
Every open set is sequentially open, and we say that $ X $ is {\em sequential} if, conversely, every sequentially open set is open.
The family of sequentially open sets defines a sequential topology refining the original topology on  $X$. 
We denote this new topological space by  $  \mathcal{S} X $, the {\em sequentialisation} of $X$, see \cite{Fra_65}.

Let $ X $ and $ Y $ be topological spaces.
A function $ f: X \rightarrow Y $ is {\em sequentially continuous} if it maps convergent sequences in $ X$ to convergent sequences in $Y$.
In particular, every continuous function is sequentially continuous, and if $ X$ is sequential, the two notions coincide.
Let   $ [ X \rightarrow_{  \omega } Y ] $ be the topological space  with the set of sequentially  continuous functions $f: X\to Y$ as underlying set and 
topology generated from sub-basic open sets of the form $  O( n_0 ; U) $. Here,
\[ O( n_0 ; U) := \{ f  :   f  [  \{  x_{n } : n_{0}  \leq n \leq \infty \}  ] \subseteq U  \}, \]
with $n_0 $ some natural number,  $ x_{ \infty } $ the limit of a convergent sequence $ \{ x_{n}\}_{n \in \nat } $ in $ X $ and 
$ U $  an open subset of $ Y$.

A (sequential) {\em pseudobase} for a topological space $ X $ is  a set 
$ \mathsf{P}  $ of non-empty subsets of $ X$,
containing $ X$, closed under non-empty finite intersections and such that
if $ \lim x_{n} = x_{ \infty } \in U $ and $U$ is open in $ X$, 
there exists $ B \in \mathsf{P} $  and $ n_{0} \in \nat $
such that $ \{  x_{n } : n_{0}  \leq n \leq \infty \} \subseteq B \subseteq U $.
The closure under finite intersections of an arbitrary superset of $ \mathsf{P}  $   is a pseudobase for $X$ as well.

A topological space is said to be a {\em $qcb$ space} if it is the topological quotient of some countably based space.
It is {\em $qcb_{0} $} if, in addition, it is $T_{0} $, see \cite{BSS_07}. 
It is well-known that a $T_{0} $  space is $ qcb $ if and only if it is sequential and has a countable pseudobase, see \cite{Sch_03}. 

Let $ \cqcb $ be the category with $qcb_{0}$ spaces as objects and continuous functions as morphisms.
The category $ \cqcb $ admits countable products and coproducts and is Cartesian closed, see \cite{BSS_07,ELS_04}.
The finite product in $ \cqcb $ is the sequentialisation of the usual product, denoted by $ \cdot \times^{s} \cdot $ in the binary case.
The exponentiation in $ \cqcb $ is the sequentialisation of $ [  \cdot \rightarrow_{  \omega }  \cdot ] $, 
denoted by $ [ \cdot   \Rightarrow^{s} \cdot ]  $.
This topology can similarly be obtained as the sequentialisation of the compact-open topology on the set of continuous functions.
The disjoint sum $ X \uplus Y $ of $ qcb_{0} $ spaces $X$ and $Y$ is trivially $qcb_{0} $.


\subsection{Admissible domain representations}
A {\em domain with totality} is a pair $ (D, D^{R} ) $, where $ D $ is a domain and 
$ D^{R} $ is a subspace of $ D$ (with the Scott topology). For most purposes, we may 
assume $ \bot_{D} \notin D^{R} $. A  domain with totality $ (D, D^{R} ) $ is {\em dense} 
if $ D^{R} $ is dense as a subspace of $ D$.

A {\em domain representation} of a topological space $X$ consists of a domain with totality 
$ (D, D^{R} ) $ and  a representation map  $ \delta : D^{R} \to X $, which is a surjective, 
continuous function. The representation is {\em countably based} if  $ D$ is separable,
and {\em dense} if  $ (D, D^{R} ) $ is dense.

If $ (D, D^{R} , \delta ) $ and $ (E, E^{R}, \varepsilon ) $ are domain representations of 
$ X $ and $Y$, respectively, a map $ g : X \rightarrow Y $ is {\em $(\delta , \varepsilon 
)$-representable}  if there exists some continuous function $ f:D\rightarrow E $ with $  
f [ D^{R} ] \subseteq E^{R} $ and  $  g \circ \delta   = \varepsilon \circ f\vert_{ D^{R} }  $.
Such an  $ f : D \rightarrow E $ is  {\em $(\delta , \varepsilon )$-total}, which means that
$ f [ D^{R} ] \subseteq E^{R} $ and 
$ \delta (x) = \delta (y) \Rightarrow \varepsilon( f( x)) = \varepsilon ( f(y)) $ for all $ x,y \in D^{R} $.
A $(\delta , \varepsilon )$-total function  $ f : D \rightarrow E $ represents a unique map $ g: X \rightarrow Y $.
If $ \delta $ is a quotient map, all $(\delta , \varepsilon )$-representable maps are continuous.

\begin{defi}
A countably based domain representation $ (D, D^{R}, \delta ) $ of a topological space  $ X$ is {\em admissible} 
if every continuous map $ \varphi : E^{R} \rightarrow  X $,
with $ (E , E^{R} ) $ a countably based, dense domain with totality, factors through $ \delta $,
i.e. there exists a continuous map
$ \hat { \varphi } : E \rightarrow D $ such that
$ \hat { \varphi } [  E^{R} ] \subseteq D^{R} $ and 
$ \delta \circ  \hat { \varphi } ( e) = \varphi ( e) $ for every $ e \in E^{R} $.
\end{defi}

\begin{rem}
This is actually the definition of  $\omega$-admissibility, but for countably based domain representations the notions
of admissibility and  $\omega$-admissibility coincide,  see  \cite{Ham_05}. 
The general definition of admissibility, which is of no interest in this paper, is more restrictive.
\end{rem}

\begin{thm}  \label{t_adm1}
A topological space $X$ has a countably based, admissible domain representation if and only if it is $T_{0} $ and 
has a countable pseudobase.
\end{thm}

\begin{proof}
We give a sketch of the proof. For details, see \cite{Ham_05}. 

If $ \mathsf{P} $ is a pseudobase for $X$, then  $ ( \mathsf{P}, \supseteq ) $ is a cusl, so let 
$ D = \textrm{Idl}  ( \mathsf{P}, \supseteq )  $, the domain obtained by ideal completion.
Define a relation $ \rightarrow_{  \mathsf{P} }  $ as follows:
If $ I \in D $ and $ x \in X $, let $ I \rightarrow_{  \mathsf{P} }  x $  
if firstly $ x \in B $ for every $ B \in I $ and secondly there exists $ B \in I $ with $ x \in B 
\subseteq U $ for every open $ U \subseteq X $ with $ x \in U $.
Let $ D^{R}  $ be the set of ideals $ I$ such that  $ I \rightarrow_{  \mathsf{P} }  x $ for some $ x \in X $,
and on condition that $X$ is $T_0$, define $ \delta : D^{R} \rightarrow X $ by $ \delta (I) = x $.
Then it can be verified that $ (D, D^{R}, \delta ) $  is an admissible domain representation of $ X $.

Conversely, if $ (D, D^{R}, \delta ) $  is an admissible domain representation of $ X $, it can be verified
that $ \{ \delta [ \upset{ p } \cap D^{R} ] : p \in D_{c} \} $ is a pseudobase for $ X $ and that $X$ is $T_0$.
\end{proof}

\begin{rem}
The admissible domain representation of $ X $ constructed from  $ \mathsf{P} $ in the proof above, is known as the {\em standard} representation of $X$ w.r.t.  $ \mathsf{P} $.
An important aspect of this representation is the existence of a greatest representative
$ I^{x} = \{ B \in  \mathsf{P} : x \in B \}  $ for every  $x \in  X$.
Since every $ B \in  \mathsf{P} $ is non-empty, it is clear that the representation is dense. 
\end{rem}

\begin{rem}
An alternative approach to dense, admissible domain representations is via continuous reductions, as defined in \cite{Bla_08}.
If $ (D, D^{R} , \delta ) $  and  $ (E, E^{R}, \varepsilon ) $ are domain representations of $X$, a {\em continuous reduction} is a 
$ ( \delta , \varepsilon ) $-total map representing $ \id_{X} $.
A dense representation of $ X$ is then admissible if and only if it is universal among all dense representations of $X$ w.r.t. continuous reductions.
\end{rem}

An important motivation for  admissible domain representations is the lifting of continuous functions.

\begin{lem} \label{l_adm3}
Let $ (D, D^{R}, \delta ) $  and  $ (E, E^{R}, \varepsilon ) $ be admissible domain representations of $X$ and $ Y $, respectively.
If $ g: X \rightarrow Y $ is $(\delta , \varepsilon )$-representable, then it is sequentially continuous.

Conversely, if $  D^{R} $ is dense in $D$, then every sequentially continuous function $ g: X \rightarrow Y $ is $(\delta , \varepsilon )$-representable.
\end{lem}

\begin{lem}  \label{l_adm5}
Let $ (D, D^{R}, \delta ) $ and $ (E, E^{R}, \varepsilon ) $ be admissible representations of $ X $ and $Y $, respectively.
Then there exist
\begin{enumerate}[$\bullet$]
\item
 a representation map $ \varrho :  D^{R} \uplus  E^{R} \rightarrow  X  \uplus Y  $ such that
$ (D + E,  D^{R} \uplus  E^{R}, \varrho ) $ is admissible;
\item
 a representation map $ \varrho :  D^{R} \times  E^{R}  \rightarrow  X  \times Y  $ such that
$ (D \times E,  D^{R} \times  E^{R}, \varrho) $ is admissible; and
\item
 a representation map $ \varrho :  [D \to  E]^{R} \rightarrow [ X  \to_{ \omega }  Y ] $,
such that $ ( [ D \to E],  [D \to  E]^{R}, \varrho) $ is admissible,
on condition that $ D^{R} $ is dense in $D$. 
Here, $  [D \to  E]^{R} $ is set of  $(\delta , \varepsilon )$-total continuous maps.
\end{enumerate} \end{lem}

\noindent The representation maps $ \varrho $ in the lemma are the expected ones. In particular, in the latter case, $ \varrho (f) : X \rightarrow Y $ is the sequentially continuous function represented by $ f$.

For sequential spaces,  the notions of continuity and sequential continuity coincide.
 This is the situation when we consider  quotient domain representations.

\begin{lem} \label{l_adm6}
Let  $ (D, D^{R}, \delta ) $  be an admissible domain representation of $ X $.
Then $ \delta $ is a quotient map if and only if $X$ is a sequential space. 
\end{lem}

\begin{cor} \label{c_adm7}
A topological space $ X $ has a countably based, admissible quotient domain representation if and only if it is a $ qcb_{0} $ space.
\end{cor}

\section{Domains with partial equivalence relations}
\label{s_dwp}

\noindent We review the theory of partial equivalence relations on domains, as presented in \cite{BBS_04}.
We introduce a new category of domains with  partial equivalence relations and 
show by a transfinite  induction that a strictly positive functor over this category has a least fixed point.

Furthermore, we make the connection between admissible quotient domain representations  and domains with partial equivalence 
relations, and show how the least fixed point obtained  can be replaced by a dense one. Finally, we use the intuition 
acquired from the fundamental example to show that this dense least fixed point induces an admissible domain representation.

Least fixed points in similar categories have been studied previously \cite{Hyl_88,RS_99}. For our purpose, however, 
the inductive construction of the least fixed point, as we know it from domain theory, is crucial when we later will relate 
our result to $qcb_0$ spaces through the notion of admissibility.

\subsection{Introduction}

A {\em partial equivalence relation (per)}  on a set $X$ is a binary relation  which is  symmetric and transitive.
A per $\approx $  induces an equivalence relation on its domain,
i.e. there is a subset $ (X, \approx)^{R}  := \{ x \in X : x \approx x  \} $ of $ X$ such that $ \approx $ restricted to $ (X, \approx)^{R}  $ is an equivalence relation.

A  {\em domain-per}, short for a domain with a per, is a pair $ \mathcal{D} = ( D,  \approx   )  $, where $ D   $ is a domain and
$ \approx   $ is a per on $D$.
We will denote domain-pers by calligraphic letters $ \mathcal{D} ,\mathcal{E} , \ldots $  and the respective 
underlying domains by  $ D,E, \ldots $, unless stated otherwise.
The per of  $ \mathcal{D} $ is usually denoted  by $ \approx $, but with a subscript $ \mathcal{D} $ 
if the domain-per is not clear from the context.

If $ \mathcal{D} $ is a domain-per, let $ \mathcal{D}^{R} $ be the topological space of $ \{ x 
\in D : x \approx  x \} $,   the set of {\em $\mathcal{D}$-total} elements of $D$, with the 
subspace topology inherited from  $ D $.
We say that $\mathcal{D} $ is {\em trivial} if $\mathcal{D}^{R} = \emptyset $ and {\em non-trivial} otherwise.
$\mathcal{D} $ is {\em dense} if  $ \mathcal{D}^{R} $ is  dense  in $ D $.

We let $ \mathcal{Q} \mathcal{D}  := \mathcal{D}^{R}/_{\approx_{ \mathcal{D} }} $, i.e. the topological space with the 
underlying set $ \{ [x]_{ \mathcal{D} } : x \in \mathcal{D}^{R}  \} $, 
and with quotient topology generated from the equivalence relation  $ \approx $ on $ \mathcal{D} ^{R} $,
where $  [x]_{ \mathcal{D} }  := \{  y \in D : x \approx  y \} $,
the partial equivalence class of $x$.
If $ p \in D_{c} $ and $ x \in  \mathcal{D}^{R} $, we say that 
$ p \prec_{  \mathcal{D} }  [x] $, if  there exists some $ y \in  [x]_{ \mathcal{D}}  $ with $ p \sqsubseteq y $.

If $ \mathcal{D}  $ and $ \mathcal{E}   $ are domain-pers, let $ [ \mathcal{D}  \rightarrow  \mathcal{E} ] $ be the 
domain-per with $  [ D \rightarrow E ] $ as the underlying domain and per  defined by: $ f \approx g $ if and only if
\[  \forall x, y \in D \; ( x \approx_{\mathcal{D}} y \Rightarrow f(x) \approx_{\mathcal{E}} g(y) ).\]
An {\em equivariant} mapping $ f:\mathcal{D}  \to \mathcal{E} $ is a continuous function  $ f :D  \to  E $  such that $ f \approx  f $.
Thus, $ [  \mathcal{D}  \rightarrow  \mathcal{E} ]^{R}  $ is the set of equivariant mappings.
An equivariant $ f :  \mathcal{D} \to  \mathcal{E} $ induces a unique continuous function
$ f^{\mathcal{Q}} : \mathcal{Q}  \mathcal{D}  \to  \mathcal{Q}  \mathcal{E}  $ defined by
$  f^{\mathcal{Q}}  (  [x]_{ \mathcal{D} } ) =  [ f(x) ]_{ \mathcal{E} } $,
with $ f^{\mathcal{Q}} = g^{\mathcal{Q}}  $ if and only if  $ f \approx g $.

The following technical result will prove itself useful. The proof is straight-forward and therefore omitted.
\begin{lem}  \label{l_equiv1}
If $ f : \mathcal{D} \rightarrow  \mathcal{E} $ is equivariant and $ g :D \rightarrow E $ is continuous,
then $  f \approx_{[ \mathcal{D} \to \mathcal{E}] } g $ if and only if
$ f(x) \approx_{ \mathcal{E} } g (x) $ for every $ x \in  \mathcal{D}^{R} $.
\end{lem}

An equivariant map $ f:\mathcal{D}  \to \mathcal{E} $ is {\em equi-injective} if 
 $ f(x) \approx_{ \mathcal{E} } f (y) \Rightarrow x \approx_{  \mathcal{D} } y $ for all $ x, y \in D $.

Let  $ \cdomainwp $ be the category with domain-pers as objects and equivalence classes of equivariant mappings as morphisms, as defined in \cite{BBS_04}. 
It is easily verified that this is a well-defined category; the identity function is equivariant and the composition of two equivariant functions is always equivariant.

Note that  $  \mathcal{D} $ and $  \mathcal{E}$ are  isomorphic in  $ \cdomainwp $ if and only if there exist equivariant maps
$ f :  \mathcal{D} \to  \mathcal{E}$ and $ g :  \mathcal{E} \to  \mathcal{D}$
such that $ g \circ f \approx \id_{D} $ and $  f \circ  g  \approx \id_{E} $.
It is  immediate that $  \mathcal{Q}  \mathcal{D}  \cong  \mathcal{Q}  \mathcal{E}  $ whenever such a pair exists.
In this case, we say that $  \mathcal{D} $ and $  \mathcal{E}$ are {\em weakly isomorphic}
(and that $  ( f , g ) $ is a weak isomorphism pair), 
weakly in the sense that the underlying domains are not, in general, isomorphic.

If $ (D, D^{R} , \delta ) $ is a quotient domain representation of  a topological space $X$, we may  define 
a domain-per $ \mathcal{D}^\delta = (D, \approx_{ \delta } ) $
by letting $  x  \approx_{ \delta } y $ if  $ x,y \in D^{R}  $ and $  \delta  (x) =  \delta ( y ) $.
We then have $ (\mathcal{D}^\delta)^{R} = D^{R} $ and  $ X  \cong   \mathcal{Q}   (\mathcal{D}^\delta)   $.
We will refer to $ \mathcal{D}^{ \delta } $ as the domain-per {\em associated to} $ (D, D^{R} , \delta ) $.
Conversely, an arbitrary domain-per induces a unique quotient domain representation
$ ( D,  \mathcal{D}^{R} , \delta_{  \mathcal{D} } ) $ of $ \mathcal{Q}   \mathcal{D}  $.
Note that for the function space  $ [ \mathcal{D}  \to \mathcal{E} ] $ defined above, $ [ \mathcal{D}  \to \mathcal{E} ]^{R} $ is exactly the set of 
$ ( \delta_{  \mathcal{D}},   \delta_{ \mathcal{E} } )$-total maps.



\begin{defi}
Let $ \mathcal{D} $ and $ \mathcal{E} $ be domain-pers.
\begin{enumerate}[$\bullet$]
\item{}
The {\em disjoint sum} of   $ \mathcal{D} $ and $ \mathcal{E} $,
 $ \mathcal{D} + \mathcal{E} $, is the domain-per with $ D + E $ as underlying domain and per $\approx $ defined by 
$ (i,x) \approx (j,y) $ if and only if either
$i=j= 0 $  and  $ x \approx_{  \mathcal{D} } y $ or $i=j= 1 $  and  $ x \approx_{  \mathcal{E} } y $.
\item{}
The {\em Cartesian product} of $ \mathcal{D} $ and $ \mathcal{E} $,
$ \mathcal{D} \times \mathcal{E} $, is the domain-per with $ D \times E $ as underlying domain and per $\approx $ defined by 
$ ( x,y ) \approx ( x' , y' ) $ if and only if
$  x \approx_{  \mathcal{D} } x' $ and $ y \approx_{  \mathcal{E} } y' $.
\item{}
The {\em exponentiation} of $ \mathcal{E} $ by $ \mathcal{D}$ is the domain-per $ [ \mathcal{D}  \to \mathcal{E} ] $ defined above.
\end{enumerate}
These constructed domain-pers have {\em strict counterparts}  $ \mathcal{D} \oplus \mathcal{E} $,  $ \mathcal{D} \otimes \mathcal{E} $ and  $ [ \mathcal{D} \to_\bot \mathcal{E} ] $. 
We define these using the respective strict counterparts from domain theory as underlying domains and the restrictions of the respective pers as pers.

If $\mathcal{D} $ is a domain-per, the {\em lifting} of $ \mathcal{D} $ is the lifting $D_\bot $ of $D$ with the extension of $\approx_{\mathcal{D}} $
as per.
\end{defi}

We now have natural operations of binary sum, binary product and function space on domain-pers.
We say that an operation on domain-pers is  {\em strictly positive} if the underlying operation on domains is strictly positive.

\begin{rem}
It is easy to see that the binary sums and products defined above  extend to finite sums and products.
In fact, $\cdomainwp $ is a Cartesian closed category, with categorical finite product and exponentiation corresponding to 
the finite product and function space defined above; see \cite{BBS_04} for details.
\end{rem}

\begin{lem} \label{l_equiv2}
Let $ f :  \mathcal{D} \rightarrow  \mathcal{D}' $ and $ g:  \mathcal{E} \rightarrow  \mathcal{E'} $ be equivariant.
Then the following maps are equivariant:
\[ (f  + g ) :  ( \mathcal{D} +   \mathcal{E} ) \rightarrow  ( \mathcal{D}' + \mathcal{E}' ) \]
\[ (f  \times g ) :  ( \mathcal{D} \times   \mathcal{E} ) \rightarrow  ( \mathcal{D}' \times \mathcal{E}' ) \]
\[ (f  \rightarrow g ) :  [ \mathcal{D}' \rightarrow   \mathcal{E} ] \rightarrow  [ \mathcal{D}  \rightarrow \mathcal{E}' ] \]
\end{lem}

\proof
The cases $ (f  + g )  $ and $ (f  \times g ) $ are straight-forward and left for the reader.

If $ x \approx_{[\mathcal{D}'\rightarrow\mathcal{E}]} y $, then $ g \circ x \circ f  \approx_{
[\mathcal{D}\rightarrow\mathcal{E}']} g\circ y\circ f$, and  this shows that 
that $  (f  \rightarrow g )  $ is equivariant.
\qed

\subsection{A category with inductive limits}
Embedding-projection pairs play a crucial role in the least fixed point
construction used to solve recursive domain equations. More precisely,
they are necessary for the construction of inductive limits of directed
systems.

We now introduce a category of domain-pers which has inductive limits
and for which strictly positive operations are functorial.

\begin{defi} \label{d_equiembedding}
An {\em equiembedding} is a map $ f:  \mathcal{D}   \rightarrow \mathcal{E} $ 
such that 
\begin{enumerate}[$\bullet$]
\item $ f: D \rightarrow E  $ is an embedding;
\item $ f:  \mathcal{D}   \rightarrow \mathcal{E} $ is equivariant; and
\item $  \forall x \in\mathcal{D}^{R} \: 
( f(x) \approx_{\mathcal{E}} y\Rightarrow x\approx_{\mathcal{D}}f^- (y)) $.
\end{enumerate}
\end{defi}

First, we show that this is a  valid choice of morphisms. The identity map on
the underlying domain of a domain-per is clearly  an equiembedding, so it remains to 
prove that equiembeddings are closed under composition.
\begin{lem} \label{l_equim1}
Let $ f: \mathcal{D}   \rightarrow \mathcal{E} $ and $ g : \mathcal{E}   \rightarrow \mathcal{F} $ be equiembeddings.
Then $ g \circ f $ is an equiembedding.
\end{lem}

\proof
Embedding-projection pairs and equivariant maps are both closed under composition, so it remains to
verify the third requirement.
Assume $x\in \mathcal{D}^R$ and $g(f(x))\approx_{\mathcal{F}}y $.
Then  $ f(x )\in\mathcal{E}^R$, so $ f(x)\approx_\mathcal{E}g^- (y) $, because $ g $ is an equiembedding.
Moreover, \[ x\approx_\mathcal{D}f^-( g^-(y)) = ( g \circ f)^-(y) \] because $f$ is an equiembedding.\qed

\begin{prop} \label{p_equim2}
Strictly positive operations  are functorial w.r.t.\ equiembeddings.
\end{prop}

\proof
The proof is by structural induction on strictly positive operations:

\begin{enumerate}[$\bullet$]
\item{Equiembeddings are closed under disjoint sums:\ }
Let $ f:\mathcal{D}\rightarrow\mathcal{D}' $ and $g:\mathcal{E}\rightarrow \mathcal{E}' $ 
be equiembeddings.
Then $f+g$ is an embedding of domains and equivariant by lemma~\ref{l_equiv2}. We need to verify the third
requirement of definition~\ref{d_equiembedding}.

Assume  that $  (i,x) \in   ( \mathcal{D}  +   \mathcal{E} )^{ R} $ and 
$ (f + g)  (i,x) \approx_{ \mathcal{D}'   +   \mathcal{E}'} (i,y)$.
If $i=0$, then $  f  (x ) \approx_{ \mathcal{D}' }  y $ and, since $f$ is an equiembedding, $x\approx_\mathcal{D} f^-(y)$.
If $i=1$, then $  g(x ) \approx_{ \mathcal{E}' }  y $ and, since $g$ is an equiembedding, $x\approx_\mathcal{E} g^-(y)$.
Either way, we obtain $(i,x)\approx_{\mathcal{D}  +   \mathcal{E}}(f+g)^-(i,y)$.

\item{Equiembeddings are closed under Cartesian products:\ }
Let $ f:\mathcal{D}\rightarrow\mathcal{D}' $ and $g:\mathcal{E}\rightarrow \mathcal{E}' $ 
be equiembeddings.
Then $f\times g$ is an embedding of domains and equivariant by lemma~\ref{l_equiv2}. Again, we need to verify the last
requirement of definition~\ref{d_equiembedding}.

Assume that  $ x\in (\mathcal{D}\times\mathcal{E})^R $ and $ (f\times g)(x)\approx_{ \mathcal{D}'  \times   \mathcal{E}' } y$. 
If we let $x=(x_1,x_2)$ and $y=(y_1,y_2)$, then $f(x_1)\approx_\mathcal{D'}y_1$ and $g(x_2)\approx_\mathcal{E'}y_2$.
Since $ f $ and $g$ both are equiembeddings,  we can conclude that $x_1\approx_\mathcal{D}f^-(y_1)$ and  $x_2\approx_\mathcal{E}g^-(y_2)$.
This means that $x\approx_{\mathcal{D}\times\mathcal{E}}f^-(y)$.

\item{Equiembeddings are closed under exponentiations by a fixed
  domain-per:\ }
Let $ f:\mathcal{D}\rightarrow\mathcal{D}' $ be an equiembedding and let $\mathcal{B} $ be a
fixed domain-per. Then the embedding $\id_{B} \rightarrow f$ is equivariant by lemma~\ref{l_equiv2}, and we need only 
verify the third condition in definition~\ref{d_equiembedding}.

Assume $ x \in  [  \mathcal{B} \to \mathcal{D}  ]^{R} $. We need to show that if $ ( \id_{B} \rightarrow f) (x)  
\approx_{[\mathcal{B}\to\mathcal{D}']} y $, then $ x  \approx_{[\mathcal{B}\to\mathcal{D}]}( \id_{B} \rightarrow f)^{-} (y) $.
Put differently, we have $f\circ x \approx_{[ \mathcal{B} \to  \mathcal{D}'] } y $, and need to show that this implies
 $ x  \approx_{[ \mathcal{B} \to \mathcal{D}]  }     f^{-}  \circ y $.
For an arbitrary $b \in \mathcal{B}^{R} $, we have $ f ( x (b) ) \approx_{ \mathcal{D}'  } y(b) $.
Since $ f $ is an equiembedding, this implies $ x (b)  \approx _{ \mathcal{D}  }  f^{-} (y (b) ) $, and we are through.
\qed
\end{enumerate}
For the inductive limits to be well-defined, we need to restrict ourselves to certain well-structured
domain-pers.

\begin{defi}
A domain-per $ \mathcal{D} $ is {\em weakly convex} if  
$ x \approx y \Rightarrow x \approx  x \sqcup p $ 
whenever  $ x \sqsubseteq y $ and $ p \in \textrm{approx} (y) $,
and {\em convex} if
$ x \approx y \Rightarrow x \approx x \sqcup z $ 
whenever $ x, z \sqsubseteq y $.

A domain-per $ \mathcal{D} $ is {\em local} if $ [ x ]_{ \mathcal{D}}  $ is consistent 
in $D$ for every $ x  \in  \mathcal{D}^{R} $, and {\em strongly local} if $ [ x ]_{ 
\mathcal{D}}  $ is directed for every $ x  \in  \mathcal{D}^{R} $.

A local domain-per $ \mathcal{D} $ is {\em complete} if
$ x \approx \bigsqcup [ x ]_{ \mathcal{D}}  $
for every $ x  \in  \mathcal{D}^{R} $.
\end{defi}

\noindent First, we look at a useful technical lemma: 
\begin{lem}  \label{l_lwc3}
Let $ \mathcal{B} $  be a dense domain-per and let 
$\mathcal{D} $ be a weakly convex and strongly local domain-per.
Let $ \bigsqcup_{ j \in J } [ p^{j} ; q^{j} ]  \in  [ B \rightarrow D ]_{c} $ and  
let $ f \in  [  \mathcal{B} \rightarrow  \mathcal{D} ]^{R} $.

If $ \bigsqcup_{ j \in J } \{ q^{j} : p^{j} \sqsubseteq x \} \prec_{ \mathcal{D}  } [ f(x)]  $ for every $ x \in \mathcal{B}^{R} $, then
$ \bigsqcup_{ j \in J } [ p^{j} ; q^{j} ]  \prec_{ [  \mathcal{B} \rightarrow  \mathcal{D} ] }  [f]  $.
\end{lem}

\proof
Let $ x \in  \mathcal{B}^{R} $ and assume that  $ \bigsqcup_{ j \in J } \{ q^{j} : p^{j} \sqsubseteq x \} \prec_{ \mathcal{D}  } [ f(x)]  $.
By definition of $ \prec_{ \mathcal{D}  }$, there exists  some $ y \in [f(x)] $ 
such that $ \bigsqcup_{ j \in J } \{ q^{j} : p^{j} \sqsubseteq x \}$ is a compact approximation of $y $,
and we can assume  $ f(x) \sqsubseteq y $ since  $ \mathcal{D} $ is strongly local.
Then $ f(x) $ is consistent with $ \{ q^{j} : p^{j} \sqsubseteq x \}_{j\in J}  $, and since  
$  \mathcal{D} $ is weakly convex, we have
\[ f (x) \approx_{ \mathcal{D}}  f(x) \sqcup  \bigsqcup_{ j \in J } \{ q^{j} : p^{j} \sqsubseteq x \} .\]
Moreover, since $  \mathcal{B}^{R} $ is dense in $B$, 
$ f(p) $ and $  \bigsqcup_{ j \in J } \{ q^{j} : p^{j} \sqsubseteq p \}  $ are consistent for all $p \in B_c$,
so $ g :=   f    \sqcup \bigsqcup_{ j \in J } [ p^{j} ; q^{j} ]  $ exists.
By lemma~\ref{l_equiv1}, $ f \approx_{[ \mathcal{B} \rightarrow  \mathcal{D} ]} g $, and this shows that  $ \bigsqcup_{ j \in J } [ p^{j} ; q^{j} ]  \prec_{   
[ \mathcal{B} \rightarrow  \mathcal{D} ] }  [ f]  $.
\qed

Clearly, a local and complete domain-per is also strongly local, so in particular the lemma holds 
for convex, local and complete $\mathcal{D}$. The converse of this lemma holds trivially and 
without any restrictions on   $ \mathcal{B} $ or  $ \mathcal{D}  $.
Also note that for an arbitrary finite subset $ \{ (p^{j} , q^{j} ) \}_{ j \in J }  $ of $ B_{c} \times D_{c} $,
if $ \bigsqcup_{ j \in J } \{ q^{j} : p^{j} \sqsubseteq x \} \prec_{\mathcal{D}} [ f(x)]  $ for all $ x \in \mathcal{B}^{R} $, 
then $ \{ [p^{j} ; q^{j} ] \}_{ j \in J }  $ is $ [ B \rightarrow D ] $-consistent. This holds because  
$ \mathcal{B} $ is dense.

\begin{defi}
Let $ \cclcdomwp $ be the category with convex, local and complete domain-pers  as objects and equiembeddings 
as morphisms.
\end{defi}
Domain-pers $  \mathcal{D} $ and $  \mathcal{E}$ are isomorphic in $ \cclcdomwp$ if there exists 
an  isomorphism pair  $ ( f,  g ) : D \rightarrow E $ such that both $ f $ and $  g $ are equivariant.
As for domains, it is possible to construct a pair of non-isomorphic domain-pers with equiembeddings in both directions.
Note that this is a stronger kind of isomorphism than what we have for $\cdomainwp$. In fact,  the category $ \cclcdomwp $ is not  even 
closed under weak isomorphisms. The notion of weak isomorphism will still be useful for us at a later stage.

\begin{rem} \label{r_undfunctor}
If $\funcf: \cclcdomwp \to \cclcdomwp$ is a functor, there is a unique {\em underlying functor} 
$\hat{\funcf}: \cdom \to \cdom $ which maps the underlying domain of a domain-per $\mathcal{D}$ to
the underlying domain of $\funcf(\mathcal{D})$. 
\end{rem}

The following lemma shows that the restriction to convex, local and complete domain-pers works well
with strictly positive operations.

\begin{lem} \label{l_lwc4}
Let $ \Gamma $ be a strictly positive operation  on domain-pers. Assume that all non-positive 
parameters in $ \Gamma $ are dense and  that all positive parameters are convex, local and complete.

If  $  \mathcal{D}  $ is  a convex, local and complete domain-per, then $ \Gamma (  \mathcal{D}  ) $ is convex, 
local and complete.
\end{lem}

\proof
The proof is by structural induction on $ \Gamma $.
The base cases are trivial, and it is easily verified that both the  product and
the sum of two convex, local and complete domain-pers are again convex, local 
and complete, so we only prove the step involving exponentiation:

If $  \mathcal{B} $ is an arbitrary domain-per and $  \mathcal{D} $ is convex,
then $ [ \mathcal{B} \rightarrow  \mathcal{D} ] $ is convex:
Let $f, g, h \in [B \to D] $ and assume that $ f \sqsubseteq g \sqsubseteq h $ and 
$ f \approx_{ [ \mathcal{B} \rightarrow  \mathcal{D} ]} h $. If $ x \in   \mathcal{B}^{R} $, we have
$ f(x) \sqsubseteq g(x) \sqsubseteq h(x) $ and $ f(x) \approx_{ \mathcal{D}} h(x) $, which implies
$ f(x) \approx_{ \mathcal{D}} g(x) $ since $  \mathcal{D} $ is convex.
This shows that $ f \approx_{ [ \mathcal{B} \rightarrow  \mathcal{D} ]} g $ by lemma~\ref{l_equiv1}.

If $  \mathcal{B} $ is dense and $  \mathcal{D} $ is convex, local and complete, then
 $  [ \mathcal{B} \rightarrow  \mathcal{D} ] $ is local and complete:
Assume that $  f \approx_{ [  \mathcal{B} \rightarrow  \mathcal{D} ]} g  $.
If $ x \in   \mathcal{B}^{R} $, then $  f(x) \approx_{ \mathcal{D}} g(x) $,   and since $  \mathcal{D} $ is local,
$ f(x) $ and $ g(x) $ are consistent.
This implies that $f(p) $ and $g(p) $ are consistent for every $p \in B_c$ since  $  \mathcal{B}^R $ is dense
in  $  \mathcal{B} $, so $ f $ and $ g$ are consistent in $ [ B \rightarrow D ] $.
Moreover, $ [ f ]_{ [ \mathcal{B} \rightarrow  \mathcal{D} ]  } $ is a consistent set and $
 h:= \bigsqcup [ f ]_{ [ \mathcal{B} \rightarrow  \mathcal{D} ]  } $ exists.
If  $ x \in  \mathcal{B}^{R} $, then $ h ( x) = \bigsqcup \{ g(x) : f \approx_{ [ \mathcal{B} \rightarrow  \mathcal{D} ]} g  \} $,
and in particular $ f( x) \sqsubseteq h(x) \sqsubseteq \bigsqcup [ f(x) ]_{  \mathcal{D}  } $.
This shows that  $ f( x) \approx_{ \mathcal{D}} h(x) $ since $  \mathcal{D} $ is convex, and 
$ f \approx_{ [ \mathcal{B} \rightarrow  \mathcal{D} ]  } h $  by lemma~\ref{l_equiv1}. \qed

\begin{prop} \label{p_equid2}
Directed systems in $ \cclcdomwp $ admit inductive limits.
\end{prop}

\proof
Let  $ I= ( I, \leq ) $ be a directed partial order and let $ ( \{ \mathcal{D}_{i}  \}_{ i \in I }, \{ f_{i,j} \}_{ i \leq j \in I })  $ be a directed system over $I$ in $ \cclcdomwp $. Let $ \approx_{i} $ be the per on $  \mathcal{D}_{i} $. 
 
$ ( \{ D_{i}  \}_{ i \in I }, \{ f_{i,j} \}_{ i \leq j \in I })  $
is a directed system in $ \cdom $, so let $  (D_{I} ,  \{ f_{i }  \}_{ i \in I } ) $ be its inductive limit.
We define a binary relation $ \approx_{I} $ on $ D_{I} $ as follows:
$ x \approx_I x'$  if and only if there exists $ i \in I $  such that
$x_i \approx_i x'_i$ and such that \[ \forall k\geq i \: (f_{i,k}(x_i)\approx_k x_k \wedge
f_{i,k}(x'_i)\approx_k x'_k).\]
We then say that $ x \approx_I x'$ is {\em witnessed by} $i$. This is clearly a symmetric relation, and it also
has a number of other nice properties:
\begin{enumerate}[$\bullet$]

\item{Let  $ x \approx_I x'$ be witnessed by $i$ and let $j\geq i
  $. Then  $ x \approx_I x'$ is witnessed by $j$:\ }
We have $x_j \approx_j f_{i,j}(x_i) \approx_j f_{i,j}(x'_i) \approx_j x'_j$ because $j\geq i $ and $f_{i,j} $ is equivariant.
Moreover, if $k \geq j$, then $f_{j,k} $ is equivariant, and  it follows that
\[f_{j,k}(x_j) \approx_k f_{j,k} (f_{i,j} (x_i)) = f_{i,k}(x_i) \approx_k x_k .\]
This shows that  $f_{j,k}(x_j) \approx_k x_k $ and by symmetry, $f_{j,k}(x'_j) \approx_k x'_k$.

\item{$\approx_I$ is transitive:\ }
Assume that $x \approx_I x'$ is witnessed by $i $ and that $x' \approx_I x'' $ is
witnessed by $j$. We can then choose a common witness $k \geq i,j$. The 
transitivity of $\approx_k$ gives $x_k \approx_k x''_k$. If $ l \geq k $,
then \[f_{k,l}(x_k) \approx_l f_{k,l} ( f_{i,k} (x_i)) = f_{i,l}(x_i) \approx_l x_l.\]  
By symmetry, using $j$, we obtain $f_{k,l}(x''_k) \approx_l x''_l$. This shows that
$ x \approx_I x''$, witnessed by $k$.

\item{If $ x \approx_I x$ is witnessed by $i$ and $ x \approx_I x'$,
  then $ x \approx_I x'$ is witnessed by $i$:\ }  
Choose $j \geq i $ such that $ x \approx_I x'$ is witnessed by $j$. We show that $x_k \approx_k x'_k$ for $k \geq i $:
Choose some $l \geq j,k$. Firstly, $f_{k,l} (x_k) \approx_l x_l$  since $l \geq i $ is a witness that $ x \approx_I x$.
Secondly, $x_l \approx_l x'_l$ since $l \geq j $ is a witness that  $ x \approx_I x'$. 
Transitivity of $\approx_l$ gives $f_{k,l} (x_k) \approx_l x'_l$ which implies $x_k \approx_k x'_k $ since $f_{k,l} $ is
an equiembedding.

In particular,  $x_i \approx_i x'_i$ and $ f_{i,k}(x'_i) \approx_k f_{i,k}(x_i) \approx_k x_k \approx_k x'_k $   for  
arbitrary $k\geq i$.
\end{enumerate}

\noindent We have now shown that $\approx_I$ is a partial equivalence relation. We denote the domain-per $(D_I, \approx_I )$ 
by $\mathcal{D}_I$. We have also seen that the equivalence classes formed by this per have uniform witnesses, so we
may choose representatives as we like.

We will now show that $\mathcal{D}_I$ is convex, local and complete: \begin{enumerate}[$\bullet$]
\item{$\mathcal{D}_I$ is convex:\ }
Assume that $x \approx_I y $ is witnessed by $i$ and that $x,z \sqsubseteq y $, and  let $ k \geq i $.
By projection  $x_k \approx_k y_k $ and $ x_k , z_k\sqsubseteq y_k $, and since $\mathcal{D}_k $ is 
convex this gives $x_k \approx_k x_k \sqcup z_k $. In particular,  $x_i \approx_i x_i \sqcup z_i $.
Moreover, $ x_k \approx_k f_{i,k}(x_i) $, and $ f_{i,k}(x_i)  \approx_k f_{i,k}(x_i \sqcup z_i) $ since 
$f_{i,k}$ is equivariant. By transitivity of $ \approx_k$, this gives  $  f_{i,k}(x_i \sqcup z_i
)\approx_k  x_k \sqcup z_k $. This shows that $ x \approx_I x \sqcup z $, witnessed by $i$.

\item{$\mathcal{D}_I$ is local:\ } 
Let  $ x \in \mathcal{D}_I^R$. We prove that $ [x]_{ \mathcal{D}_I} $ is consistent by showing that an arbitrary finite subset of  $ [x]_{ \mathcal{D}_I} $ is consistent:

Choose $x' , x'' , \dots , x^{(n)} \in  [x]_{ \mathcal{D}_I}$. Then there exists some uniform witness $i \in I $ such that $x_i  \in  \mathcal{D}_i^R$
and  $x'_i , x''_i , \dots , x^{(n)}_i \in  [x_i]_{ \mathcal{D}_i}$.

For an arbitrary $j \in I $, choose some $ k \geq i,j$. 
We have $x'_k , x''_k , \dots , x^{(n)}_k \in  [x_k]_{ \mathcal{D}_k}$.
Since $\mathcal{D}_k$ is local, this means that $\{ x'_k , x''_k , \dots , x^{(n)}_k \} $ is consistent in $D_k$.
The set of projections, $\{ x'_j , x''_j , \dots , x^{(n)}_j \} $, is then consistent in $D_j$.
Since $j $ was arbitrarily chosen, this shows that $ \{ x' , x'' , \dots , x^{(n)} \} $ is consistent in $D_I$.

\item{$\mathcal{D}_I$ is complete:\ }
Assume that $x \approx_I x$ is witnessed by $i$. For each $k \geq i $, we have
\[ \{ x'_k : x \approx_I x' \} = \{ u \in D_k : x_k \approx_k u \}. \]
Let $\overline{x}:=\bigsqcup\{x': x\approx_Ix'\}$. We show that  
$x \approx_I \overline{x}$, witnessed by $i$:
\begin{enumerate}[$-$]
\item $\overline{x}_i=\bigsqcup\{x'_i:x\approx_Ix'\}=\bigsqcup\{u:x_i\approx_i
u\}\approx_ix_i$ since $\mathcal{D}_k$ is complete; and
\item for each $k\geq i $, we have $\{ f_{i,k}(u):x_{i}\approx_iu\}\subseteq
\{v:x_k\approx_kv\}$, and \[f_{i,k}(x_i) \sqsubseteq f_{i,k}(\overline{x}_i)
=\bigsqcup \{ f_{i,k}(u):x_{i}\approx_iu\}\sqsubseteq\bigsqcup\{v:x_k
\approx_k v\}=\overline{x}_k.\]
Since $\mathcal{D}_k$ is convex, this shows that $f_{i,k}(\overline{x}_i) 
\approx_k \overline{x}_k$.
\end{enumerate}
\end{enumerate}

It remains to show that $(\mathcal{D}_I,\{ f_i\}_{i\in I})$ is an inductive limit:
\begin{enumerate}[$\bullet$]

\item{Each $f_i:\mathcal{D}_{i}\to \mathcal{D}_{I}$ is an
  equiembedding:\ }
Trivially, $f_i$ is an equivariant embedding. 

Assume that $u \in \mathcal{D}_i^R$ and that $f_i(u)\approx_I x$ is witnessed by some $j\geq i$.
Then $f_{i,j}(u)\approx_j x_j$, and since $f_{i,j} $ is an equiembedding, this implies that
$u\approx_i x_i$.

\item{$(\mathcal{D}_I,\{ f_i\}_{i\in I})$ is universal in
  $\cclcdomwp$:\ }
Let $ \mathcal{E}$ be a convex, local and complete domain-per and let $\{ g_i:
\mathcal{D}_i \to \mathcal{E} \}_{i \in I} $ be a family of equiembeddings such
that $ g_i=g_j\circ f_{i,j} $ whenever $i\leq j$.
At domain level, there exists a unique embedding $g_I :D_I \to E $ such that
$ g_i=g_I \circ f_i $ for all $i \in I $, and it is defined by $ g_I(x) =
\bigsqcup_{i\in I} g_i(x_i)$.
We show that $g_I: \mathcal{D}_I \to \mathcal{E} $ is an equiembedding:
\begin{enumerate}[$-$]
\item
If $x \approx_I x'$ and this is witnessed by $i$, then $x_k \approx_k x'_k$ for all $k
\geq i$ and $g_I (x) = \bigsqcup_{k\geq i} g_k(x_k) \approx_{\mathcal{E}}g_i(x_i)$,
since $\mathcal{E}$ is convex. By symmetry, $g_I (x')  \approx_{
\mathcal{E}}g_i(x'_i)$. Moreover, $g_i(x_i) \approx_{\mathcal{E}}
g_i(x'_i)$ since $g_i$ is equivariant. This shows that $g_I(x)  \approx_{\mathcal{E}}
g_I(x')$, and that $g_I$ is equivariant.

\item
Assume that $x \in \mathcal{D}_I^R$, witnessed by $i$, and that $g_I (x) \approx_{
\mathcal{E}} y$.
Then, for each $k\geq i$, we have $f_{ik}(x_i)\approx_kx_k$  and $g_i(x_i)
\approx_{\mathcal{E}} g_k (x_k) $, since $g_k$ is equivariant. Thus,
$g_I(x)=\bigsqcup_{i\in I} g_i(x_i)\approx_{\mathcal{E}} g_i(x_i)$. It follows that
$g_k(x_k) \approx_{\mathcal{E}}y$ and $x_k \approx_k g_I^-(y)_k$. This holds for
arbitrary $k\geq i $, so it shows that $x \approx_I g_I^-(y)$. \qed
\end{enumerate}
\end{enumerate}

\begin{rem}
In the case of $ (I, \leq ) $ being a well-order, we define 
\[ \rank_{I} (x) : =  \min \{ i \in I :  x \in f_{i} (  \mathcal{D}_{i}^{R} ) \} ,\]
for an arbitrary $ x \in \mathcal{D}_{I}^{R} $.
If $ x \approx_{I} y $, then $ \rank_{I} (x) = \rank_{I} (y) $, since the fact that each $ f_{i,j} $ is an equiembedding ensures that equivalent elements are introduced at the same level.
\end{rem}

\begin{defi} \label{d_uniformmapping}
Let $ ( \mathcal{D} , f ) :=  ( \{ \mathcal{D}_{i}  \}_{ i \in I }, \{ f_{i,j} \}_{ i \leq j \in I })  $
and $  ( \mathcal{E} , g ) := ( \{ \mathcal{E}_{i}  \}_{ i \in I }, \{ g_{i,j} \}_{ i \leq j \in I })  $ 
be directed systems over the same directed partial order $ (I, \leq )$.

A {\em uniform mapping} from $  ( \mathcal{D} , f ) $ into $ ( \mathcal{E} , g ) $ is a 
a family $ \varphi = \{ \varphi_{i} :  \mathcal{D}_{i} \rightarrow  \mathcal{E}_{i}  \}_{i \in I }  $ of equivariant maps such that
$ g_{i,j} \circ \varphi_{i} = \varphi_{j} \circ f_{i,j}  $ whenever $ i \leq j \in I $.
\end{defi}

\begin{lem} \label{l_equid3}
Let $ \varphi  : ( \mathcal{D} , f ) \rightarrow  ( \mathcal{E} , g )$ be a uniform mapping.

Then there exists a unique equivariant $ \varphi_{I} : \mathcal{D}_{I}  \rightarrow \mathcal{E}_{I}  $ such that 
$  \varphi_{I} \circ f_{i} = g_{i} \circ \varphi_{i} $ for every $ i \in I $.
Moreover, if $ \chi = \{ \chi_{i} \}_{ i \in I } : ( \mathcal{E} , g ) \rightarrow ( \mathcal{D} , f ) $ is a uniform mapping
such that $ ( \varphi_{i} , \chi_{i} ) $ is a weak  isomorphism pair for every $ i \in I $,
then  $ ( \varphi_{I} , \chi_{I} ) $ is  a weak  isomorphism pair.
\end{lem}

\proof
If $ i \leq j $, then 
$ g_{i} \circ \varphi_{i} \circ f_{i}^{-} 
= g_{j} \circ g_{i, j} \circ \varphi_{i} \circ f_{i}^{-} 
= g_{j} \circ \varphi_{j} \circ f_{i,j} \circ f_{i}^{-} 
\sqsubseteq g_{j} \circ \varphi_{j} \circ f^{-}_{j} $, 
so $ \{ g_{i} \circ \varphi_{i} \circ f_{i}^{-}  : i \in I \} $ is a directed set in $ [ D_{I} \rightarrow E_{I} ] $.
Let $ \varphi_{I}  $ be its least upper bound.
Furthermore,
$ g_{j} \circ \varphi_{j} \circ f^{-}_{j} \circ f_{i} 
= g_{j} \circ \varphi_{j} \circ f_{i, j}
= g_{j} \circ  g_{i,j} \circ \varphi_{i} 
= g_{i} \circ \varphi_{i} $, whenever $ i \leq j$, thus
$  \varphi_{I} \circ f_{i} =   \bigsqcup_{j \geq i}  ( g_{j} \circ \varphi_{j} \circ f^{-}_{j} ) \circ f_{i} =  g_{i} \circ \varphi_{i} $ for every $ i \in I $. 

$  \varphi_{I} $ is equivariant:
Let $ x, y \in D_{I} $ and assume that $ x \approx_{I} y $. Then there exists some $ i \in I $ such that
$ x = f_{i} (x_{i} ) $, $ y = f_{i} (y_{i} )  $ and $ x_{i} \approx_{i}  y_{i} $.
Both $ g_{i} $ and  $ \varphi_{i} $ are equivariant, so this gives 
\[ g_{I} (x) = (g_{i} \circ \varphi_{i}) ( x_{i} )  \approx_{I} ( g_{i} \circ \varphi_{i}) (  y_{i} ) =  g_{I} (y). \]

Let $ \chi = \{ \chi_{i} \}_{ i \in I } :  ( \mathcal{E} , g ) \to  ( \mathcal{D} , f ) $ be a uniform mapping such that 
$ ( \varphi_{i} , \chi_{i} ) $ is a weak  isomorphism pair for every $ i \in I $.
Let $ x \in \mathcal{D}_{I}^{R} $ and choose $ i \in I $ such that $ x= f_{i} (x_{i} ) $ and $ x_{i} \in \mathcal{D}_{i}^{R} $.
Then
\[ (\chi_{I} \circ \varphi_{I})  (x) =  (\chi_{I} \circ  g_{i} \circ \varphi_{i} )( x_{i} ) 
= ( f_{i} \circ \chi_{i} \circ \varphi_{i} ) (x_{i} ) 
\approx f_{i} (x_{i} ) = x .\]
This shows that $\chi_I \circ \varphi_I $, and  $ \varphi_{I} \circ \chi_{I}  \approx \id_{E} $ by a symmetric argument.
\qed

Isomorphism of domains is preserved under inductive limits of directed systems.
This means that if the weak isomorphism pairs are actual isomorphisms in $\cclcdomwp $, i.e.
isomorphism pairs of the underlying domains, then the inductive limits are isomorphic as well.
In particular, the inductive limit of a directed system in $\cclcdomwp $ is unique up to isomorphism.

\subsection{A least fixed point}

A functor $ \funcf : \cclcdomwp \to \cclcdomwp $ is  {\em strictly positive} if it is the functorial
extension of a strictly positive operation on domain-pers.
We will now show that such a functor has a least 
fixed point.
In domain theory, least fixed points are constructed by the means of $\omega$-chains.
For domain-pers, we will need uncountable chains over $\cclcdomwp$.

\begin{defi}
Let $ \funcf : \cclcdomwp \to \cclcdomwp $ be a  strictly positive functor 
and let $\beta $ be a limit ordinal.
We construct a {\em $\beta$-chain  $  ( \{  \mathcal{D}_{\alpha } \}_{ \alpha \in \beta } , \{  
f_{ \alpha, \alpha' } \}_{ \alpha \leq \alpha' \in \beta } ) $  from  $ \funcf $} as follows:

\begin{enumerate}[$\bullet$]
\item Let $ \mathcal{D}_{0} $ be the  initial object in $ \cclcdomwp $, i.e.
the trivial domain $ \{ \bot \} $ with the empty per.

\item If $ \alpha \in \beta $, let $  
\mathcal{D}_{\alpha + 1} :=  \funcf (  \mathcal{D}_{\alpha } ) $ and let
 $ f_{0, \alpha +1  }  $ be the unique equiembedding from $  \mathcal{D}_{ 0 } $ into  $\mathcal{D}_{  \alpha +1 } $.

\item If $ \alpha \leq \alpha' \in \beta  $, let $  f_{ \alpha+1 , \alpha'+1} := \funcf (  f_{ \alpha , \alpha'} ) $.

\item If $ \gamma \in \beta $ is a limit ordinal, let  $  (  \mathcal{D}_{\gamma },  \{  f_{ \alpha, \gamma } \}_{ 
\alpha  \in \gamma } ) $ be the inductive limit of  the $ \gamma $-chain $  ( \{  \mathcal{D}_{\alpha } \}_{ \alpha 
\in \gamma } , \{  f_{ \alpha, \alpha' } \}_{ \alpha \leq \alpha' \in \gamma } ) $.
\end{enumerate}
\end{defi}

\noindent The underlying functor $  \hat{ \funcf} : \cdom \to \cdom  $ of a strictly positive $\funcf$ 
is obtained by replacing each parameter in $\funcf$ by its underlying domain and each basic 
operation by the corresponding basic operation on domains, so  clearly it is strictly positive as well.
The $\omega$-chain $ (  \{   D_{n} \}_{n \in \omega } , \{  f_{m,n} \}_{  
m \leq n \in \omega } ) $ coincides with the $\omega$-chain used in the least
fixed point construction for $\hat{\funcf}$, and $ D_{ \omega } $ is a least fixed point of $ \hat{ \funcf } $. 
However,  $  \mathcal{D}_{\omega} $ is not in general a fixed point of $ \funcf $, as
the example below shows.

We  have $ D_{ \omega} \cong  D_{ \alpha } $ for all $ \alpha\geq\omega$, since isomorphisms are preserved under inductive 
limits in $ \cdom $, so   $ f_{ \alpha , \alpha + 1 } $ is an isomorphism in $ \cclcdomwp $ 
if and only if $ f_{ \alpha , \alpha +1  }^{-} $ is equivariant.

\begin{exa}
Let  $  \mathcal{A} $ be some non-trivial domain-per and let $  \mathcal{N} = ( \mathbb{N}_{ \bot} , =\vert_{ \mathbb{N }} )$, 
with  $  \mathbb{N}_{ \bot} $  the  flat domain of natural numbers.
Let $ \funcf : \cclcdomwp \to \cclcdomwp $ be defined by $ \funcf ( X) = \mathcal{A} +  [ \mathcal{N} \to X ] $.

We show that $f^-_{\omega,\omega+1} $ is not equivariant:  
Choose some $  a \in \mathcal{A}^{R} $ and let $ x_{0} := (0,a) $.
If $ x_{n} \in \mathcal{D}_{\omega}^{R} $, let $ \varphi_{n}:\nat_\bot \to D_\omega $ be the function 
constantly  equal to $ x_{n} $  and let $ x_{n+1} := (1, \varphi_{n} ) $.
Let $ \varphi :\nat_\bot \to D_\omega$ be the strict function defined by $ \varphi(n) = x_n$.
Then $\varphi\in\funcf(  \mathcal{D}_{\omega})^{R} $.
However, $ f^-_{\omega,\omega+1}(\varphi) = (1, \varphi )  \notin \bigcup_{ n \in \omega } \mathcal{D}_{n}^{R} = \mathcal{D}_{\omega }^{R}$,
since  $ \rank_{ \omega } (x_{n} ) = n $ for every $n\in\nat$.
\end{exa}

\begin{lem} \label{l_lfp1}
Let $ \funcf : \cclcdomwp \rightarrow \cclcdomwp $ be strictly positive.

Then there exists a limit ordinal $ \gamma_0 $ such that if 
$  ( \{  \mathcal{D}_{\alpha } \}_{ \alpha \in \gamma_0 } , \{  f_{ \alpha, \beta } \}_{ \alpha \leq \beta \in \gamma_0 } ) $
is the $ \gamma_0 $-chain constructed from $ \funcf $, then
 $ f_{ \alpha , \alpha + 1} $ is an isomorphism in $ \cclcdomwp $ for some $\alpha\in\gamma_0$.
\end{lem}

\proof
Choose $ \vert \gamma_0 \vert >  \vert D_{ \omega } \vert $ and assume for contradiction that $  f_{ \alpha , \alpha +1}^{-} $ 
is not equivariant for any $ \alpha \in \gamma_0 $. 
For every $ \alpha \in \gamma_0 $, there exists some $ x \in \mathcal{D}_{\alpha+1 }^{R} $ such that 
$  f_{ \alpha , \alpha +1}^{-} (x) \notin \mathcal{D}_{\alpha }^{R} $ and $ \rank_{\gamma_0 } (x) = \alpha + 1 $.
This shows that $ \rank_{\gamma_0 } $ is  a surjective function from $  \mathcal{D}_{\gamma_0 }^{R} $ onto the set of successor ordinals below $ \gamma_0 $,
so $ \vert \gamma_0 \vert  \leq \vert   \mathcal{D}_{\gamma_0 }^{R} \vert $.
On the other hand, 
$  \vert \mathcal{D}_{\gamma_0 }^{R}  \vert \leq \vert D_{ \gamma_0 } \vert  = \vert  D_{ \omega } \vert < \vert \gamma_0 \vert $,
which is a contradiction.
\qed

\begin{prop} \label{p_lfp2}
If $ \funcf : \cclcdomwp \to \cclcdomwp $ is   strictly positive,
there is an initial $ \funcf$-algebra.
\end{prop}

\proof
By lemma~\ref{l_lfp1}, we can choose $ \gamma_0 \geq \omega $ such that 
$   ( \mathcal{D}_{\gamma_0 }, f_{ \gamma_0 , \gamma_0 +1 }^{-} )  $ is an $ \funcf $-algebra, but it 
remains to prove that it   is initial. 

Let $ ( \mathcal{E} , g ) $ be an arbitrary 
$ \funcf $-algebra. We will then show that there exists  a unique equiembedding 
$h:\mathcal{D}_{\gamma_0} \to\mathcal{E} $ satisfying
$  h \circ  f_{ \gamma_0 , \gamma_0 +1 }^{-} = g \circ \funcf (h )$.

\begin{claim}
There exists a family  $  \{ h_{ \beta } :  \mathcal{D}_{ \beta } \rightarrow  \mathcal{E} \}_{ \beta \leq \gamma_0 } $ of equiembeddings such that,
for each $ \beta \in \gamma_0 $,
$ h_{ \beta } = g \circ \funcf (h_{ \beta } ) \circ f_{ \beta , \beta + 1}  = h_{ \beta + 1} \circ f_{ \beta , \beta + 1 } $. 
\end{claim}

The claim is proved by  transfinite induction on $ \beta $, see appendix~\ref{app.proofs} for details.
In particular, this gives us an equiembedding $h_{\gamma_0} : \mathcal{D}_{\gamma_0}\to\mathcal{E}$. 

From domain theory, we  see that $ (E,g) $ is an  $ \hat{  \funcf  }  $-algebra and that $ h_{ \gamma_0 } $ is an $\hat{  \funcf  }  $-morphism
from $  ( D_{ \gamma_0 } , f_{ \gamma_0, \gamma_0 +1 }^{-} ) $ into  $ (E,g) $.
It remains to prove that $  h_{ \gamma_0 } \circ  f_{ \gamma_0 , \gamma_0 +1 }^{-} = g \circ \funcf (h_{ \gamma_0  } )$:

\begin{claim}
Let $ \omega \leq \beta \leq \gamma_0 $ and assume that $ f_{ \beta , \beta + 1} $ is an isomorphism. 
Then  $  h_{ \beta } : D_{ \beta } \to E $ is the unique 
$\hat{  \funcf  }  $-morphism from 
$  ( D_{ \beta } , f_{ \beta, \beta +1 }^{-} ) $
into  $ (E,g) $.
\end{claim}

The proof is by transfinite induction on  $ \beta \leq \gamma_0$, see appendix~\ref{app.proofs}.
As a consequence of uniqueness, we have  $  h_{ \beta } \circ  f_{ \beta , \beta +1 }^{-} = g \circ \funcf (h_{ \beta  } )$ for all infinite $\beta\leq \gamma_0$.
In particular, $ h_{ \gamma_0 } : D_{ \gamma_0 } \to E $ is the unique embedding satisfying
\[  h_{ \gamma_0 } \circ  f_{ \gamma_0 , \gamma_0 +1 }^{-} = g \circ \funcf (h_{ \gamma_0  } ). \]
Then $h_{\gamma_0} : \mathcal{D}_{\gamma_0}\to\mathcal{E}$ is the unique equiembedding for which this
equality holds, and this shows that the $ \funcf $-algebra $   ( \mathcal{D}_{\gamma_0 }, f_{ \gamma_0 , \gamma_0 +1 }^{-} )  $ 
is initial.
\qed

When we consider domain representations of $qcb_0$ spaces, countably based domain-pers are of particular interest.
It is therefore important to note that if we start with countably based  parameters, then the least fixed point is countably based, 
even though we might have to use an uncountable transfinite induction to construct it:

\begin{obs} \label{o_lfp2}
Let $ \funcf : \cclcdomwp \rightarrow \cclcdomwp $ be strictly positive and assume that all  parameters are countably based.

Then the least fixed point of $ \funcf $ is countably based.
\end{obs}

\proof
The parameters in the  underlying functor $ \hat{ \funcf } $ are countably based, so  the least fixed point $ D_{ \omega } $
of $\hat{\funcf}$ is countably based.
If $\mathcal{D}_{\gamma_0} $ is the least fixed point of $\funcf$, then $ D_{ \gamma_0 } \cong D_{ \omega } $ and $D_{\gamma_0}$ is
countably based.
\qed

In many examples of interest, e.g. for representation of a countably based regular space, we can choose a domain-per 
$ \mathcal{D} $ which is upwards-closed, i.e. a domain-per which satisfies
\[ \forall x, y \in D \: 
( x \in \mathcal{D}^{R} \wedge x \sqsubseteq y  \Rightarrow x \approx_{ \mathcal{D} } y ) .\]
An upwards-closed domain-per is convex, local and complete.
We verify that the property of being upwards-closed is preserved under the least fixed point construction in $\cclcdomwp$:

\begin{obs}  
Let $ \funcf : \cclcdomwp \rightarrow \cclcdomwp $ be strictly positive and assume that all  positive parameters are upwards-closed.

Then the least fixed point of $ \funcf $ is upwards-closed.
\end{obs}

\proof
It is sufficient to prove that if
$ ( \{   \mathcal{D}_{i} \}_{i \in I } , \{ f_{i,j} \}_{ i \leq j \in I } ) $ is a directed system in $ \cclcdomwp $
and every $  \mathcal{D}_{i} $ is upwards-closed, 
then the inductive limit $  \mathcal{D}_{I} $ is upwards-closed.

Let $ x,y \in D_{I} $ and assume that $ x \in \mathcal{D}_{I}^{R} $ and $ x \sqsubseteq y $.
Then $ x_{i} \sqsubseteq y_{i} $ for every $ i \in I $, so if  $ x \in \mathcal{D}_{I}^{R} $ is witnessed by $i$,
then $ x_{i} \approx_{i} y_{i} $ by the upwards-closedness of $ \mathcal{D}_{i} $.
If $ k \geq i $, then $ f_{i,k} ( x_{i} ) = x_{k} \approx_{k} y_{k} $, and because $ f_{i,k} $ is an equiembedding, this implies
that $ y_{k} = f_{ i,k} ( y_{i} ) $.
For an arbitrary $ j \in I $, choose $ k \geq i,j $.
Then $ y_{j} = f_{j,k}^{-} (y_{k} ) = f_{j,k}^{-} (  f_{ i,k} ( y_{i} ) ) = f_{i} (y_{i} )_{j} $.
This shows that $ y =  f_{i} (y_{i} ) $ and that $ x \approx_{I} y $.
\qed

\subsection{Admissible  domain-pers}

We say that a domain-per $ \mathcal{D} $ is {\em admissible} if the associated 
domain representation  $ (  \mathcal{D} ,  \mathcal{D}^{R} , \delta_{  \mathcal{D} } ) $ 
of $ \mathcal{Q}   \mathcal{D}  $ is admissible. Our definition of admissibility applies only to
countably based domain representations, so it will be implicit that an admissible domain-per is
countably based in what follows, even though this is of no significance for the results obtained.

We  show that admissibility is preserved both under strictly positive operations and
under weak isomorphisms.

\begin{lem} \label{l_admdp1}
Let  $ \Gamma $ be a strictly positive operation on domain-pers with admissible positive parameters and dense, admissible non-positive parameters.

If  $ \mathcal{D} $ is an admissible domain-per, then  $ \Gamma ( \mathcal{D} ) $ is admissible.
\end{lem}

\proof
By  lemma~\ref{l_adm5} and structural induction on  $\Gamma$.
\qed

\begin{lem}  \label{l_admdp2}
Let  $ \mathcal{D} $ and  $ \mathcal{E} $ be weakly isomorphic domain-pers.

Then $ \mathcal{Q}  \mathcal{D}  \cong \mathcal{Q}   \mathcal{E}  $, and if
 $ \mathcal{D} $ is admissible, then  $ \mathcal{E} $ is admissible.
\end{lem}

\proof 
Let $ (f,g):  \mathcal{D}   \rightarrow  \mathcal{E}  $ be a weak isomorphism pair, i.e. $ g \circ f \approx \id_D $ and 
$f \circ g \approx \id_E $.

Then  $ f^{\mathcal{Q}} :  \mathcal{Q}  \mathcal{D}    \rightarrow  \mathcal{Q}  \mathcal{E}   $
and  $ g^{\mathcal{Q}} :  \mathcal{Q}  \mathcal{E}    \rightarrow  \mathcal{Q}  \mathcal{D}   $
are continuous maps with $ g^{\mathcal{Q}}\circ f^{\mathcal{Q}}= \id_{\mathcal{Q}\mathcal{D}} $ and
$ f^{\mathcal{Q}}\circ g^{\mathcal{Q}}= \id_{\mathcal{Q}\mathcal{E}} $. This shows that $
\mathcal{Q}\mathcal{D} $ and $ \mathcal{Q}\mathcal{E}$ are homeomorphic topological spaces.

Now, assume that $\mathcal{D}$ is admissible. Let $ F $ be a domain, let $ F^{R}  $ be a dense 
subset  and let $ \varphi : F^{R} \rightarrow  \mathcal{Q}  \mathcal{E} $ be a continuous function.
Then  $ g^{\mathcal{Q}}\circ\varphi: F^{R} \rightarrow  \mathcal{Q}  \mathcal{D} $ is a continuous
function which factors through $\delta_{\mathcal{D}} $ via $\chi:F\to D$, by the admissibility of 
$\mathcal{D}  $.
Then $f\circ\chi:F\to E$ is a continuous function which satisfies \begin{enumerate}[(1)]
\item $(f\circ\chi)[F^R]\subseteq f[\mathcal{D}^R]\subseteq\mathcal{E}^{R} $ ; and 
\item if $x\in F^R$, then $ \delta_{   \mathcal{D}  }(\chi(x))= g^{\mathcal{Q}}(\varphi(x))$ and
$  \delta_{ \mathcal{E} }(f(\chi(x)))= f^{\mathcal{Q}}(g^{\mathcal{Q}}(\varphi(x)))=\varphi(x)$.
\end{enumerate} This shows that $ \varphi$  factors through  $  \delta_{  \mathcal{E} }  $ via $ f \circ \chi$
and that $  \mathcal{E} $ is  admissible.\qed

The converse of this lemma is not true in general, but it does hold if we consider dense domain-pers:

\begin{lem} \label{l_admdp4}
If $ \mathcal{D} $ and  $ \mathcal{E} $ are dense, admissible domain-pers
and $ \mathcal{Q}  \mathcal{D}  \cong \mathcal{Q}   \mathcal{E}  $,
then $ \mathcal{D} $ and  $ \mathcal{E} $ are weakly isomorphic.
\end{lem}

\proof Continuous functions between dense, admissible domain-pers are representable.
In particular, a homeomorphism pair 
$ ( f , g ):  \mathcal{Q}  \mathcal{D}  \rightarrow  \mathcal{Q}  \mathcal{E}  $
is represented by a pair of continuous functions $ \hat{f}:  \mathcal{D}  \rightarrow  
\mathcal{E}$ and $ \hat{g}:   \mathcal{E}  \rightarrow  \mathcal{D} $.
The composite functions $  \hat{g} \circ  \hat{f} $ and $\hat{f} \circ  \hat{g}$
represent the respective identities, and this shows that
$  ( \hat{f} ,  \hat{g} ) $ is a weak isomorphism pair, i.e.
$  \hat{g} \circ  \hat{f} \approx \id_{D} $ and 
$ \hat{f} \circ  \hat{g} \approx \id_{ E } $.\qed

\begin{rem}
Lemma~\ref{l_admdp4} is just a reformulation of a  well-known result:
All  dense, admissible domain representations of a given topological space are continuously equivalent,
see   \cite{Bla_08}.
\end{rem}

The following observation makes an important connection between equiembeddings and
admissible domain-pers. In particular, it implies that the inductive limit of a directed system of 
domain-pers cannot be admissible unless all the domain-pers in the directed system are admissible.

\begin{obs} \label{o_admdp5}
Let $ f:  \mathcal{D} \rightarrow \mathcal{E} $ be an equiembedding and let  $ \mathcal{E}  $ be admissible.

Then $ \mathcal{D} $ is admissible.
\end{obs}

\proof
Let $ (F, F^{R} ) $ be a dense domain with totality  and assume that $ \varphi: F^{R} \to \mathcal{Q}   
\mathcal{D}  $ is continuous. Then $ \chi := f^{\mathcal{Q}} \circ \varphi :  F^{R} \to \mathcal{Q}   
\mathcal{E}  $ is continuous, and by the admissibility of $ \mathcal{E} $,
there exists a continuous $ \hat{ \chi } : F \rightarrow E $
such that $  \hat{ \chi } [ F^{R} ] \subseteq \mathcal{E}^{R} $ and
$ [  \hat{ \chi } (x) ]_{  \mathcal{E} } = \chi (x) $ for every $  x \in F^{R} $.

Let $ \hat{ \varphi } := f^{-} \circ  \hat{ \chi }  :F \rightarrow D $. We will show that $\varphi$ factors 
through $\delta_{\mathcal{D}}$ via $\hat{\varphi}$:
Let $ x \in F^{R} $, and choose some $ d \in \mathcal{D}^{R} $ such that $ [d]_{ \mathcal{D} } = \varphi (x) $.
Then $   f^{\mathcal{Q}} ( \varphi  (x))= [  \hat{ \chi } (x) ]_{  \mathcal{E} } $ 
and $ f(d) \approx_{ \mathcal{E} }  \hat{ \chi } (x) $.
Since $f $ is an equiembedding, this implies that $ d \approx_{ \mathcal{D} } f^{-}  ( \hat{ \chi }  (x)) =  \hat{ \varphi } (x) $.
Hence, $  \hat{ \varphi } [ F^{R} ] \subseteq \mathcal{D}^{R} $ and 
$ [ \hat{ \varphi } (x) ]_{ \mathcal{D} } = \varphi (x) $ for all $ x \in F^{R} $.
\qed

\subsection{A dense least fixed point} \label{ss_density}
Density is an important but problematic notion in the study of domain representations and domain-pers, see \cite{BBS_04,Ber_93,
Bla_00,Dah_07,Ham_05}. One major advantage is that it helps lifting of continuous functions, see lemma~\ref{l_adm3}. A
major issue with density is that it is not preserved by the  function space construction.
We will now show how  a domain-per which is defined by a strictly positive induction with
dense parameters, can be replaced by a dense domain-per.

It is well-known that given any domain representation of a topological space $X$, there is a
dense domain representation of the same space, see \cite{Dah_07,Ham_05}.
The following definition is just a  reformulation of this result.
From a given domain-per $\mathcal{D}$, we construct a dense domain-per with
the same set of total elements, using the topological closure of $ \mathcal{D}^{R} $.

\begin{defi} \label{d_densepart}
If $ \mathcal{D} $ is a domain-per, we define  $ \mathcal{D}^{d} $, {\em the dense part} of  $ \mathcal{D} $, 
as follows: \begin{enumerate}[$\bullet$] 
\item If $ \mathcal{D}$ is trivial, let $ \mathcal{D}^{d} = \mathcal{D}_{ 0} $, the initial domain-per.
\item If $ \mathcal{D}   $ is non-trivial, let $ \mathcal{D}^{d} $ be the domain-per with
\[ D^{d} =  \{ x \in D : \textrm{approx} (x) \subseteq  \{ p \in D_{c} : \upset{p} \cap  \mathcal{D}^{R} \neq \emptyset \} 
 \}  ,\]
partially ordered by the restriction of $ \sqsubseteq_{D} $ 
as the underlying domain, and with $  \approx_{  \mathcal{D}  } $ restricted to $ D^{d} $ as the per.
\end{enumerate} \end{defi}

\noindent Clearly, $  \mathcal{Q}  \mathcal{D} \cong  \mathcal{Q} ( \mathcal{D}^d)$, as the topology on
$ \mathcal{D}^R $ is the same when it is considered as a subspace of $ \mathcal{D}^d$.

For each of the basic operations, there is  a {\em basic dense operation} obtained by left composition with
the dense part construction.
Let $\Gamma$ be a strictly positive operation and let $\Gamma^d$ be the operation obtained by inductively
replacing all parameters and basic operations by their dense counter-parts. A simple structural induction on
$\Gamma$ then shows that $\Gamma^d(\mathcal{D})$ and $\Gamma(\mathcal{D})^d$ are equal up to isomorphism of
domain-pers.
Note here that density is preserved under disjoint sum and Cartesian product, so it is only the function space construction which demands extra attention.

We make some further important observations, but skip the proofs.

\begin{obs}
Let $ \mathcal{D} $ be a convex, local and complete domain-per.
Then $ \mathcal{D}^{d} $ is convex, local and complete.
\end{obs}

\begin{obs}  \label{o_densepartadmissible}
Let $ \mathcal{D} $ be an admissible domain-per.
Then $   \mathcal{D}^{d} $ is admissible.
\end{obs}

In fact, this is lemma 7.5 in \cite{Ham_05}: The dense part of an admissible representation is itself admissible.

\begin{obs}  \label{o_dense2}
Let $  ( \{  \mathcal{D}_{i} \}_{ i \in I } ,  \{f_{i,j} \}_{ i \leq j \in I }  ) $ be a directed system in $ \cclcdomwp $ and 
assume that  for every $ i \in I $ there exists some $ j \geq i $ such that 
$  \mathcal{D}_{j}  $ is dense.
If $  (   \mathcal{D}_{I}  ,  \{ f_{i} \}_{ i \in I }  ) $ is the inductive limit, then
$  \mathcal{D}_{I} $ is dense.
\end{obs}

Unfortunately, our choice of morphisms in $\cclcdomwp$ obstructs any restriction to a full subcategory of domain-pers which
are either dense or trivial. This is in contrast to the case for  $  \cdomainwp $, see  \cite{BBS_04}.
The specific problem which cannot be overcome is that the natural restriction of an equiembedding $ f :  \mathcal{D} \rightarrow  
\mathcal{E} $ to $\mathcal{D}^{d}$ is not in general an embedding into the underlying domain of $  \mathcal{E}^{d} $.
This means that the dense strictly positive operations would not be functorial in such a category.

Nevertheless, proposition~\ref{p_dense3} below shows that the dense part of  a strictly positive functor  $\funcf$ (or more precisely of the
underlying operation $\Gamma$)  produces a $\gamma$-chain in $ \cclcdomwp $. Moreover, the  inductive limit coincides 
with the dense part of the inductive limit of the $\gamma$-chain constructed from $\funcf$.

For certain strictly positive functors, the $\gamma$-chain will contain  trivial domain-pers only. 
Then even the least fixed point is trivial, and it will be convenient to leave these trivial cases aside. 
This motivates  the following definition:

\begin{defi}
A strictly positive functor $ \funcf : \cclcdomwp \to \cclcdomwp $  is   {\em trivial} if $ 
\funcf ( \mathcal{D}_{0}  ) $ is trivial and
{\em non-trivial} if $ \funcf ( \mathcal{D}_{0}  ) $ is non-trivial. 
\end{defi}

Non-trivial functors are characterized by the following lemma. The straight-forward inductive proof by cases is omitted here.

\begin{lem} \label{l_trivialfunctor}
Let  $ \funcf :\cclcdomwp \to  \cclcdomwp $ be a strictly positive functor. 

Then $ \funcf $ is  non-trivial if and only if at least one of the following statements hold:
\begin{enumerate}[$\bullet$]
\item
$\funcf$ is constantly equal to a non-trivial domain-per $ \mathcal{A} $. 
\item{} 
$\funcf$ is the disjoint sum of functors of which at least one  is non-trivial.
\item{} 
$\funcf$ is the Cartesian product of functors which both  are non-trivial.
\item{}
$\funcf$ is the exponentiation of  a non-trivial functor by a domain-per $ \mathcal{B} $.
\end{enumerate}
\end{lem}

\noindent The potential problem with trivial parameters is avoided by assuming that all parameters are dense. This lemma will simplify some proofs by
induction on the structure of a functor.

\begin{prop} \label{p_dense3}
Let $ \funcf :\cclcdomwp \to  \cclcdomwp $ be a  strictly positive functor with dense  parameters and 
let  $ \gamma $ be a limit ordinal.
Let $ ( \{ \mathcal{D}_{ \alpha}  \}_{ \alpha \in \gamma } , \{ f_{\alpha , \beta } \}_{ \alpha \leq \beta \in \gamma  } ) $
be the $ \gamma $-chain  constructed from $ \funcf $,
and let  $    ( \mathcal{D}_{ \gamma } , \{ f_{\alpha , \gamma } \}_{ \alpha \in \gamma}  )$ be its inductive limit.

Then $ ( \{ \mathcal{D}_{ \alpha}^{d}  \}_{ \alpha \in \gamma } , \{ f_{\alpha , \beta }^{d} \}_{ \alpha \leq \beta \in \gamma }  )$ is a 
$ \gamma $-chain in $ \cclcdomwp $ with  inductive limit
 $    ( \mathcal{D}_{ \gamma }^{d} , \{ f_{\alpha , \gamma }^{d} \}_{ \alpha \in \gamma}  )$.
\end{prop}

\proof
If $ \funcf $ is trivial, then  $  \mathcal{D}_{ \alpha} $ is trivial for each $\alpha\in\gamma$,
and the result holds trivially. Therefore, we may assume that $ \funcf $ is non-trivial.

We can use the same underlying domain $D$ of $ \mathcal{D}_{ \alpha} $ for all $\alpha\in\gamma$: 
If $ n \in \omega $, then $ D_{n} $ is isomorphic to a subdomain of $ D $,
so we may assume that $ \mathcal{D}_{ n} $ has $ D_\omega$ as the underlying domain.
If  $ \alpha \geq \omega $, then $ D_{ \alpha } $ is isomorphic to  $ D_{ \omega} $.
We let $ D = D_\omega$. 
Moreover, we assume that $f_{\alpha , \beta} =\id_D$ for all $ \alpha,\beta\in\gamma$.
This is a valid assumption, since we simply can redefine the directed system inductively if it does not hold.

We denote the per on  $ \mathcal{D}_{ \alpha}$ by $ \approx_{ \alpha }  $ for all $ \alpha \leq \gamma $.
If $ \alpha \leq \beta $,  
$  x \in  \mathcal{D}_{ \alpha}^{R} $ and $ x' \in D $, then $ x \approx_{ \alpha } x' \Leftrightarrow 
x \approx_{ \beta } x'$. If $ \beta $ is a limit ordinal, then 
$ x \approx_{\beta } x' \Leftrightarrow \exists \alpha \in \beta \; ( x \approx_{\alpha } x' ) $ for all $x,x'\in D$.

We denote the underlying domain of $\mathcal{D}^d_\alpha$ by $ D_{\alpha }^{d} $. Then 
 $ f_{\alpha , \beta }^{d} :  D_{\alpha }^{d}  \rightarrow  D_{\beta }^{d} $ is the inclusion map.
It remains to show that it is an embedding of domains and that 
$ D_{ \gamma }^{d} $ is the inductive limit of $ \{ D_{ \alpha}^{d}  \}_{ \alpha \in \gamma } $ in $ \cdom $.

\begin{claim}
There exists a family $ \{ \Delta_{n}\}_{n\in\omega} $ of closed subsets  of $ D$, closed under binary lubs, such that 
\begin{enumerate}[(1)]
\item $ \Delta_{n } \subseteq \Delta_{n+1} $; and 
\item $ \Delta_{n} \cap  \mathcal{D}_{ \alpha}^{R} = \mathcal{D}_{n}^{R} $ for every  $ \alpha \geq 
\omega $. \end{enumerate}
\end{claim}

We define the subsets $  \Delta_{n} \subseteq D$ by induction on $ n \in \omega $, see  appendix~\ref{app.proofs}
for the entire proof.
Note that $ D^{d}_{n} \subseteq  \Delta_{n} $ since $ \Delta_{n} $ is a closed superset of $  \mathcal{D}_{n}^{R} $ and
$ D^{d}_{n} $ is its closure.

\begin{claim}
Let $ n \geq 1 $. Then there exists a continuous map $ r_{n} : D \rightarrow D $ 
such that \begin{enumerate}[(1)]
\item $ \Delta_{n} = \{ x \in D:  r_{n } (x) = x \}$; and 
\item if  $ \alpha \geq n $, then $ r_{n } : \mathcal{D}_{ \alpha } \rightarrow  
\mathcal{D}_{n}  $ is equivariant. \end{enumerate} 
\end{claim}

We prove the claim by induction on $n$, see  appendix~\ref{app.proofs}.\medskip

\noindent These two claims together show that we, in a continuous way, can project the $\alpha$-total elements
onto the $n$-total elements whenever $ 1\leq n \leq \alpha \in \gamma$.

\begin{claim}
Let $ p \in D_{c} $ and  assume that $ \upset{p} \cap \mathcal{D}_{ \alpha }^{R} \neq \emptyset $ for some $\alpha \in \gamma$. 
Then $ p \in \bigcup_{n \in \omega } \Delta_{n} $.
\end{claim}

The proof is by transfinite induction on $ \alpha $, see  appendix~\ref{app.proofs}.\medskip

\noindent This shows that if a $p\in D_c$ has a total extension at some arbitrary  level $\alpha$, then it has a total extension at some finite level $n$.

This implies that $f_{ \alpha , \beta }^{d} $ is an embedding whenever $ \alpha \leq \beta$:
Let $ p_{1} , p_{2} $ be compact elements of $  D^{d}_{ \alpha } $,
and assume that they are consistent in    $  D^{d}_{ \beta } $.
Then there exists $ n \in \omega $ such that $ p_{1}, p_{2} \in \Delta_{n} $.
Moreover, $\Delta_n $ is  by construction closed under binary lubs in $D$, so  $ p_{1} \sqcup_{D} p_{2} \in \Delta_{n} $.
On the other hand, there exists (by the assumption that  $ p_{1} , p_{2}  \in D^{d}_{ \alpha } $)
 some $ x \in \mathcal{D}_{  \beta }^{R} $ with 
$ p_{1}, p_{2} \sqsubseteq x $.
Then $  p_{1} \sqcup_{D} p_{2} = r_{n} (  p_{1} \sqcup_{D} p_{2} ) \sqsubseteq r_{n} (x)  \in  \mathcal{D}_{ n}^{R} $.
This shows that $  p_{1} \sqcup_{D} p_{2} \in D^{d}_{n} \subseteq D^{d}_{ \alpha } $.

Moreover, $ D_{ \gamma }^{d}  $ is the inductive limit of $ \{ D_{ \alpha}^{d}  \}_{ \alpha \in \gamma }  $:
For an arbitrary $ \alpha \geq \omega $ and $ p \in ( D^{d}_{ \alpha })_{c} $, 
there exists $n \in \omega $ such that
$ \upset{p} \cap \mathcal{D}_{  n}^{R} \neq \emptyset $.
This implies that $ ( D^{d}_{ \omega })_{c}  = \bigcup_{ n \in \omega } (  D^{d}_{n} )_{c} $ and that
$ D^{d}_{ \omega } $ is the inductive limit of $ \{  D^{d}_{n} \}_{ n \in \omega } $.
Moreover, for $ \alpha \geq \omega $, it follows that
$ \upset{p} \cap \mathcal{D}_{ \alpha }^{R} \neq \emptyset \Leftrightarrow  \upset{p} \cap \mathcal{D}_{ \omega }^{R} \neq \emptyset $,
which means that  $ D^{d}_{\alpha } = D^{d}_{\omega } $.

Finally, $ f_{\alpha , \beta }^{d} $ is an equiembedding if $ \alpha \leq \beta $, since if
$  x \in  \mathcal{D}_{ \alpha}^{R} $ and $ x' \in D $, then $ x \approx_{ \alpha } x' \Leftrightarrow 
x \approx_{ \beta } x'$.
Moreover, $ \mathcal{D}_{ \gamma }^{d} $ is isomorphic to the inductive limit of $ \{ \mathcal{D}_{ 
\alpha}^{d}  \}_{ \alpha \in \gamma } $, since   
$ x \approx_{\gamma } x' \Leftrightarrow \exists \alpha \in \gamma \; ( x \approx_{\alpha } x' ) $
for $x,x'\in D$. \qed

Thus,  for every strictly positive operation $ \Gamma $ on domain-pers with dense parameters, there exists 
a dense domain-per $  \mathcal{D} $ such that $ \Gamma (  \mathcal{D} )^{d} \cong  \mathcal{D} $. Moreover, 
if $ \funcf :  \cclcdomwp \rightarrow \cclcdomwp $ is the functorial extension of $ \Gamma $, then
$  \mathcal{D} $ is the dense part of the least fixed point of  $ \funcf $.
By abuse of notation, we refer to $\mathcal{D} $ as the {\em dense least fixed point} of $\funcf$.

\begin{rem}
A setback with the dense part construction is that it does not preserve effectivity in general \cite{Dah_07}.
As a consequence, the dense least fixed point of $\funcf $ is not effective just because the parameters of $\funcf $ are.

In fact, this is a returning problem with density. For a simple example of the difficulty of obtaining an effective, dense subset 
in the set of continuous functionals, see Example~4.1 in \cite{Normann09}. 
\end{rem}

\subsection{An admissible least fixed point}
We will now show that the dense least fixed point of a strictly positive functor over $\cclcdomwp$
is admissible if all the parameters involved are dense, admissible.
This will ensure that the resulting $qcb$ space is $T_0$ and that all continuous functions are
representable.

\begin{exa} \label{e_eta}
Let  $ \funcf ( X ) =  \mathcal{A}  +  [ \mathcal{B} \to  X   ] $, with $  \mathcal{A}  $ and 
$  \mathcal{B}  $  some dense, admissible parameters, and let $  \mathcal{D} $ be the dense least fixed point of $ \funcf $.

An $ x \in \mathcal{D}^{R} $ can be represented as a well-founded tree with branching in $ \mathcal{B}^R $ 
and leaf nodes in $ \mathcal{A}^{R} $. A branch $ \{ x_{n} \}_{n \leq N} $ is obtained by iterated 
evaluation of $x $ over a sequence $ \{ b_{n} \}_{n \in \omega } $ over $ \mathcal{B}^{R} $,
i.e. by starting with $ x_{0} = x(b_0)  $ and extending the branch with $ x_{n+1} = x_{n} ( b_{n+1} ) $ while
$ x_{n} \in  [ \mathcal{B} \to  \mathcal{D}   ]^{R}$. 
Ultimately, this process yields an $x_N \in \mathcal{A}^R
$ for some finite $N$.
From this we  can construct an equivariant and equi-injective map \[ \eta :  
\mathcal{D} \rightarrow [ \mathcal{B}^\omega \rightarrow   \mathcal{A}_\bot   \otimes \mathcal{N}   ]. \] \end{exa}

The example gives a rough idea of the method we will use more generally  for a strictly positive $\funcf$ with dense, admissible parameters.
The domain-per  of input sequences, $  \mathcal{B} $ in the example, will be constructed from the non-positive parameters of $\funcf$.
The domain-per of evaluation results, $ \mathcal{A}_\bot $ in the example, will be the disjoint union of all   the
positive parameters of $\funcf$.

We will then show that the dense least fixed point is weakly isomorphic to its dense image under $ \eta $,
and that the dense image is admissible if the function space is admissible. Before we start, we must explain what we
mean by the image of an equivariant map:

\begin{defi} \label{d_image}
If $ \varphi :  \mathcal{D} \rightarrow \mathcal{E} $ is equivariant,  the {\em image} of $ \mathcal{D} $ under 
$ \varphi $ is the domain $  E $ with the partial equivalence relation $  
\approx_{ \varphi  } $ defined by \[ x \approx_{ \varphi  }   y \Leftrightarrow \exists u \in \mathcal{D}^{R}  
\;  ( x \approx_{  \mathcal{E} }  \varphi (u)  \approx_{  \mathcal{E} }  y ) .\]
We denote $(E,\approx_{ \varphi  })$ by  $   \varphi [ \mathcal{D} ]$.
\end{defi}


\begin{lem} \label{l_image}
Let $ \varphi :  \mathcal{D} \rightarrow \mathcal{E} $ be equivariant.
Then  $ \id_E: \varphi [ \mathcal{D} ]  \to \mathcal{E} $ is an equiembedding.
\end{lem}

\proof
 If $x\approx_\varphi y$, then $x\approx_{\mathcal{E}} y $ by definition. 
If $ x\in \varphi [ \mathcal{D} ]^R$ and $x\approx_{ \mathcal{E} } y $, then there is some $u \in  \mathcal{D}^R $
such that $ x\approx_{\mathcal{E}} \varphi(u)$. Hence, $ y\approx_{\mathcal{E}} \varphi(u)$ and $x\approx_\varphi y$.
\qed

In order to show that the dense least fixed point is  weakly isomorphic to its dense image under $  \bar{ \eta} $,
we will define an equivariant lower adjoint $\bar{ \vartheta } $ of $ \funcf $.
The idea is to represent iterated evaluation of elements of the dense least fixed point by the non-positive parameters.
Lemma~\ref{l_eta1} and lemma~\ref{l_eta2} will describe the situation for one-step evaluations. For these results, we
 look at an arbitrary domain-per and not the dense least fixed point.
In lemma~\ref{l_eta4} and lemma~\ref{l_eta5}, we define the maps $ \bar{\eta} $ and $ \bar{ \vartheta} $ using the results
for one-step evaluations.

If $ \funcf : \cclcdomwp \to \cclcdomwp $ is a strictly  positive functor,  we let $ \{ \funcf^{k} \}_{ k \in K_{\funcf} } $ be the set of  atomic subfunctors  of $ \funcf $ with
repetition allowed. Then each $  \funcf^{k} $ represents an occurrence either of  the identity functor or of some 
constant functor.

\begin{lem} \label{l_eta1}
Let $ \funcf : \cclcdomwp \to \cclcdomwp $ be a strictly positive functor with dense non-positive parameters. 
Then there exists  a  dense domain-per $ \mathcal{T}_\funcf $, and
for every convex, local and complete domain-per $\mathcal{D}$, an equivariant and equi-injective map
\[ \eta_\funcf : \funcf (  \mathcal{D} )^{d}  \rightarrow  [   \mathcal{T}_\funcf  \rightarrow  \biguplus_{ 
k \in K_\funcf }  \funcf^{k} ( \mathcal{D}  ) ] .\] 
If the non-positive parameters in $ \funcf $ are admissible, then  $ \mathcal{T}_\funcf  $ is admissible.
\end{lem}

\proof
We fix a convex, local and complete domain-per $ \mathcal{D} $. Independently of $ \mathcal{D} $, we  
define $ \mathcal{T}_{ \funcf }  $ by structural induction on  $ \funcf $. 
Simultaneously, we define a map 
\[ \eta_{ \funcf } :
 \funcf (  \mathcal{D} ) \rightarrow  [   \mathcal{T}_{ \funcf }  \rightarrow  \biguplus_{ k \in K_{ \funcf }}  \funcf^{k} ( \mathcal{D}  ) ] ,\]
which is equivariant and equi-injective, i.e. such that for all $x,y\in D$, \[ x\approx_{\mathcal{D}}y\Leftrightarrow
 \eta_{ \funcf }(x)\approx_{ [   \mathcal{T}_{ \funcf }  \rightarrow  \biguplus_{ k \in K_{ \funcf }}  \funcf^{k} ( \mathcal{D}  ) ]} \eta_{ \funcf }(y) .\]

\begin{enumerate}[$\bullet$]
\item{}
If $ \funcf $ is atomic, let $ \mathcal{T}_{ \funcf } $ be the domain-per with $ T_{ \funcf } = \{ t \} $ as underlying domain
and with $ t \approx_{\mathcal{T}_\funcf} t $. $ \mathcal{T}_{ \funcf } $ is trivially dense.
Observe that $ \biguplus_{ k \in K_{ \funcf }}  \funcf^{k} ( \mathcal{D}  )  = \funcf ( \mathcal{D}  )_{ \bot }  $.
If  $ x \in  \hat{ \funcf } ( D ) $, let
$ \eta_{ \funcf } (x) $ be the map which sends $ t $ to $ (0,x ) \in \hat{ \funcf } ( D )_{\bot }  $. Then
\[ x \approx_{ \funcf (  \mathcal{D} )}   y \Leftrightarrow  (0,x) \approx_{ \funcf ( \mathcal{D})_{ \bot } }   (0, y) 
\Leftrightarrow \eta_{ \funcf } (x)  \approx_{[\mathcal{T}_\funcf \to  \funcf ( \mathcal{D}  )_{ \bot }]}  \eta_{ \funcf } (y) .\] 


\item{}
If $ \funcf =  \funcf_{0}  +  \funcf_{1}  $, let
$ \mathcal{T}_{ \funcf } =    \mathcal{T}_{ \funcf_{0}} \times    \mathcal{T}_{ \funcf_{1}}  $, which is dense by induction.
Note that
\[ \biguplus_{ k \in K_{ \funcf }} \funcf^{k} ( \mathcal{D} ) \cong  
(  \biguplus_{ k \in K_{ \funcf_{0} }} \funcf^{k} ( \mathcal{D} ) ) \oplus
(  \biguplus_{ k \in K_{ \funcf_{1} }} \funcf^{k} ( \mathcal{D} ) ) .\]
We define 
$ \eta_{ \funcf }  $ as the strict map with
$ \eta_{ \funcf } ( i,x) =  \lambda t. \Eval ( \eta_{ \funcf_{i} } ( x ) , t_{i} ) $
for all $ (i ,x) \in  \hat{ \funcf } ( D)  \setminus \{ \bot \} $.
By the induction hypothesis, we have 
\[ (i,x ) \approx_{  \funcf (  \mathcal{D} ) } (j,y) \Leftrightarrow i=j \wedge 
 \eta_{ \funcf _{i}} (x)  \approx_{[\mathcal{T}_{\funcf_i}\to\biguplus_{ k \in K_{ \funcf_{i} }} \funcf^{k} ( \mathcal{D} )]} \eta_{ \funcf_{j}} (y). \]
The right hand side implies directly that $ \eta_{ \funcf } (i, x)  \approx \eta_{ \funcf } (j ,y) $.
Moreover, $ \eta_{ \funcf } (i,x)  \approx \eta_{ \funcf } (j,y) $ is possible only if $ i=j$, since
$ \eta_{ \funcf } (i,x)  $ and $ \eta_{ \funcf } (j,y) $ then map
every $ t \in  \mathcal{T}_{ \funcf }^{R} $ into the same $  \funcf^{k} ( \mathcal{D} ) $ for some $k\in K_\funcf$, the 
disjoint union of $K_{\funcf_0}$ and $K_{\funcf_1}$. 
If $ \neg ( (i,x) \approx_{\funcf(\mathcal{D})} (j,y) )$, we must show that $ \neg ( \eta_{ \funcf } (i,x)  \approx \eta_{ \funcf } (j,y) ) $,
and there are   two cases to consider:
\begin{enumerate}[$-$]
\item If $ i \neq j $, it follows from the observations above.
\item If  $i = j $, then
$  \neg (\eta_{ \funcf _{i}} (x)  \approx \eta_{ \funcf _{i}} (y) ) $ and there
exist $s,t\in T_{\funcf_i}$ with $ s \approx_{ \mathcal{T}_{ \funcf_{i} } } t $ such that
$  \neg ( \Eval ( \eta_{ \funcf_{i} }  (x), s)  \approx \Eval ( \eta_{ \funcf }  (y) , t)  ) $.
By the density of $  \mathcal{T}_{ \funcf } $, there exist $s',t'\in T_\funcf$ such that $ s'_{i} = s $, $ t'_{i} = t $ 
and $  s' \approx_{ \mathcal{T}_{ \funcf } } t' $, and  
\[  \neg ( \Eval ( \eta_{ \funcf }  (i,x), s')  \approx_{( \biguplus_{ k \in K_{ \funcf }} \funcf^{k} ( \mathcal{D} ))} 
\Eval ( \eta_{ \funcf }  (i,y), t')  ) .\]
\end{enumerate}


 \item{}
If $ \funcf =  \funcf_{0} \times  \funcf_{1}  $, let
$ \mathcal{T}_{ \funcf } =    \mathcal{T}_{ \funcf_{0}}  +   \mathcal{T}_{ \funcf_{1}}  $, which is dense by induction. Once again, we have
\[ \biguplus_{ k \in K_{ \funcf }} \funcf^{k} ( \mathcal{D} ) \cong  
(  \biguplus_{ k \in K_{ \funcf_{0} }} \funcf^{k} ( \mathcal{D} ) ) \oplus
(  \biguplus_{ k \in K_{ \funcf_{1} }} \funcf^{k} ( \mathcal{D} ) ) .\]
We define $  \eta_{ \funcf } (x ) $ as the strict map such that
$ \Eval ( \eta_{ \funcf } (x ) , ( i, t ) ) = \Eval ( \eta_{ \funcf_{i} } ( x_{i} ) , t ) $. 
Again by the induction hypothesis,
$ x \approx_{\funcf( \mathcal{D}} y $ if and only if 
$ \eta_{ \funcf_{0} } ( x_{0} ) \approx  \eta_{ \funcf_{0} } ( y_{0} ) $ and 
$ \eta_{ \funcf_{1} } ( x_{1} ) \approx  \eta_{ \funcf_{1} } ( y_{1} ) $,
and this is  equivalent to 
$  \eta_{ \funcf } (x ) \approx  \eta_{ \funcf } (y ) $.


\item{}
If $ \funcf = [ \mathcal{B} \rightarrow  \funcf_{1} ] $, 
where $ \mathcal{B} $ is a non-positive parameter,
let $ \mathcal{T}_{ \funcf } =    \mathcal{B}  \times    \mathcal{T}_{ \funcf_{1}}  $, which is dense by induction since $\mathcal{B}$
is assumed to be dense. 
Then $ \funcf $ has the same atomic subfunctors as $ \funcf_{1} $, so
$  K_{ \funcf } =  K_{ \funcf_{1} } $.

For every  $ x \in [ B \rightarrow  \hat{ \funcf_{1} } (D) ]$,
we define $ \eta_{ \funcf }  ( x ) $  by
$ \eta_{ \funcf }  ( x) 
= \lambda t . \Eval (  \eta_{ \funcf_{1} } ( x (t_{0}) ), t_{1} )$.
By the induction hypothesis,
$ x \approx_{  \funcf ( \mathcal{D} ) } y $ if and only if
$ \eta_{ \funcf_{1} } ( x (b)  ) \approx  \eta_{ \funcf_{1} } ( y (c) ) $
whenever $ b \approx_{  \mathcal{B} } c $,
or equivalently if
\[ b \approx_{  \mathcal{B} } c \Rightarrow 
( s \approx_{  \mathcal{T}_{ \funcf_1 }}  t \Rightarrow 
 \Eval (  \eta_{ \funcf_{1} } ( x ( b )) , s ) \approx  \Eval (  \eta_{ \funcf_{1} } ( y ( c )), t )  ) .\]
A simple paraphrasing is
\[ (b ,s ) \approx_{  \mathcal{T_{ \funcf }} }  (c,t ) \Rightarrow 
 \Eval (  \eta_{ \funcf_{1} } ( x ( b )) , s ) \approx  \Eval (  \eta_{ \funcf_{1} } ( y ( c )), t )  ) ,\]
which holds if and only if $  \eta_{ \funcf } (x ) \approx  \eta_{ \funcf } (y ) $.

\end{enumerate}
It is a trivial inductive verification that $\mathcal{T}_\funcf$ is  admissible if every  
 non-positive parameter in $\funcf$ is admissible, in the last induction step because $\mathcal{B}$ is admissible.
\qed

Note that the assumption that the domain-pers were convex, local and complete was of no importance in this proof.
However, this restriction must be made in what follows, so we included it above for the sake of consistency in the 
presentation.
As our next step, we will now define the lower adjoint of $\eta_\funcf$.

\begin{lem} \label{l_eta2}
Let  $ \funcf : \cclcdomwp \to \cclcdomwp $ be a strictly positive functor with dense non-positive parameters, and let $ \mathcal{D}  $   
be a convex, local and complete domain-per.
Let $ \eta_\funcf : \funcf (  \mathcal{D} )^{d}  \rightarrow  [   \mathcal{T}_\funcf  \rightarrow  \biguplus_{ 
k \in K_\funcf }  \funcf^{k} ( \mathcal{D}  ) ]   $ be defined as in lemma~\ref{l_eta1}.
Let $ D^{ \funcf } $ be the underlying domain of $  \funcf  ( \mathcal{D} )^{d} $ and let $ E^{ \funcf } $ be the underlying domain of $  \eta_\funcf [  \funcf  ( \mathcal{D} ) ]^{d} $.

Then there exists a continuous map $ \vartheta_\funcf :   E^{ \funcf } \rightarrow D^{ \funcf }  $ which is the lower adjoint of
$ \eta_\funcf $ and such that  for every $ x \in \funcf ( \mathcal{D} ) ^{R} $ and  $ q \in E^{ \funcf }_{c} $, 
\[  q \prec_{ \eta_\funcf [  \funcf  ( \mathcal{D} ) ]^{d}} [ \eta_\funcf (x) ]  \Rightarrow \vartheta_\funcf (q) \prec_{ \funcf  ( \mathcal{D} )} [x] .\]
\end{lem}

\proof
First, observe that $ \biguplus_{ k \in K_{ \funcf }}  \funcf^{k} ( \mathcal{D}  ) $ is  convex, local and complete 
since  $ \funcf ( \mathcal{D} ) $  and all the positive parameters are, and recall that $ \mathcal{T}_\funcf $ is dense.
The domain-per $  \eta_\funcf [  \funcf  ( \mathcal{D} ) ]^{d} $ is defined as the dense part of the  image of $ \eta_\funcf$, 
which means that a  
$ \bigsqcup_{ j \in J } [ p^{j} ; q^{j} ] \in   [  T_\funcf  \rightarrow \biguplus_{ k \in K_{ \funcf }} \hat{ \funcf^{k} }( D ) ]_{c} $
is an element of $   E^{ \funcf }_{c} $ if and only if  there exists some $ x \in  \funcf ( \mathcal{D} )^{R} $ such that 
$ \bigsqcup_{ j \in J } [ p^{j} ; q^{j} ]  \prec_{ [   \mathcal{T}_\funcf  \rightarrow  \biguplus_{ 
k \in K_\funcf }  \funcf^{k} ( \mathcal{D}  ) ] } [ \eta_\funcf (x) ] $. Moreover,
 by lemma~\ref{l_lwc3},  this is  equivalent to
\[ \forall t \in   \mathcal{T}_{\funcf}^{R} \:
( \bigsqcup_{ j \in J} \{ q^{j} : p^{j} \sqsubseteq t \} \prec_{ (\biguplus_{ 
k \in K }  \funcf^{k} ( \mathcal{D}  )) } [  \Eval ( \eta_\funcf (x) ,t )  ] )  \]
for the same choice of $ x \in  \funcf ( \mathcal{D})^{R} $. In this case, we say that
 $ \bigsqcup_{ j \in J } [ p^{j} ; q^{j} ] \in  E^{ \funcf }_{c} $ is {\em witnessed} by $x$.

Following the inductive definition of $\eta_\funcf $ given in the proof of lemma~\ref{l_eta1},
we define by structural induction on  $ \funcf $ a monotone map
$ \vartheta_{ \funcf} :  E_{c}^{ \funcf } \to  \hat{ \funcf } ( D )_{c} $ which satisfies
\begin{enumerate}[(1)]
\item
if $  \bigsqcup_{ j \in J } [ p^{j} ; q^{j} ]  \in   E_{c}^{ \funcf }  $ is
witnessed by  $ x \in  \funcf ( \mathcal{D})^{R} $,
then  $ \vartheta_{ \funcf} (   \bigsqcup_{ j \in J } [ p^{j} ; q^{j} ] ) \prec_{ \funcf  ( \mathcal{D} )} [x] $ which implies
 $ \vartheta_{ \funcf} (   \bigsqcup_{ j \in J } [ p^{j} ; q^{j} ] ) \in D_{c}^{ \funcf } $; and
\item
if $  \bigsqcup_{ j \in J } [ p^{j} ; q^{j} ]  \in   E_{c}^{ \funcf }  $
and $ r \in   \funcf  ( D )_{c} $, then 
$  \vartheta_{ \funcf} (  \bigsqcup_{ j \in J } [ p^{j} ; q^{j} ] ) \sqsubseteq r \Leftrightarrow   \bigsqcup_{ j \in J } [ p^{j} ; q^{j} ] \sqsubseteq \eta_{ \funcf} (r)$. 
\end{enumerate}
This yields a monotone map $ \vartheta_\funcf :   E^{ \funcf }_{c} \rightarrow D^{ \funcf }   $ satisfying 
$ \vartheta_\funcf (q ) \sqsubseteq r \Leftrightarrow q \sqsubseteq \eta_\funcf (r) )$ for all $q \in E^{ \funcf }_{c}$  and  
$r \in D^{ \funcf }_{c} $. The continuous extension to $  E^{ \funcf }   $   is then the  lower adjoint of $ \eta_\funcf $
with $ q \prec_{ \eta_\funcf [  \funcf  ( \mathcal{D} ) ]^{d}} [ \eta_\funcf (x) ]  \Rightarrow \vartheta_\funcf (q) \prec_{ \funcf  ( \mathcal{D} )} [x] $
 for all $ x \in \funcf ( \mathcal{D} ) ^{R} $ and  $ q \in E^{ \funcf }_{c} $.

Since $ \vartheta_\funcf $ by necessity is strict, we consider only those 
$ \bigsqcup_{ j \in J } [ p^{j} ; q^{j} ]  \in  E_{c}^{ \funcf } $
for which $ J $ is non-empty and $  q^{j} \neq \bot $ for every $ j \in J $.

\begin{enumerate}[$\bullet$]
\item{$ \funcf  $  atomic:\ } 
If  $ \bigsqcup_{ j \in J } [ t  ; (0 , q^{j} ) ]  \in  E_{c}^{ \funcf } $, let
$ \vartheta_{ \funcf} ( \bigsqcup_{j  \in J } [ t  ; (0 , q^{j} ) ]  ) := \bigsqcup_{j \in J } q^{j} $.

If  $ x \in  \funcf ( \mathcal{D}  )^{R} $ is a witness, 
then $  \Eval ( \eta_{ \funcf } (x) , t ) = (0,x ) $ and $ (0, q^{j} ) \prec_{(\funcf(\mathcal{D})_\bot)} [ (0, x) ] $.
Since $ \funcf ( \mathcal{D}  ) $ is local, it follows that $ \{ q^{j} \}_{j \in J} $ is consistent
with $ \bigsqcup_{j \in J } q^{j} \prec_{\funcf(\mathcal{D})} [x] $. Then
 $ \vartheta_{ \funcf} $ is  clearly well-defined and monotone, 
and if $ r \in \hat{\funcf } (D)_{c} $, then
\begin{eqnarray*}
 \vartheta_{ \funcf } (  \bigsqcup_{j  \in J }  [ t  ; (0 , q^{j} ) ]  ) \sqsubseteq  r &
\Leftrightarrow &
\bigsqcup_{j \in J } q^{j} \sqsubseteq r \\
& \Leftrightarrow & 
\forall j \in J \: ( q^{j}  \sqsubseteq r) \\
& \Leftrightarrow & 
\forall j \in J \: ( ( 0 ,q^{j} ) \sqsubseteq  \Eval ( \eta_{ \funcf } (r)  , t ) ) \\
& \Leftrightarrow &   
\bigsqcup_{j  \in J }  [ t  ; (0 , q^{j} ) ]  \sqsubseteq  \eta_{ \funcf} ( r).  
\end{eqnarray*}

\item{$ \funcf = \funcf_{0} +  \funcf_{1} $:\ } 
If $  \bigsqcup_{j  \in J } [ p^{j}; q^{j} ]  \in  E_{c}^{ \funcf } $ and this is witnessed by 
$ (i , x) \in \funcf ( \mathcal{D}  )^{R} $, let
\[ \vartheta_{ \funcf } (  \bigsqcup_{j  \in J } [ p^{j}; q^{j} ]  ) 
:= (i, 
\bigsqcup_{ J' \subseteq J}
 \{ \vartheta_{ \funcf_{i} }  ( \bigsqcup_{j  \in J'} [ p^{j}_{i};  q^{j}  ] ) : \mbox{$\{ p^{j} \}_{ j \in J'} $ consistent}  \} ) . \]
If  $ (i', x' ) $ is another  witness that $  \bigsqcup_{j  \in J } [ p^{j}; q^{j} ]  \in  E_{c}^{ \funcf } $
and $ i \neq i' $, then 
\[ q^{j} \in   ( \biguplus_{ k \in K_{ \funcf_{0} }} \hat{ \funcf^{k} }( D )) \bigcap ( \biguplus_{ k \in K_{ \funcf_{1} }} \hat{ \funcf^{k} }( D )) 
= \{ \bot \} ,\]
for any $j\in J$. This contradicts the assumption that $ q^{j} \neq \bot $. 
Thus the index $ i$ is uniquely determined by $  \bigsqcup_{j  \in J } [ p^{j}; q^{j} ]  $.

Choose $ j \in J $ and  $  t \in \upset{ p^{j}_{i}} \cap \mathcal{T}_{ \funcf_{i} }^{R} $.
Both $ \mathcal{T}_{\funcf_{0}}  $ and $ \mathcal{T}_{\funcf_{1}}  $ are dense, so we may choose
 $ t' \in  \upset{ p^{j}} \cap \mathcal{T}_{ \funcf }^{R}  $ such that
$  \Eval ( \eta_{ \funcf_{i} } (x) , t) = \Eval ( \eta_{ \funcf } (i,x) , t' ) $.
Then $ q^j \sqsubseteq  \Eval ( \eta_{ \funcf } (i,x) , t' ) $, and this shows that
 $ q^{j}   \prec_{\funcf(\mathcal{D})} [ \Eval ( \eta_{ \funcf_{i} } (x) , t) ]  $.
Hence, $  \bigsqcup_{j  \in J } [ p^{j}_{i}; q^{j} ]  \in  E_{c}^{ \funcf_{i} } $,
and this is witnessed by $x \in  \funcf_i ( \mathcal{D}  )^{R}$.

If $  J'  \subseteq J $, then
$  \vartheta_{ \funcf_{i} } ( \bigsqcup_{j  \in J' } [ p^{j}_{i};  q^{j}  ]  ) $
is bounded by 
$ \vartheta_{ \funcf_{i} }  ( \bigsqcup_{j  \in J }   [ p^{j}_{i};  q^{j}  ] )  $, since
$ \vartheta_{ \funcf_{i} } $ is monotone by the induction hypothesis.
This means that
\[
 \{ \vartheta_{ \funcf_{i} }  ( \bigsqcup_{j  \in J' } [ p^{j}_{i};  q^{j}  ] ) : \mbox{$\{ p^{j} \}_{ j \in J'} $ consistent and $J' \subseteq J $}   \} \] 
is a consistent set and has a least upper bound.

Let $  \bigsqcup_{j  \in J } [ p^{j}; q^{j} ] ,  \bigsqcup_{k  \in K} [ p^{k}; q^{k} ]  \in  E_{c}^{ \funcf } $
(where we assume that $J$ and $K$ are disjoint finite index sets).
If $  \bigsqcup_{j  \in J } [ p^{j}; q^{j} ]  \sqsubseteq  \bigsqcup_{k  \in K} [ p^{k}; q^{k} ]   $,
we must show that 
\[   \vartheta_{ \funcf }( \bigsqcup_{j  \in J } [ p^{j}; q^{j} ])  \sqsubseteq  
 \vartheta_{ \funcf }( \bigsqcup_{k  \in K} [ p^{k}; q^{k} ] ) . \]
For every $ J' \subseteq J $ with $\{ p^{j} \}_{ j \in J'} $ consistent, we can construct a
$ K' \subseteq K $ with $\{ p^{k} \}_{ k \in K'} $ consistent such that
$  \vartheta_{ \funcf_{i} }  ( \bigsqcup_{j  \in J'} [ p^{j}_{i};  q^{j}  ]  ) \sqsubseteq  
\vartheta_{ \funcf_{i} }  ( \bigsqcup_{k  \in K'} [ p^{k}_{i}; q^{k} ]   ) $:
Let $ K' := \bigcup_{j \in J'} \{ k \in K :  p^{k} \sqsubseteq p^{j}  \} $. Then for every $ j \in J' $, we have
\[q^{j} \sqsubseteq 
\bigsqcup_{ k \in K } \{ q^{k} : p^{k}  \sqsubseteq p^{j}  \} =
\bigsqcup_{ k \in K' } \{ q^{k} : p^{k}  \sqsubseteq p^{j}  \} \sqsubseteq
\bigsqcup_{ k \in K' } \{ q^{k} : p^{k}_{i}  \sqsubseteq p^{j}_{i}
\}\ . 
\]
This shows that 
$ \bigsqcup_{j  \in J'} [ p^{j}_{i};  q^{j}  ]   \sqsubseteq  \bigsqcup_{k  \in K'} [ p^{k}_{i}; q^{k} ] $,
and $  \vartheta_{ \funcf_{i} }  $ is monotone by the induction hypothesis.
Hence,
\[   \vartheta_{ \funcf }( \bigsqcup_{j  \in J } [ p^{j}; q^{j} ])  \sqsubseteq  
\bigsqcup_{ J' \subseteq J}
 \{ \vartheta_{ \funcf_{i} }  ( \bigsqcup_{k  \in K'} [ p^{k}_{i};  q^{k}  ] ) : \mbox{$\{ p^{j} \}_{ j \in J'} $ consistent}  \} ) 
\sqsubseteq
 \vartheta_{ \funcf }( \bigsqcup_{k  \in K} [ p^{k}; q^{k} ] ) . \]

\begin{claim}
Let $ p \in E_{c}^{ \funcf }  $ and $ r \in  \hat{ \funcf } (D)_{c} $.
Then $   \vartheta_{ \funcf } (p) \sqsubseteq r \Leftrightarrow p \sqsubseteq  \eta_{ \funcf } (r) $.
\end{claim}

This shows that $(\vartheta_\funcf, \eta_\funcf)$ is an adjunction pair. In the proof, we use the induction hypothesis 
that $(\vartheta_{\funcf_i}, \eta_{\funcf_i})$ is an adjunction pair for $i=0,1$.
See appendix~\ref{app.proofs} for the full details.

Let $ \bigsqcup_{j  \in J } [ p^{j}; q^{j} ] \in E_{c}^{ \funcf } $ be witnessed by $(i,x)$. If
$ J' \subseteq J $ with $\{ p^{j} \}_{ j \in J'} $ consistent, then $ \bigsqcup_{j  \in J' } [ p^{j}_{i}; q^{j} ] \in E_{c}^{ \funcf_i } $
is witnessed by $x$. In this case, 
$ \vartheta_{ \funcf_{i} }  ( \bigsqcup_{ j \in J' } [ p^{j}_{i}; q^{j} ] )   \prec_{\funcf_i ( \mathcal{D} )} [ x] $
by the induction hypothesis. This implies that
$ \vartheta_{ \funcf } ( \bigsqcup_{j  \in J } [ p^{j}; q^{j} ] ) \prec_{\funcf ( \mathcal{D} )} [ (i,x) ] $,
since $ \funcf_{i} ( \mathcal{D} ) $ is local.

\item{ $ \funcf = \funcf_{0} \times \funcf_{1} $:\ }
Let $  \bigsqcup_{j  \in J } [ ( i_j, p^{j} ); q^{j} ] \in E_{c}^{ \funcf }  $. 
Then there are  complementary subsets $J_0$ and $J_1$  of $J $ such that
\[  \bigsqcup_{j  \in J } [ ( i_j, p^{j} ); q^{j} ]  =
 \bigsqcup_{j  \in J_{0} } [ ( 0, p^{j} ); q^{j} ]  \sqcup  \bigsqcup_{j  \in J_{1} } [ ( 1, p^{j} ); q^{j} ]  .\]
This decomposition is unique, because $ \eta_{ \funcf } (x) $ is a strict map for every $ x \in   \funcf ( \mathcal{D} ) $.

We define $  \vartheta_{ \funcf } $ on $ E_{c}^{ \funcf }  $ by
\[ \vartheta_{ \funcf } (   \bigsqcup_{j  \in J } [ ( i_j , p^{j} ); q^{j} ]  )
:= ( \vartheta_{ \funcf_{0} } (  \bigsqcup_{j  \in J_{0} } [ p^{j} ; q^{j} ] ) ,
     \vartheta_{ \funcf_{1} } (  \bigsqcup_{j  \in J_{1} } [ p^{j} ; q^{j} ] )    ) .\]
By induction, $  \vartheta_{ \funcf } $  is clearly well-defined and monotone.

If $ r \in   \funcf ( \mathcal{D} )_{c} $, then it is easily verified that
$  \bigsqcup_{j  \in J } [ ( i_j, p^{j} ); q^{j} ] \sqsubseteq r  $ if and only if 
$ \bigsqcup_{ j \in J_{i} } \{ q^{j} : p^{j} \sqsubseteq t \} \sqsubseteq \Eval ( \eta_{ \funcf_{i} } (r_{i}) , t ) $
for all $ t \in T_{ \funcf_{i} } $ for  $ i=0 $ and  for $ i=1$.
It follows, from the induction hypothesis, that 
$ ( \vartheta_{ \funcf } ,  \eta_{ \funcf } ) $ is an adjunction pair.

Finally,  assume that $  \bigsqcup_{j  \in J } [ ( i_j, p^{j} ); q^{j} ] \in E_{c}^{ \funcf }  $ is witnessed by 
$ x \in   \funcf ( \mathcal{D} )^{R} $ .
Then, if $i=0,1$, $ j \in J_{i} $ and  $ t \in \upset{ p^{j}} \cap  \mathcal{T}_{ \funcf_{i} }^{R} $, we have
\[ q^{j} \prec_{ \funcf_{i} ( \mathcal{D} )} [ \Eval ( \eta_{ \funcf } (x) , (i,t)  ) ]
=   [ \Eval ( \eta_{ \funcf_{i} } (x_{i} ),  t  ) ], \] which shows that
$ \bigsqcup_{j  \in J_{i} } [ p^{j} ; q^{j} ] \in E_{c}^{ \funcf_{i} }  $ is witnessed by
$ x_{i} \in  \funcf_{i} ( \mathcal{D} )^{R} $. By the induction hypothesis,  there exists
$x'_i \in \hat{\funcf_i}(D) $ with $ x'_{i} \approx_{\funcf_i (\mathcal{D}) } x_{i} $ such that
$ \vartheta_{ \funcf_{i} } (  \bigsqcup_{j  \in J_{i} } [  p^{j} ; q^{j} ] )   \sqsubseteq x'_{i} $.
It follows that $ (x'_{0} , x'_{1} ) \approx_{\funcf (\mathcal{D})} x $ and that
\[  \vartheta_{ \funcf } (  \bigsqcup_{j  \in J_{0} } [ ( 0, p^{j} ); q^{j} ]  \sqcup  \bigsqcup_{j  \in J_{1} } [ ( 1, p^{j} ); q^{j} ]  )
\prec_{  \funcf ( \mathcal{D} )  } [ x]  .\]

\item{ $ \funcf =  [\mathcal{B} \rightarrow \funcf_{1} ]$:\ }
If $ \bigsqcup_{j  \in J } [  p^{j} ;  q^{j}  ]   \in  E_{c}^{  \funcf } $, we let
\[ \vartheta_{ \funcf } ( \bigsqcup_{j  \in J } [ p^{j} ;q^{j}]  ) 
:=  \bigsqcup_{j \in J } 
[ p_{0}^{j};  
\vartheta_{ \funcf_{1} } ( \bigsqcup_{ k \in J }  \{  [ p_{1}^{k};   q^{k} ] :  p_{0}^{k} \sqsubseteq p_{0}^{j}  \} )  ] .\]
Note that in this case $ \bigsqcup_{j  \in J } [  p^{j} ;  q^{j}  ]   \in  [ B \times T_{ \funcf_{1} } 
\rightarrow \biguplus_{ k \in K_{ \funcf } } \hat{ \funcf^{k} } ( D ) ] $, and 
$  \{ [ p_{1}^{j} ; q^{j} ] : p_{0}^{j} \sqsubseteq p  \}_{ j \in J } $
is $ [ T_{ \funcf_{1}} \rightarrow \biguplus_{ k \in K_{ \funcf } } \hat{ \funcf^{k} } ( D ) ]$-consistent
for every $ p \in B_{c}  $.

\begin{claim}
Let $ b \in \mathcal{B}^{R} $ and $p  \in \textrm{approx} (b) $.
Let $ \bigsqcup_{j  \in J } [  p^{j} ;  q^{j}  ]   \in  E_{c}^{  \funcf } $ be witnessed by
$  x \in [ \mathcal{B} \rightarrow \funcf_{1} (  \mathcal{D} ) ] ^{R} $.

Then $ \bigsqcup_{ j \in J } \{ [ p_{1}^{j} ; q^{j} ] : p_{0}^{j} \sqsubseteq p  \}  \in E_{c}^{ \funcf_{1} } $ is
witnessed by $ x(b) $.
\end{claim}

In particular, this means that $ \bigsqcup_{ k \in J } \{ [ p_{1}^{k} ; q^{k} ] : p_{0}^{k} \sqsubseteq p_0^j  \}  \in E_{c}^{ \funcf_{1} } $ for
every $j\in J$.

\begin{claim}
Let $ \bigsqcup_{j  \in J } [  p^{j} ;  q^{j}  ]   \in  E_{c}^{  \funcf } $.

If $ J'$ is a subset of $ J $  such that $ \{   p_{0}^{j} \}_{ j \in J' } $ is consistent in $B$,
then $ \{  \vartheta_{ \funcf_{1} }  ( \bigsqcup_{ k \in J } \{ [ p_{1}^{k} ; q^{k} ] : p_{0}^{k} \sqsubseteq p_{0}^{j}  \}  ) \}_{ j \in J'} $
is consistent in $ \hat{ \funcf_{1} } (D)  $.
\end{claim}

This shows that $\{  [ p_{0}^{j} ;  \vartheta_{ \funcf_{1} }  ( \bigsqcup_{ k \in J } \{ [ p_{1}^{k} ; q^{k} ] : p_{0}^{k} \sqsubseteq p_{0}^{j}\} )  ]  \}_{j \in J} $ is a consistent set of step functions in $ [ B \rightarrow \hat{ \funcf_{1} } (D) ] $.
In  particular, this means that
\[ \bigsqcup_{ j \in J} [ p_{0}^{j} ;  \vartheta_{ \funcf_{1} }  ( \bigsqcup_{ k \in J } \{ [ p_{1}^{k} ; q^{k} ] : p_{0}^{k} \sqsubseteq p_{0}^{j}\} )  ] 
\in [ B \rightarrow \hat{ \funcf_{1} } (D) ] .\]

\begin{claim}
Assume $  \bigsqcup_{j  \in J } [  p^{j} ;  q^{j}  ] \sqsubseteq \bigsqcup_{ k \in K} [ p^{k} ; q^{k} ]  \in E_{c}^{ \funcf }$.
Then
\[ \bigsqcup_{ j \in J} [ p_{0}^{j} ;  \vartheta_{ \funcf_{1} }  ( \bigsqcup_{ l \in J } \{ [ p_{1}^{l} ; q^{l} ] : p_{0}^{l} \sqsubseteq p_{0}^{j} \} ) ]  \sqsubseteq
\bigsqcup_{ k \in K} [ p_{0}^{k} ;  \vartheta_{ \funcf_{1} }  ( \bigsqcup_{ m \in K } \{ [ p_{1}^{m} ; q^{m} ] : p_{0}^{m} \sqsubseteq p_{0}^{k} \} ) ]  . \]
\end{claim}

This shows that $\vartheta_\funcf $ is a well-defined and monotone map.

\begin{claim}
Let  $ \bigsqcup_{j  \in J } [  p^{j} ;q^{j}  ] \in    E_{c}^{  \funcf } $
and let $ r \in D_c^\funcf $. Then
$ \bigsqcup_{j  \in J } [  p^{j} ;q^{j}  ]  \sqsubseteq  \eta_{ \funcf } (r)$ 
if and only if,
for every $j\in J$,
\[ \bigsqcup_{k \in J } \{ [  p_{1}^{k} ;q^{k}  ]  : p_{0}^{k} \sqsubseteq p_{0}^{j} \}
\sqsubseteq \eta_{ \funcf_{1} } (r ( p_{0}^{j}  ) ) . \]
\end{claim}

The  induction hypothesis is that $ ( \vartheta_{ \funcf_1 } ,  \eta_{ \funcf_1 } ) $ is an adjunction pair:
\[ \bigsqcup_{k \in J } \{ [  p_{1}^{k} ;q^{k}  ]  : p_{0}^{k} \sqsubseteq p_{0}^{j} \}  \sqsubseteq \eta_{ \funcf_{1} } (r ( p_{0}^{j}  ) )  
\Leftrightarrow
\vartheta_{ \funcf_{1} } ( \bigsqcup_{k \in J } \{ [  p_{1}^{k} ;q^{k}  ]  : p_{0}^{k} \sqsubseteq p_{0}^{j} \} )
\sqsubseteq  r ( p_{0}^{j}  )   \]
for every $ j \in J $,
which again is equivalent to $ \vartheta_{ \funcf }  ( \bigsqcup_{j  \in J } [  p^{j} ;q^{j}  ]  ) \sqsubseteq  r $
by the definition of  $ \vartheta_{ \funcf }   $.
This shows that $ ( \vartheta_{ \funcf } ,  \eta_{ \funcf } ) $ is an adjunction pair:

Assume now that $ x \in [  \mathcal{B} \rightarrow  \funcf_{1} (  \mathcal{D} ) ]^{R} $ is a witness that $  \bigsqcup_{j  \in J } [  p^{j} ;q^{j}  ]   \in E_{c}^{ \funcf } $.
Let $ b \in  \mathcal{B}^{R} $.
If $ j \in J $ and $ p^{j}_{0} \sqsubseteq  b $, then 
$ x(b) $ is a witness that
$ \bigsqcup_{ k \in J } \{ [ p_{1}^{k} ; q^{k} ] : p_{0}^{k} \sqsubseteq p_{0}^{j}  \}  \in E_{c}^{ \funcf_{1} } $.
By the induction hypothesis, this implies that 
\[ \vartheta_{ \funcf_{1} } ( \bigsqcup_{k \in J } \{ [  p_{1}^{k} ;q^{k}  ]  : p_{0}^{k} \sqsubseteq p_{0}^{j} \} ) \prec [ x(b) ] .\]
Then, since $ \mathcal{B} $ is dense and 
$  \funcf_{1} (  \mathcal{D} ) ) $ is convex, local and complete, 
we have (by lemma~\ref{l_lwc3})
\[ \bigsqcup_{ j \in J }
\{ [ p_{0}^{j};  \vartheta_{ \funcf_{1} } ( \bigsqcup_{ k \in J }  \{  [ p_{1}^{k};   q^{k} ] :  p_{0}^{k} \sqsubseteq p_{0}^{j}  \} )  ] \}
\prec [x] .\eqno{\qEd}\]
\end{enumerate}

\noindent In  example~\ref{e_eta},  the representation of a total element $x$ of the dense least fixed point of $ \funcf $ as 
a well-founded tree used  iterated evaluation of $x$ over some input parameter.
We will now use the adjunction pair $ (\vartheta_\funcf , \eta_\funcf ) $, which represents one-step evaluations over  
the dense least fixed point $  \mathcal{D}   $, to show that the situation of the example extends to the more general case
of a strictly positive functor with dense, admissible parameters.
We do this by means of an adjunction pair $(\bar{\vartheta}, \bar{\eta})$.

\begin{lem} \label{l_eta4}
Let $ \funcf : \cclcdomwp \to \cclcdomwp $ be a strictly positive functor with dense, admissible parameters.  
Let $  \mathcal{D} $ be the dense  least fixed point of $ \funcf $.

Then there exist dense, admissible domain-pers $\mathcal{U}$ and $\mathcal{E} $, and an equivariant and 
equi-injective map $ \bar{\eta } :  \mathcal{D} \rightarrow    [ \mathcal{U} \rightarrow \mathcal{E}  ]  $.
\end{lem}

\proof
Let $ \mathcal{T} $ be the dense, admissible domain-per $\mathcal{T}_\funcf$ as defined in the proof of 
lemma~\ref{l_eta1}. 
The domain-per $ \mathcal{U} $ is defined as follows: 
Let $U$ be the domain   of sequences $ x= \{ x_{m} \}_{ m \in \nat} $ over $ T $, partially ordered by
$ x \sqsubseteq_{ U } y \Leftrightarrow \forall m \in \omega \:  (x_{m} \sqsubseteq_{  T } y_{m}) $.
Let $\approx_\mathcal{U}$ be the partial equivalence relation defined by
$  x \approx_{  \mathcal{U} } y \leftrightarrow \forall m \in \omega ( x_{m} \approx_{  \mathcal{T} } y_{m} ) $.
\begin{claim}
$\mathcal{U} $ is dense and admissible.
\end{claim}

Let $ \mathcal{E} $ be the domain-per $ (  \biguplus_{ n \leq N}  \mathcal{A}_{n}   ) \otimes \mathcal{N} $, 
where  $  \mathcal{A}_{0}, \ldots ,  \mathcal{A}_{N} $ are the positive parameters of $ \funcf $.
Clearly, $ \mathcal{E}  $ is admissible since $  \mathcal{A}_{0} , \dots ,  \mathcal{A}_{N}  $ are admissible.
 It is  simply a matter of convenience
 that we use the strict product in the definition of $ \mathcal{E} $.

In what follows, we will consider $  \biguplus_{ n \leq N}  A_{n} $ as a subdomain of the underlying domain of 
$  \biguplus_{ k \in K_{ \funcf} } \funcf^{k} ( \mathcal{D}  )  $ in the obvious way.
We observe that if $ z,z' \in  \biguplus_{ k \in K_{ \funcf} } \hat{\funcf^{k}} ( D  ) $ and $ \bot \neq z \sqsubseteq z'$, then $z,z'$ are either both in $  \biguplus_{ n \leq N}  A_{n} $ or both in its complement. 

Let $  \eta : \mathcal{D}  \rightarrow^{} [  \mathcal{T}  \rightarrow \biguplus_{ k \in K } \funcf^{k} ( \mathcal{D} ) ]$
be the equivariant and equi-injective map as defined in the proof of lemma~\ref{l_eta1}. This map exists because 
 $ \mathcal{D} \cong \funcf (   \mathcal{D} )^{d} $.

In order to define the map $\bar{\eta}$, we must first describe the evaluation tree which $\eta$
produces from an $x\in D$.
For a fixed $ (x, u) \in D \times U$, we define a (finite or infinite) sequence over $\biguplus_{ k \in K_{ \funcf } } \hat{ \funcf ^{k} } (D)  $
as follows: \begin{enumerate}[$\bullet$]
\item Let $ z^{0}_{(x,u)} :=  \Eval ( \eta ( x )  , u_{0} )  \in \biguplus_{ k \in K_{ \funcf } } \hat{ \funcf ^{k} } (D)  $.
\item If $m \in \omega $ and $ z^{m}_{(x,u)} = ( k^{m}_{ (x,u)} , d^{m}_{ (x,u) } ) \notin   \biguplus_{ n \leq N}  A_{n}$,
let \[ z^{m+1}_{ (x,u)}  :=  \Eval ( \eta ( d^{m}_{(x,u)}  )  , u_{m+1} ) \in \biguplus_{ k \in K_{ \funcf } } \hat{ \funcf ^{k} } (D)  .\]
\item Let $ M_{ (x,u) } :=  \{ m \in \omega :  z^{m}_{ (x,u) } \notin   \biguplus_{ n \leq N}  A_{n}  \} $.
\end{enumerate} 
Note that the definition comes to a halt once $ z^{m}_{(x,u)} \in   \biguplus_{ n \leq N}  A_{n}$.
The sequence $ \{ d^{m}_{ (x,u) }  \}_{ m <  M_{ (x,u) } } $ over $ D$ is the {\em evaluation sequence} of $ (x,u) $.
The sequence $ \{ k^{m}_{ (x,u) }  \}_{ m < M_{ (x,u) }} $ over $ K_{ \funcf } $ is the {\em evaluation path} of $ (x,u) $. If $ M_{ (x,u) }  < \omega $, 
and $ \sigma $ is the finite evaluation path of $ (x,u) $,  we let $ n_{(x,u)} := \langle \sigma \rangle $, with $ \langle \cdot \rangle $
a fixed injective function  from the set of finite sequences over $ K_{ \funcf } $ into $ \mathbb{N} $. 
We say that  $ n_{(x,u)} $ is the {\em code} for the evaluation path.
In the case of a finite $ M_{ (x,u) }  $, we obtain an  {\em evaluation result} $z^{M_{ (x,u) }}_{ (x,u) } \in  \biguplus_{ n \leq N}  A_{n} $.
If $ M_{ (x,u) } = \emptyset$, then the evaluation sequence and evaluation path of $(x,u) $ are both empty.

We define a map $ \zeta : D \times U  \rightarrow E $ as follows:
Let $ \zeta (x,u) = ( z^{M_{ (x,u) } }_{ (x,u) }  , n_{(x,u)} ) \in E $
if  $ M_{ (x,u) } < \omega $ and let $ \zeta(x,u) = \bot $ if $ M_{ (x,u) } = \omega $.

Let $ \bar{ \eta }: = \textrm{curry} ( \zeta ) $.
We will show that $\zeta $ is continuous, equivariant and equi-injective, and as a consequence 
$\bar{\eta}$  will be  well-defined, continuous, equivariant and equi-injective.

If the evaluation sequence of $(x,u)$ is infinite, we get $ \zeta (x,u)=\bot$. Note that, since 
we used the strict product of $\biguplus_{n\leq N} A_n$ and $\mathcal{N}$, we get  $ \zeta (x,u)=\bot$
even when the evaluation sequence is finite with $\bot$ as evaluation result. This is because 
the evaluation path is of no interest if the evaluation result is $ \bot $.

\begin{claim}
Let $(x,u),(x',u')\in D \times U $ and assume  $ (x,u) \sqsubseteq (x',u') $.
Then $ M_{(x,u)} \leq M_{(x',u')} $, and 
$z^{m}_{ (x,u) } \sqsubseteq z^{m}_{ (x',u') }$
for every $ m \leq M_{(x,u) } $ finite.
Moreover, if $ M_{ (x,u) } <  M_{ (x',u') } $, then  
$z^{M_{(x,u)}}_{ (x,u) } = \bot $.
\end{claim}

This follows from the fact mentioned above that comparable non-terminating  elements of  $ \biguplus_{ k \in K_{ \funcf} } \hat{\funcf^{k}} ( D  )  $  are either both in $  \biguplus_{ n \leq N}  A_{n} $ or both in its complement. 
For the full proof, see appendix~\ref{app.proofs}.

This is used below to show that $\zeta$ is a monotone map.
Another consequence is that the evaluation path over $ (x,u) $ is an initial segment of the evaluation path 
over $(x',u') $ when  $ (x,u) \sqsubseteq (x',u') $.

\begin{claim}
Let $ \Delta $ be a non-empty directed subset of $ D \times U $.
For each $ m \in \omega $, let $ \Delta^{m} $ be the subset 
$ \{ (x,u) \in \Delta :  M_{(x,u) }\geq m \} $. 
If $ m \leq  M_{   \sqcup \Delta  }$ is finite, then
\begin{enumerate}[(1)]
\item $ \Delta^{m}  $ is  directed  with $ \bigsqcup  \Delta^{m} =  \bigsqcup \Delta $; and
\item $ \{ z^{m}_{ (x,u)}  : (x,u) \in \Delta^{m}  \} $ is directed with least upper bound $ z^{m}_{ \sqcup \Delta } $.
\end{enumerate}
\end{claim}

\noindent Both parts of the claim is proved by a simultaneous induction on $ m $, see  appendix~\ref{app.proofs}.

If $ M_{   \sqcup \Delta  } < \omega $, then 
$ \Delta^{M_{ \sqcup \Delta } } \neq \emptyset $ and
$ z^{M_{ \sqcup \Delta }  }_{ \sqcup \Delta }  
= \bigsqcup \{  z^{M_{ \sqcup \Delta }  }_{ (x,u) } : 
(x,u)  \in  \Delta^{M_{ \sqcup \Delta } }  \}  $.
Combining this with the previous claim, we observe
that the evaluation path of $ \bigsqcup \Delta $ is identical
to the evaluation path of $  (x,u) $ for all $  (x,u)  \in  
\Delta^{M_{ \sqcup \Delta } } $, i.e. $ n_{ \bigsqcup \Delta } = n_{ (x,u) }  $.

We can now show that $\zeta: D \times U \to E$ is continuous:
\begin{enumerate}[$\bullet$]
\item 
Assume that $ (x,u) \sqsubseteq (x',u') $. Then
either $ M_{ (x,u) } = M_{ (x',u' )} $
or $ z^{M}_{(x,u)}  = \bot $.
If either $ M_{(x,u)} = \omega $ or $ z^{M}_{(x,u)}  = \bot $,
then $ \zeta (x,u) = \bot $.
Otherwise, $ z^{M}_{(x,u)} \sqsubseteq  z^{M}_{(x',u')}  $
and $n_{(x,u)} = n_{ (x',u')} $,
which means that  $ \zeta (x,u)  \sqsubseteq \zeta (x',u') $.
\item
Let $\Delta \subseteq D \times U $ be non-empty, directed.
If $ M_{ \sqcup \Delta } = \omega $, then $ \zeta ( \sqcup \Delta)
= \zeta(x,u) = \bot $ for all $ (x,u) \in \Delta $.
If $ M_{ \sqcup \Delta } < \omega $, then there exists
$ (x,u ) \in \Delta $ such that $ M_{ (x,u) } = M_{ \sqcup \Delta } $
and $ z^{M_{ \sqcup \Delta }  }_{ \sqcup \Delta }  
= \bigsqcup \{  z^{M_{ \sqcup \Delta }  }_{ (x,u) } : 
(x,u)  \in  \Delta^{M_{ \sqcup \Delta } }  \}  $.
\end{enumerate}
This also means that
$ \bar{ \eta } : D \rightarrow [ U \rightarrow E ] $ is well-defined and continuous.

Recall that the dense least fixed point $\mathcal{D}$ is constructed as the inductive limit of a $\gamma$-chain of dense domain-pers
with $D$ as the underlying domain, for some ordinal $\gamma$. If $  x \in  \mathcal{D}^{R} $, we let $\rank(x):=\rank_\gamma (x)$, 
the level of the induction at which $x$ is introduced as a total element.
\begin{claim}
Let $ x \in  \mathcal{D}^{R} $ with $ \rank (x) = \alpha + 1 $. Then
$  \eta (x) \in [  \mathcal{T}  \rightarrow \biguplus_{ k \in K  } \funcf^{k} ( \mathcal{D}_{ \alpha } ) ]^{R} $.
\end{claim}

This shows that the evaluation under $\eta $ of a total element will give a total element of strictly lower rank,
and hence that such an evaluation will proceed in a finite number of steps.

\begin{claim}
Let $(x,u),(x',u') \in D \times U $ and  assume that $ (x,u) \approx_{\mathcal{D}\times 
\mathcal{U}} (x',u') $.
Then $ M_{ (x,u)}  =M_{ (x',u')} < \omega $, and 
$ z^{m}_{ (x,u) } \approx z^{m}_{ (x',u') }$
for every $ m \leq M_{ (x,u)} $.
\end{claim}

An immediate consequence is that the evaluation paths of $ (x,u) $  and $ (x',u') $ are identical and finite.
Hence, $ n_{ (x,u)}  = n_{ (x',u') } $, and the respective evaluation results are equivalent in 
$  \biguplus_{ n \leq N}  \mathcal{A}_{n}  $, i.e. $ z^{M_{ (x,u) }}_{ (x,u) } \approx z^{M_{ (x,u) }}_{ (x',u') }$.
In particular, this shows that $ \zeta $ is equivariant, and therefore also that 
$  \bar{\eta } $ is equivariant.

Finally, we give a direct proof that  $  \bar{\eta } $ is equi-injective since $ \eta $ is:
Choose $ x, x' \in D $ and assume that  $  \bar{\eta } (x) \approx_{ [ \mathcal{U} \rightarrow \mathcal{E} ] }  \bar{\eta } ( x' ) $.
Let $u,u'\in U $ and assume that  $ u \approx_{ \mathcal{U} } u' $. Then $  \Eval ( \bar{ \eta} (x) , u ) \approx_{  \mathcal{E} }  
\Eval ( \bar{ \eta} (x') , u' ) $.
Firstly, this means that $ n_{ (x,u) }  = n_{ (x',u') } $, so the evaluation paths of $ (x,u) $ and $(x',u')$ are identical, 
and  the evaluation sequences are of the same length, say $M$.
Secondly, this means that $ z^{M}_{ (x,u)} \approx z^{M}_{ (x',u')} $, 
and since $ u_{ m} \approx_{\mathcal{T}}  u'_{m} $ for all $ m \leq M $ and 
$  \eta $ is equi-injective, 
we obtain  $ z^{m}_{ (x,u)} \approx z^{m}_{ (x',u')} $ 
by a backwards induction on $m$,
and ultimately $ x \approx_{\mathcal{D}} x' $.
\qed

We will now define a lower adjoint of $\bar{\eta} :  \mathcal{D} \to  \bar{\eta }[ \mathcal{D}]^{d}$,
and show that these two domain-pers are weakly isomorphic.

\begin{lem} \label{l_eta5}
Let $ \funcf : \cclcdomwp \rightarrow \cclcdomwp $ be a strictly positive functor with dense, admissible parameters,
and let $\mathcal{D}$ be a dense least fixed point of $\funcf$. 
Let  the domain-pers $  \mathcal{U}$ and $ \mathcal{E} $ and 
the equivariant and equi-injective map $  \bar{\eta } :  \mathcal{D} \rightarrow    [ \mathcal{U} \rightarrow 
\mathcal{E}  ]  $ be as defined in the proof of lemma~\ref{l_eta4}.

Then $  \mathcal{D}$ and $  \bar{\eta }[ \mathcal{D}]^{d}$ are weakly isomorphic domain-pers.
\end{lem}

\proof
We will define an equivariant map $  \bar{\vartheta } :  \bar{\eta }[ \mathcal{D}]^{d}  \rightarrow   \mathcal{D}  $
such that $  \bar{\vartheta } \circ \bar{ \eta }  \approx_{[\mathcal{D}\to  \mathcal{D}]} \id_D  $ and such that
$ \bar{ \eta } \circ  \bar{\vartheta } \approx_{[ \bar{\eta }[ \mathcal{D}]^{d}\to \bar{\eta }[ \mathcal{D}]^{d}]} \id_F$,
where $F$ is the underlying domain of $ \bar{\eta }[ \mathcal{D}]^{d}$.

By means of the map $\vartheta_\funcf $ defined in  lemma~\ref{l_eta2}, we will first define  a monotone map 
$  \bar{\vartheta }: F_c \to D_c $ such that if $ q \in F_{c} $, $ x \in \mathcal{D}^{R} $ and 
$ q \prec_{ [ \mathcal{U} \rightarrow \mathcal{E}  ]} [ \bar{\eta } (x) ] $, then $   \bar{\vartheta }  (q) \prec_{ \mathcal{D}} [ x ] $.
This map extends uniquely to a continuous map $  \bar{\vartheta }: F \to D $, and we will then show that
$ ( \bar{\vartheta }  ,    \bar{ \eta }   ) $ is an adjunction pair.

The map $  \bar{\vartheta }$ is constructed as follows:
For each  $  \bigsqcup_{ j \in J } [ p^{j}; ( q^{j}, n^{j} ) ] \in F_c \setminus \{ \bot \}$, we will define an {\em evaluation tree} $T$ consisting of
finite, decreasing sequences of non-empty subsets of the index set $J$. This tree depends only on $ \{ (p^{j} , n^{j}) \}_{ j \in J } $
which is a  finite subset  of $ U_{c} \times \mathbb{N} $, since   we  w.l.o.g. assume that $( q^{j}, n^{j} )\neq \bot $ for all $ j \in J $.
In the next step, we decorate all the leaf nodes of $T$ using the finite subset $\{q^j\}_{j\in J} $  of $\biguplus_{n \leq N} A_n$.
Inductively we decorate the rest of the tree with the help of $\vartheta_{\funcf}$.
Ultimately, we decorate the empty node with an element of $D_c$, which we take as $ \bar{\vartheta } ( \bigsqcup_{ j \in J } [ p^{j}; ( q^{j}, n^{j} ) ]
)$.

Let $J$ be some finite index set. By a  finite, decreasing sequence $\varsigma $ of length  $  \vert \varsigma \vert $ over $J$,
 we will mean a non-empty finite list of non-empty sets $  \varsigma_{ \vert \varsigma \vert-1} \subseteq \dots \subseteq \varsigma_1 \subseteq 
\varsigma_0  \subseteq J$. 

For a finite subset $ \{ (p^{j} , n^{j}) \}_{ j \in J } $ of
$ U_{c} \times \mathbb{N} $, let $ T ( \{ (p^{j} , n^{j}) \}_{j \in J} )  $ be the set 
of finite, decreasing sequences $ \varsigma $ over $J $ such that
\begin{enumerate}[$\bullet$]
\item $ \{ p^{j}_m \}_{ j \in \varsigma_{m} } $ is consistent in $T_\funcf $ for every $ m <  \vert \varsigma \vert $; and
\item
there exist natural numbers $ n_{ \varsigma } $ and $ M_{ \varsigma }   $ such that
\begin{enumerate}[$-$]
\item $ n^{j} =  n_{ \varsigma }  $ for each $  j \in \varsigma_{0} $; and  
\item $ n_{ \varsigma } $ codes  a finite evaluation path of length $ M_{ \varsigma }  \geq  \vert \varsigma \vert-1  $.
\end{enumerate}
\end{enumerate}
Ordered by sequence extension, and with the empty sequence $e$ added as a root, $ T ( \{ 
(p^{j} , n^{j}) \}_{j \in J} )  $ is a tree.

The number  $ n_{ \varsigma }$ is determined by $ \varsigma_{0} $, the first entry of the sequence $ \varsigma $,
and $M_{ \varsigma } $ depends only on  $ n_{ \varsigma }$.
This means that  if  $ \varsigma \subseteq \tau $, then $ n_{ \varsigma } = n_{ \tau} $ and $M_{ \varsigma } = M_{ \tau } $.
By definition of the tree, we have $  \vert \varsigma \vert  \leq M_{  \varsigma } +1 $ for every $\varsigma$.
If $  \vert \varsigma \vert  =  M_{  \varsigma } +1 $, then no further extension is possible.
If  $  \vert \varsigma \vert  <  M_{  \varsigma } +1 $, then $ \varsigma $ has a trivial extension.
Hence,  a sequence $ \varsigma $ is  maximal if and only if $  \vert \varsigma \vert  =  M_{  \varsigma } +1 $.
Moreover, $ M_{ \varsigma }+1 $ is the upper bound on the length of an extension of  $ \varsigma $.
The tree is finite, since $J$ is finite and there is a finite number of $M_\varsigma $ to choose from.

For each $ \varsigma \in T ( \{ (p^{j} , n^{j}) \}_{j \in J} )  $, we define $ p^{ \varsigma}\in U_c $ as follows:
\begin{enumerate}[$\bullet$]
\item $ p^{ \varsigma}_{m} := \bigsqcup_{ j \in \varsigma_{m} }  p^{j}_{m}$ for every $ m <\vert \varsigma \vert $; and
\item $ p^{ \varsigma}_{m} := \bot $ for $ m \geq  \vert \varsigma \vert $.
\end{enumerate}

\begin{claim}
Let 
$  \bigsqcup_{ j \in J } [ p^{j}; ( q^{j}, n^{j} ) ] \sqsubseteq  \bigsqcup_{ k \in K } [ p^{k}; ( q^{k}, n^{k} ) ] $
(with $J \cap K = \emptyset$).

Then there exists a function
$ f: T ( \{ (p^{j} , n^{j}) \}_{j \in J} )  \rightarrow T ( \{ (p^{k} , n^{k}) \}_{k \in K} )  $ such that 
\begin{enumerate}[$\bullet$]
\item{}  $ p^{f( \varsigma)} \sqsubseteq p^{ \varsigma} $;
\item{}
$  \varsigma  \subseteq  \tau \Rightarrow  f( \varsigma )  \subseteq  f( \tau  )$;
\item{}
$  \vert f (  \varsigma ) \vert = \vert \varsigma \vert $; and
\item{} 
$ \varsigma $ is maximal in $ T ( \{ (p^{j} , n^{j}) \}_{j \in J })  $ 
$\Leftrightarrow $
$ f ( \varsigma ) $ is maximal in $ T ( \{ (p^{k} , n^{k}) \}_{k \in K} )  $.
\end{enumerate}
\end{claim}
We define the sequence $ f( \varsigma ) $ inductively, starting with the first entry. For the details, see  appendix~\ref{app.proofs}.
This shows that if  $ \bigsqcup_{ j \in J } [ p^{j}, ( q^{j}, n^{j} ) ]   \in  [ U \rightarrow E]_{c} $, then the
 evaluation tree $  T ( \{ (p^{j} , n^{j} ) \}_{j \in J } )  $ is uniquely defined up to isomorphism of trees.

For a given $ \bigsqcup_{ j \in J } [ p^{j}; ( q^{j}, n^{j} ) ]   \in  F_{c} $, 
we will now give a decoration $q^{ \varsigma } $ for each non-empty node $\varsigma$ 
of the evaluation tree, starting with the leaf nodes. For this purpose, we fix some $  x \in \mathcal{D}^{R} $ such that
$ \bigsqcup_{ j \in J } [ p^{j} ; ( q^{j}, n^{j} ) ]  \prec_{[\mathcal{U}\to\mathcal{E}]} [ \bar{ \eta } (x) ] $.

\begin{claim}
Let $ \varsigma $ be maximal 
with $ \vert \varsigma \vert  = M+1  $, and let $ j \in \varsigma_{ M } $.
If $ u \in \upset{ p^{ \varsigma } } \cap \mathcal{U}^{R}  $, 
then $ q^{ j } \prec_{( \biguplus_{ n \leq N}  \mathcal{A}_{n} )  } [ z^{ M }_{(x,u)} ]$, where
$z^{ M }_{(x,u)}$ is the evaluation result of $(x,u)$.
\end{claim}

This shows that  $ \{ q^{j} : j \in \varsigma_{M_\varsigma} \} $ is consistent for a maximal $\varsigma$, since
$ \biguplus_{ n \leq N} \mathcal{A}_{n} $ is a local domain-per.
We let 
\[ q^{ \varsigma } := \bigsqcup_{ j \in \varsigma_{ M_\varsigma}}  q^{j} \in   \biguplus_{ n \leq N}  
A_{n} \subseteq  \  \biguplus_{k \in K_\funcf} \hat{\funcf^k} (D).\]

For a non-maximal $ \varsigma $, let $ S ( \varsigma ) $ be the set of immediate successors of $ \varsigma $
in the evaluation tree.
\begin{claim}
Let  $ \varsigma $ be non-maximal and non-empty,
and assume that
\[ \forall \tau \in S( \varsigma ) \; \forall v \in   \mathcal{U}^{R}  \; 
(  p^{ \tau } \sqsubseteq  v   \Rightarrow q^{ \tau } \prec_{ (\biguplus_{k \in K_\funcf}  \funcf^k (\mathcal{D})) } 
[ z^{\vert \tau \vert -1}_{(x,v)} ]) .\]
Then $  \{ [p_{\vert \varsigma \vert}^{ \tau } ; q^{ \tau } ] : \tau \in S ( \varsigma ) \}  $ is consistent and 
if  $ u \in \upset{ p^{ \varsigma } } \cap  \mathcal{U}^{R} $, then 
\[ \bigsqcup \{ [p_{\vert \varsigma \vert}^{ \tau } ; q^{ \tau } ] : \tau \in S ( \varsigma ) \}  
\prec_{\eta_\funcf[ \funcf (\mathcal{D})]^d}
[\eta_\funcf( d^{\vert \varsigma \vert-1}_{(x,u) })].\] 
\end{claim}

This shows that we can apply $\vartheta_\funcf$ on $ \bigsqcup_{ 
\tau \in S ( \varsigma )}  [p^{ \tau }_{ \vert \varsigma \vert} ; q^{ \tau } ] $
if $ \varsigma $ is non-maximal and non-empty and $ q^\tau $ is well-defined for all $ \tau \in S ( \varsigma ) $.
With the additional condition that  $ \varsigma $ is non-empty, let
\[ q^{ \varsigma } := (k^{ \vert \varsigma \vert -1}_{\varsigma } , \vartheta_\funcf ( \bigsqcup_{ 
\tau \in S ( \varsigma )}  [p^{ \tau }_{ \vert \varsigma \vert} ; q^{ \tau } ] ) ) \in \biguplus_{k \in K_\funcf} \hat{\funcf^k} (D),\]
where  $ \{ k^{m}_{ \varsigma } \}_{ m <  M_{ \varsigma} } $ is the evaluation path coded by $ n_{ \varsigma } $.

We now have a decoration $ q^{ \varsigma } \in  \biguplus_{k \in K_\funcf} \hat{\funcf^k} (D) $ for all 
non-empty $\varsigma$, with the additional property that 
$q^{ \varsigma }  \prec_{ (\biguplus_{k \in K_\funcf}  \funcf^k (\mathcal{D})) } [ z^{\vert \varsigma \vert-1}_{(x,u)} ] $ 
for all $ u \in  \upset{ p^{ \varsigma } } \cap \mathcal{U}^{R}  $.
Inductively, we see that 
\[ \bigsqcup \{ [p_{0}^{ \varsigma  } ; q^{ \varsigma } ] : \varsigma \in S ( e ) \} \prec_{[\mathcal{T}_\funcf \to 
\biguplus_{k \in K_\funcf} \funcf^k ( \mathcal{D} )] } [\eta_\funcf (x)] .\] 
We can  now define:
\[ \bar{ \vartheta } (  \bigsqcup_{ j \in J } [ p^{j}, ( q^{j}, n^{j} ) ]   )
:= \vartheta_\funcf ( \bigsqcup \{ [p_{0}^{ \varsigma  } ; q^{ \varsigma } ] : \varsigma \in S ( e ) \}) \in D_c,\]
where $ e$ is the empty sequence. By lemma~\ref{l_eta2}, we even have 
$  \bar{ \vartheta } (  \bigsqcup_{ j \in J } [ p^{j}, ( q^{j}, n^{j} ) ] ) \prec_{\mathcal{D}} [x] $.

We take $ \bar{ \vartheta } ( \bot_F) := \bot_D$.
We show that $ \bar{ \vartheta } :  F_c \rightarrow D_{c} $ is a well-defined and monotone map by a leaf-to-root
induction on the evaluation tree.

\begin{claim}
Let $  \bigsqcup_{ j \in J } [ p^{j}; ( q^{j}, n^{j} ) ] \sqsubseteq  
\bigsqcup_{ k \in K } [ p^{k}; ( q^{k}, n^{k} ) ] \in  F_{c} $, and 
let $ f: T ( \{ (p^{j} , n^{j}) \}_{j \in J} )  \rightarrow T ( \{ (p^{k} , n^{k}) \}_{k \in K} )  $
be as in the claim above.
If $ \varsigma \in  T ( \{ (p^{j} , n^{j}) \}_{j \in J} ) $ is non-empty, then
$q^\varsigma \sqsubseteq q^{f(\varsigma)} $.
\end{claim}

If $  \bigsqcup_{ j \in J } [ p^{j}; ( q^{j}, n^{j} ) ] \sqsubseteq  
\bigsqcup_{ k \in K } [ p^{k}; ( q^{k}, n^{k} ) ] $, then for each sequence 
 $ \sigma \in T ( \{ (p^{j} , n^{j}) \}_{j \in J} ) $  of length $ 1$, there 
is a sequence $ f( \varsigma )  \in T ( \{ (p^{k} , n^{k}) \}_{k \in K} ) $ of length $1$, with 
$ p^{ f( \varsigma ) } \sqsubseteq p^{ \varsigma } $ and 
$ q^{ \varsigma } \sqsubseteq q^{ f( \varsigma ) } $. This shows that
\[ \bar{ \vartheta } (  \bigsqcup_{ j \in J } [ p^{j}; ( q^{j}, n^{j} ) ]  ) \sqsubseteq 
\bar{ \vartheta } (  \bigsqcup_{ k \in K } [ p^{k}; ( q^{k}, n^{k} ) ]  ) .\]
We have a unique continuous extension $ \bar{\vartheta} : F \to D $. 
We will now show that $ (\bar{\vartheta},\bar{\eta} ) $ is an adjunction pair.

\begin{claim}
Let $ \bigsqcup_{ j \in J } [ p^{j}; ( q^{j}, n^{j} ) ]  \in F_c$ and let $ r \in D_c$.
Then $ \bigsqcup_{ j \in J } [ p^{j}; ( q^{j}, n^{j} ) ]  \sqsubseteq \bar{ \eta } (r)  $ if and only if
$ q^{\varsigma } \sqsubseteq z^{ \vert \varsigma \vert-1}_{ (r, p^{ \varsigma } )}  $ for every non-empty $ \varsigma \in T $.
\end{claim}

The claim shows that  $ \bigsqcup_{ j \in J } [ p^{j}; ( q^{j}, n^{j} ) ]  \sqsubseteq \bar{ \eta } (r)  $ if and only if
$ q^\varsigma \sqsubseteq z^0_{(r,p_0^\varsigma)} = \Eval (\eta_\funcf (r),p_0^\varsigma) $ for each $ \varsigma \in S(e)$.
This is again equivalent to 
\[ \bigsqcup_{\varsigma \in S(e)} [p_0^\varsigma; q^\varsigma ] \sqsubseteq \eta_\funcf (r) , \]
and furthermore to 
$ \vartheta_\funcf ( \bigsqcup_{\varsigma \in S(e)} [p_0^\varsigma; q^\varsigma ] ) \sqsubseteq r $
since $ ( \vartheta_\funcf, \eta_\funcf ) $ is an adjunction pair.
This shows that $ (\bar{\vartheta},\bar{\eta} ) $ is an adjunction pair.

The monotone map $ \bar{ \vartheta } $ extends uniquely to a continuous map on $ \bar{ \eta }[ D]^{d} $, and
$ ( \bar{ \vartheta } , \bar{ \eta } ) $ is an adjunction pair with the required $ \prec $-property.

\begin{claim} 
Let $ y \in \bar{\eta}[ \mathcal{D}]^R $. Then  $  \bar{ \eta } (  \bar{ \vartheta } (y)) \approx_{ [ \mathcal{U} \rightarrow \mathcal{E} ]} y $.
\end{claim}
In this proof, we use the fact that $\mathcal{D} $ is local and complete. See appendix~\ref{app.proofs} for details.
A direct consequence is that if $ \bar{\eta}(x)  \approx_{ [ \mathcal{U} \rightarrow \mathcal{E} ] } y $, then
$ \bar{\eta}(x)  \approx_{ [ \mathcal{U} \rightarrow \mathcal{E} ] }  \bar{ \eta } (  \bar{ \vartheta } (y)) $
which implies $ x \approx_{ \mathcal{D}} \bar{ \vartheta} (y) $, since $\bar{\eta} $ is equi-injective. 
In particular, if $ x \in \mathcal{D}^R $, then $ x \approx_{ \mathcal{D}} \bar{ \vartheta} ( \bar{\eta}(x)) $. 
This also shows that
 $  \bar{ \vartheta }:  \bar{\eta}[ \mathcal{D}]^d \to \mathcal{D} $ is equivariant:
If $ y \approx_{  \bar{\eta}[ \mathcal{D}]} y' $, then there is some
$ x \in \mathcal{D}^R $ such that $ y \approx_{ [ \mathcal{U} \rightarrow \mathcal{E} ]} 
\bar{ \eta } ( x) \approx_{ [ \mathcal{U} \rightarrow \mathcal{E} ] } y' $
and therefore $ \bar{ \vartheta} (y)  \approx_{ \mathcal{D}} 
x \approx_{ \mathcal{D}} \bar{ \vartheta} (y') $.
Hence,   $  ( \bar{ \vartheta }, \bar{\eta} ):  \bar{\eta}[ \mathcal{D}]^d \to \mathcal{D} $ is a weak isomorphism pair. \qed

We can now prove our main result.

\begin{thm} \label{t_dlfpadm}
Let $ \funcf : \cclcdomwp \to \cclcdomwp $  be a  strictly positive functor with dense, admissible  parameters.

Then the  dense  least fixed point of $ \funcf $ is admissible.
\end{thm}

\proof
Let $  \mathcal{D}  $ be  the dense least fixed point of $ \funcf $. 
By lemma~\ref{l_eta4} and lemma~\ref{l_eta5}, we have dense, admissible domain-pers $ \mathcal{U} $
and $ \mathcal{E} $ and an equivariant map $\bar{\eta} : \mathcal{D} \to [\mathcal{U}\to \mathcal{E}] $
such that $ \mathcal{D} $ and $ \bar{\eta} [\mathcal{D}]^d $ are weakly isomorphic.

We have that $  [\mathcal{U}\to \mathcal{E}] $ is admissible since $ \mathcal{U} $, $ \mathcal{E} $ are dense, admissible
by lemma~\ref{l_adm5}.
Then $ id : \bar{\eta} [\mathcal{D}] \to  [\mathcal{U}\to \mathcal{E}] $ is an equiembedding by lemma~\ref{l_image}.
This shows that $  \bar{\eta} [\mathcal{D}] $ is admissible by  observation~\ref{o_admdp5}.
Admissibility is preserved under the dense part construction, see observation~\ref{o_densepartadmissible}, so 
$  \bar{ \eta} [ \mathcal{D} ]^{d}  $ is admissible. 
Finally, since $  \mathcal{D}  $ and  $  \bar{ \eta} [ \mathcal{D} ]^{d}  $  are  weakly isomorphic, we can use  lemma~\ref{l_admdp2}
and conclude that $  \mathcal{D}  $ is admissible.
\qed

\section{Strictly positive induction in $ \cqcb $}
\label{s_qcb}

\noindent We will now use the construction of a dense, admissible least fixed point of a strictly positive
functor $\funcf: \cclcdomwp \to \cclcdomwp$ with dense, admissible parameters
to define $qcb_0$ spaces  by strictly positive induction.

First, we show that the choice of convex, local and complete domain-pers as the objects in the
category $ \cclcdomwp$ was adequate.

\begin{lem} \label{l_lwc1}
Let $X$ be a topological space. Then $X$ is  a  $ qcb_{0} $ space  if and only if
there exists a countably based, dense, admissible, convex, local and complete domain-per 
$   \mathcal{D} $ such that $ \mathcal{Q}   \mathcal{D}  \cong X $.
\end{lem}

\proof
Let $ X $ be a $ qcb_{0} $ space. 
Let $ (D, D^{R} , \delta ) $ be the standard dense and admissible representation of $X$ w.r.t to some
countable pseudobase $\mathsf{P}$, see the proof of theorem~\ref{t_adm1} for details,
and let 
$ \mathcal{D} $ be the associated domain-per.

If $ x \in \mathcal{D}^{R} $, then $ I^{\delta (x )} $ is a greatest representative for $ \delta (x) $, 
so $   I^{\delta (x )} = \bigsqcup [ x]_{ \mathcal{D}  } $.
Hence, $  \mathcal{D} $ is local and complete.

For the convexity, let $I,J,K$ be ideals over $ (\mathsf{P}, \supseteq)$ and 
assume that  $I \subseteq J \subseteq K $ and that $ I, K  \rightarrow_{ \mathsf{P} } x $.
Then $ x \in B $ for all $B \in J  $, since $ J \subseteq K $.
If $ x \in U $ and $U \subseteq X $ is open, then  there exists $ B \in I \subseteq J$ with
$ x \in B \subseteq U $. This shows that $ J \rightarrow_{ \mathsf{P} } x $. 

For the converse,  $ \mathcal{Q}   \mathcal{D}  $ is  a $ qcb$ space by since it is the quotient
space of $\mathcal{D}^R$, a countably based space. It is  a $ qcb_{0} $ space  by corollary~\ref{c_adm7}, since
$  \mathcal{D} $ is an admissible domain-per.
\qed

The  basic operations for $ qcb_{0} $ spaces are identity ($ \id $), disjoint sum ($ \cdot \uplus \cdot$), 
sequential product ($\cdot \times^{s}\cdot $), and $ \cqcb $-exponential ($ [ \cdot \Rightarrow^{s} \cdot] $), as defined in 
section~\ref{s_backg}.
An operation on $ qcb_{0} $ spaces is said to be {\em strictly positive} if it is constructed from a finite 
list  of $ qcb_{0} $ spaces (the positive parameters), using identity, disjoint sum, sequential
product and $ \cqcb $-exponentiation by  a fixed $ qcb_{0} $ space (a non-positive parameter).
The basic operations on $qcb_{0} $ spaces are representable by the corresponding basic operations 
on domain-pers:

\begin{lem} \label{l_qcb1}
Let $  \mathcal{D} $ and $ \mathcal{E} $  be admissible domain-pers. 
Then 
\[ \mathcal{Q}  ( \mathcal{D} +  \mathcal{E} ) \cong  ( \mathcal{Q}   \mathcal{D} ) \uplus  ( \mathcal{Q}   \mathcal{E} )  ; \] 
\[ \mathcal{Q}  ( \mathcal{D} \times  \mathcal{E} ) \cong  ( \mathcal{Q}   \mathcal{D} )  \times^{s}   ( \mathcal{Q}   \mathcal{E} )  . \]
If, in addition, $  \mathcal{D} $ is dense, then
$\mathcal{Q}  [ \mathcal{D} \rightarrow  \mathcal{E} ] \cong  [ \mathcal{Q}   \mathcal{D}  \Rightarrow^{s}  \mathcal{Q}   \mathcal{E}  ]$ .
\end{lem}

\proof
As already observed in lemma~\ref{l_admdp1}, the domain representations induced by
$ \mathcal{D} +  \mathcal{E} $ and $ \mathcal{D} \times  \mathcal{E}  $ 
are admissible domain representations of $ ( \mathcal{Q}   \mathcal{D} ) \uplus  ( \mathcal{Q}   \mathcal{E} ) $ and
$ (  \mathcal{Q}    \mathcal{D} )  \times  ( \mathcal{Q}   \mathcal{E} ) $, respectively.
On condition that  $  \mathcal{D} $ is dense, the domain representation induced by
$ [\mathcal{D} \rightarrow  \mathcal{E}] $ is an admissible domain representation of
 $ [ \mathcal{Q}   \mathcal{D}  \Rightarrow_{ \omega }  \mathcal{Q}   \mathcal{E}  ] $.

This implies that $ \mathcal{Q}  ( \mathcal{D} +  \mathcal{E} )$,  $  \mathcal{Q}  ( \mathcal{D} \times  \mathcal{E} ) $ 
and $ \mathcal{Q}  [ \mathcal{D} \rightarrow  \mathcal{E} ]$
are homeomorphic to the respective sequential closures
$ ( \mathcal{Q}    \mathcal{D} ) \uplus  ( \mathcal{Q}    \mathcal{E} )$,
$ ( \mathcal{Q}   \mathcal{D} )  \times^{s}    (\mathcal{Q}   \mathcal{E} )  $
and
$ [ \mathcal{Q}   \mathcal{D}  \Rightarrow^{s}  \mathcal{Q}   \mathcal{E}  ] $.
\qed

For the domain-pers we took the categorical approach to the problem of definitions by strictly
positive induction, similar to the technique for solving recursive domain equations. We then made a suitable choice of morphisms
which were, in a sense, embeddings. For the $qcb_0$ spaces, there is no obvious choice of embeddings.
We will therefore pass directly from a strictly positive operation on $qcb_0$ spaces to a functorial
representation over $\cclcdomwp$.

If $ \Gamma $ is a strictly positive operation on $ qcb_{0} $ spaces, we obtain a   strictly positive operation 
on domain-pers by replacing
\begin{enumerate}[$\bullet$]
\item  each (positive or non-positive) parameter $P$  in $ \Gamma $ by a countably based, dense, admissible,
convex, local and complete domain-per $ \mathcal{A} $  with  $ \mathcal{Q}   
\mathcal{A}  \cong P  $; and
\item 
each occurrence of one of the basic operations
$ \id $, $ \cdot \uplus \cdot$, $ \cdot \times^{s} \cdot $ or 
$[ \cdot \Rightarrow^{s} \cdot ] $ by the corresponding domain-per operation
$ \id $, $ \cdot + \cdot $, $ \cdot \times \cdot $ or $ [\cdot \to \cdot ] $.
\end{enumerate}
Strictly positive operations are functorial in $ \cclcdomwp $, and in combination with 
lemma~\ref{l_qcb1} this
implies that there exists a functor $ \funcf: \cclcdomwp \to \cclcdomwp $ such that 
for every domain-per $  \mathcal{D} $ and $ qcb$ space $X$, we have
\[  \mathcal{Q}  \mathcal{D}   \cong  X  \Rightarrow \mathcal{Q}  ( \funcf  \mathcal{D}  ) \cong \Gamma X .\]
We refer to $\funcf$  as the {\em functorial representation} of $ \Gamma $ over $ \cclcdomwp$.

\begin{prop} \label{p_qcb5}
Let $ \Gamma $ be a strictly positive operation on $ qcb_{0} $ spaces. 
Let $  \funcf $ be a functorial representation of $ \Gamma $ over $ \cclcdomwp$, 
and let $ \mathcal{D} $ be a  least fixed point of $ \funcf $.
Then $  \mathcal{Q}  \mathcal{D}  $ is a $qcb_0$ space and a fixed point of $ \Gamma $, i.e.
$ \mathcal{Q}  \mathcal{D}   \cong  \Gamma (\mathcal{Q}  \mathcal{D})  $.
\end{prop}

\proof
The functorial representation $\funcf$ of $ \Gamma $ has countably based, dense, admissible parameters by
definition. Then  $ \mathcal{D} $ is  countably based by observation~\ref{o_lfp2}.
The dense part of $ \mathcal{D} $ is a dense least fixed point of $\funcf$ and
therefore admissible by  theorem~\ref{t_dlfpadm}.
Since $  \mathcal{Q}  \mathcal{D}   \cong  \mathcal{Q} ( \mathcal{D}^d )$ by the definition
of dense part and $  \mathcal{Q} ( \mathcal{D}^d )$ is a $qcb_0$ space by 
lemma~\ref{l_lwc1}, this shows that  $  \mathcal{Q}  \mathcal{D} $ is a $qcb_0$ space.

Since  $ \mathcal{D} $ is a  least fixed point of $ \funcf $, we have $ \mathcal{D} \cong 
\funcf (\mathcal{D})$ and  $\mathcal{Q}  \mathcal{D} \cong  \mathcal{Q}  ( \funcf  \mathcal{D}  ) $.
Moreover, $  \mathcal{Q}  ( \funcf  \mathcal{D}  ) \cong \Gamma ( \mathcal{Q}  \mathcal{D} ) $
since $ \funcf$ is a functorial representation of $\Gamma$.
This shows that $ \mathcal{Q}  \mathcal{D}   \cong  \Gamma (\mathcal{Q}  \mathcal{D})  $.
\qed

This shows that  strictly positive operations on $ qcb_{0} $ spaces admit  fixed points.
Put differently, we can construct a solution of the strictly positive 'recursive $qcb_0$
equation' $ X = \Gamma (X) $ by means of  a functorial representation  of $ \Gamma $ over
$\cclcdomwp$. In order to say that this is a definition by strictly positive induction, we
need to show that it is a canonical solution, i.e.  independent of the chosen functorial 
representation of $\Gamma$.

\begin{prop} \label{p_qcb7}
Let $ \Gamma $ be a strictly positive operation on  $ qcb_{0} $ spaces.
If $ \funcf $ and  $ \funcg $ are functorial representations of $ \Gamma $ over $ \cclcdomwp $ 
and $ \mathcal{D} $ and $ \mathcal{E} $ are least fixed points of $ \funcf $ and $ \funcg $, respectively,
then $  \mathcal{Q}  \mathcal{D}  \cong  \mathcal{Q}  \mathcal{E}  $.
\end{prop}

\proof
We say that  strictly positive endofunctors over $ \cclcdomwp $ are {\em weakly equivalent} if we can obtain 
one from the other by replacing each parameter  by a weakly isomorphic domain-per.

\begin{claim}
Let $\funcf $ and $\funcg$ be weakly equivalent strictly positive endofunctors over $ \cclcdomwp $.
Then there exist assignments $ \varphi  \mapsto  \varphi^{ \funcf , \funcg } $ 
and $ \varphi  \mapsto  \varphi^{ \funcg , \funcf } $ 
from the class of equivariant maps into itself with the following properties:
\begin{enumerate}[(1)]
\item if $ \varphi :  \mathcal{D}  \rightarrow   \mathcal{E}  $, 
then $ \varphi^{ \funcf , \funcg } : \funcf (  \mathcal{D} ) \rightarrow  \funcg (  \mathcal{E} ) $ and
 $ \varphi^{ \funcg , \funcf } : \funcg (  \mathcal{D} ) \rightarrow  \funcf (  \mathcal{E} ) $;
\item if $ ( \varphi , \chi ) $ is a weak isomorphism of domain-pers $  \mathcal{D} $ and $   \mathcal{E}  $, 
then $ (  \varphi^{ \funcf , \funcg } ,  \chi^{ \funcg , \funcf } ) $ is a weak isomorphism of $ \funcf (  \mathcal{D} )  $ and $ \funcg (  \mathcal{E}  ) $; and
\item if $ \varphi $, $ \chi$ are equivariant maps and $f ,g $ are equiembeddings which satisfy
$ g \circ \varphi = \chi \circ f $, then
$ \funcg ( g)  \circ \varphi^{ \funcf , \funcg } = \chi^{ \funcf , \funcg } \circ \funcf ( f ) $.
\end{enumerate}
\end{claim}
\noindent The proof is straight-forward. Some more details are given in appendix~\ref{app.proofs}.

Now, let  $ \funcf $ and  $ \funcg $ be functorial representations of $ \Gamma $ over $ \cclcdomwp $. 
Then  $ \funcf $ and  $ \funcg $  are  weakly equivalent:
A parameter  in $ \Gamma $ is represented by a dense, admissible parameter  in $ \funcf $ and  by a dense, admissible parameter
in $\funcg$. By lemma~\ref{l_admdp4},  these parameters are  weakly isomorphic.

If $ \beta $ is an ordinal, let 
$   ( \{ \mathcal{D}_{\alpha}  \}_{ \alpha \in  \beta }, \{ f_{\alpha , \alpha' } \}_{ \alpha \leq \alpha'  \in \beta }) $ 
be  the $ \beta $-chain constructed from $ \funcf $, and let
$   ( \{ \mathcal{E}_{\alpha}  \}_{ \alpha \in  \beta }, \{ g_{\alpha , \alpha' } \}_{ \alpha \leq \alpha'  \in \beta }) $ 
be  the $ \beta $-chain constructed from $ \funcg $.
Choose a limit ordinal $\gamma$  such that $  \mathcal{D}_{ \gamma}  $ and $ \mathcal{E}_{ \gamma}  $ are  least fixed points 
of $ \funcf $ and $ \funcg $, respectively.

\begin{claim}
There are families $\{\varphi_{ \beta }: \mathcal{D}_{\beta} \to \mathcal{E}_{\beta} \}_{\beta\leq\gamma} $ 
and  $\{\chi_{ \beta }: \mathcal{E}_{\beta} \to \mathcal{D}_{\beta} \}_{\beta\leq\gamma} $ of
equivariant maps such that
 each $ ( \varphi_{ \beta } , \chi_{ \beta } ) $ is  a weak isomorphism pair.
\end{claim}

This proof is by transfinite induction on $\beta \leq \gamma $ and make use of the assignments  $ \varphi  \mapsto  \varphi^{ \funcf , \funcg } $ 
and $ \varphi  \mapsto  \varphi^{ \funcg , \funcf } $  defined above.
We also need an extra induction hypothesis and the notion of a uniform mapping ( definition~\ref{d_uniformmapping})
for the induction to go through.
The complete proof is given in appendix~\ref{app.proofs}. 

In particular,  the claim shows that $ ( \varphi_{ \gamma} ,  \chi_{ \gamma } ) $ is a weak isomorphism pair of
$ \mathcal{D}_{ \gamma } $ and $ \mathcal{E}_{ \gamma } $. By lemma~\ref{l_admdp2},
this implies that $  \mathcal{Q}  \mathcal{D}_{ \gamma }  \cong  \mathcal{Q}  \mathcal{E}_{ \gamma }  $.\qed

Now, we are ready to state and prove our main result.

\begin{thm} \label{t_fixedpoint}
Let $\Gamma$ be a strictly positive operation on $qcb_0$ spaces.
Then $\Gamma $ has a  fixed point  $ X \cong \Gamma(X) $.
Moreover, the fixed point is defined via a least fixed point construction in domain theory, and
it is is independent, up to homeomorphism, of the admissible standard domain representations
used for the parameters involved.
\end{thm}

\proof
The fixed point of $\Gamma $ exists by proposition~\ref{p_qcb5}.
It is independent of the choice of functorial representation of $\Gamma$ by proposition~\ref{p_qcb7}.
\qed

Applications of interest in analysis usually concern Hausdorff spaces, so we include the following important result.

\begin{prop} \label{p_qcb8}
If $ X$ is a $qcb_{0} $ space defined by a strictly positive induction in which all positive parameters involved are Hausdorff, then 
$ X$ is Hausdorff.
\end{prop}

\proof
Let $ \Gamma $ be the strictly positive operation used to define $ X$ and let $ A_{0} , \dots , A_{N} $ be the positive parameters involved.

Let $ \mathcal{D} $ be the dense least fixed point of some functorial representation of $ \Gamma $.
Consider the equivariant and equi-injective map
$ \overline{ \eta } :  \mathcal{D} \rightarrow [ \mathcal{U} \rightarrow \mathcal{E} ] $
defined in lemma~\ref{l_eta4}.
Let $\overline{ \eta }^{ \mathcal{Q}} $ be the induced continuous, injective map from $ X$ to 
$  [ \mathcal{Q}  \mathcal{U}  \rightarrow  ( \biguplus_{ i \leq N} A_{i} ) \times \mathbb{N}  ]$,
which is the set of continuous functions from $  \mathcal{Q}  \mathcal{U} $ into $  ( \biguplus_{ i \leq N} A_{i} ) \times \mathbb{N} $ with a topology which is finer than the compact-open topology.

Choose distinct points $ x, y \in X $. Then $  \overline{ \eta }^{ \mathcal{Q}} (x) \neq  \overline{ \eta }^{ \mathcal{Q}} (y) $, so they are evaluated differently for some $ u \in  \mathcal{Q}  \mathcal{U} $, that is $ \Eval (  \overline{ \eta }^{ \mathcal{Q}} (x) , u ) \neq \Eval (  \overline{ \eta }^{ \mathcal{Q}} (y) , u ) $.
Now, because $ ( \biguplus_{ i \leq N} A_{i} ) \times \mathbb{N}  $ is Hausdorff, these evaluation results can be separated by open neighbourhoods $ V_{x} $ and $ V_{y} $.
Then $  \{ f: f(u) \in V_{x}\} $ and  $  \{ f: f(u) \in V_{y}\} $  are  disjoint basic open sets in the compact-open topology, thus separating $  \overline{ \eta }^{ \mathcal{Q}} (x) $ and $  \overline{ \eta }^{ \mathcal{Q}} (y) $.
The inverse images are then disjoint open neighbourhoods separating $ x $ and $y$.
\qed

\begin{rem}
If $ f: \mathcal{D} \to \mathcal{E} $ is an equiembedding, then $f^{\mathcal{Q}} : \mathcal{Q}\mathcal{D} \to \mathcal{Q}\mathcal{E} $ is
an injective map.
Moreover, it has the property that a sequence is mapped to a convergent sequence if and only if it is itself convergent.
Note that these {\em sequential embeddings} are not necessarily embeddings in the topological sense, unless the spaces are countably based.

Unfortunately, there is no obvious way to lift a sequential embedding $f: X \to Y $ of $qcb_0$ spaces to an equiembedding of
representing domain-pers. Therefore, these embeddings are of limited interest.
\end{rem}

\begin{rem}
The fixed point of $\Gamma $ can be constructed as an inductive limit:
Let $ \funcf $ be a functorial representation of $ \Gamma $ over $\cclcdomwp$.
By proposition~\ref{p_dense3}, there exists a $\gamma$-chain 
 $ ( \{ \mathcal{D}_{ \alpha}  \}_{ \alpha \in \gamma } , \{ f_{\alpha , \beta } \}_{ \alpha \leq \beta \in \gamma  } ) $
of dense domain-pers and equiembeddings such that the inductive limit 
$ \mathcal{D}_\gamma$ is a dense least fixed point of $ \funcf$.
The dense least fixed point is admissible by theorem~\ref{t_dlfpadm} and so all the domain-pers are admissible 
by observation~\ref{o_admdp5}. This shows that 
$ ( \{ \mathcal{Q} \mathcal{D}_{ \alpha}  \}_{ \alpha \in \gamma } , \{ f^{ \mathcal{Q}}_{\alpha , \beta } \}_{ \alpha \leq \beta \in \gamma  } ) $
is a directed system of $qcb_0$ spaces and continuous functions.
Using the lifting of all continuous functions to dense, admissible domain representations (lemma~\ref{l_adm3}),
we can show that $ \mathcal{Q} \mathcal{D}_{ \gamma } $ is the inductive limit.

For what it is worth, the continuous functions of the directed system are sequential embeddings as described in the previous remark.
\end{rem}

The fixed point of $\Gamma$ is an example of an inductive limit, possibly uncountable,
of $qcb_0$ spaces. The category $\cqcb$ does not have uncountable inductive limits, so
the existence of the fixed point cannot be proved inductively within the class of $qcb_0$ 
spaces. On the contrary, 
the chain of $qcb_0$ spaces is constructed from the limit and down and not from the bottom and up.
Furthermore, this means that we do not know whether the fixed point is a {\em least} fixed point of some strictly positive
endofunctor over $\cqcb$.
For examples of initial algebras in $\cqcb$, see \cite{Bat_08}.



A natural extension of this work would be to  study positive inductive definitions in general.

\appendix
\section{Proof of claims} \label{app.proofs}

\subsection*{Proposition~\ref{p_lfp2}}

\begin{claimapp}
There exists a family  $  \{ h_{ \beta } :  \mathcal{D}_{ \beta } \rightarrow  \mathcal{E} \}_{ \beta \leq \gamma_0 } $ of equiembeddings such that,
for each $ \beta \in \gamma_0 $,
$ h_{ \beta } = g \circ \funcf (h_{ \beta } ) \circ f_{ \beta , \beta + 1}  = h_{ \beta + 1} \circ f_{ \beta , \beta + 1 } $. 
\end{claimapp}

\proof  The proof is by  transfinite induction on $ \beta $.

Let $ h_{0} $  be the unique equiembedding from $ \mathcal{D}_{0} $  into  $\mathcal{E} $.
Since $ g \circ \funcf (h_{ 0 } ) \circ f_{0,1} $  is an equiembedding  from $ \mathcal{D}_{0} $  
into  $\mathcal{E} $,  we have  $ h_{0} =  g \circ \funcf (h_{ 0 } ) \circ f_{0,1} $.

Assume that $  h_{ \beta } : \mathcal{D}_{ \beta } \to \mathcal{E} $ is an equiembedding satisfying
$  h_{ \beta } =  g \circ \funcf (h_{ \beta } ) \circ f_{ \beta , \beta + 1 } $. Let
$ h_{ \beta +1 } := g \circ \funcf ( h_{ \beta } ) $.
Then $ h_\beta =    h_{ \beta + 1} \circ f_{ \beta , \beta + 1 } $ and
\[ h_{ \beta +1 } =  g \circ \funcf ( h_{ \beta + 1} \circ f_{ \beta , \beta + 1 } )
= g \circ \funcf ( h_{ \beta + 1} ) \circ  \funcf ( f_{ \beta , \beta + 1 } )
= g \circ \funcf ( h_{ \beta + 1} ) \circ   f_{ \beta +1 , \beta + 2 }  . \]

Assume that $ h_{ \alpha } =  h_{ \alpha + 1}  \circ  f_{ \alpha , \alpha + 1 } $ for every $ \alpha \in\beta  $, where 
$ \beta\leq\gamma_0 $ is some limit ordinal.
Let $ h_{ \beta }:  \mathcal{D}_{ \beta } \rightarrow   \mathcal{E} $ be the mediating morphism from
$ ( \mathcal{D}_{\beta }, \{f_{\alpha,\beta}\}_{\alpha\in\beta} )$ to $ ( \mathcal{E} , \{ h_{ \alpha } \}_{ \alpha\in\beta  } ) $,
i.e. the unique 
equiembedding such that $ h_{ \alpha } = h_{ \beta} \circ f_{\alpha, \beta } $ for every $ \alpha\in\beta  $, which exists since
$ ( \mathcal{D}_{\beta }, \{f_{\alpha,\beta}\}_{\alpha\in\beta} )$ is the inductive limit.
Then, for each  $ \alpha\in\beta $, we have
\begin{eqnarray*} 
 h_{ \alpha } 
& = & g \circ \funcf (h_{ \alpha } ) \circ f_{ \alpha , \alpha + 1}  \\
& = & g \circ \funcf (  h_{ \beta} \circ f_{ \alpha , \beta } ) \circ f_{ \alpha , \alpha + 1}  \\
& = & g \circ \funcf (  h_{ \beta} ) \circ  \funcf ( f_{ \alpha , \beta } ) \circ f_{ \alpha , \alpha + 1}  \\
& = & g \circ \funcf (  h_{ \beta} ) \circ f_{ \alpha , \beta +1 } \\
& = & g \circ \funcf (  h_{ \beta} ) \circ f_{ \beta , \beta +1 } \circ f_{ \alpha, \beta }^-
\end{eqnarray*}
Together with the uniqueness of $ h_{ \beta } $, this implies that
$ h_{ \beta } =  g \circ \funcf (  h_{ \beta} ) \circ f_{ \beta , \beta +1 } $.
\qed

\begin{claimapp}
Let $ \omega \leq \beta \leq \gamma_0 $ and assume that $ f_{ \beta , \beta + 1} $ is an isomorphism. 
Then  $  h_{ \beta } : D_{ \beta } \to E $ is the unique 
$\hat{  \funcf  }  $-morphism from 
$  ( D_{ \beta } , f_{ \beta, \beta +1 }^{-} ) $
into  $ (E,g) $.
\end{claimapp}

\proof The proof is by transfinite induction on  $ \beta \leq \gamma_0$.
From domain theory, we know that there exists an initial $ \hat{ \funcf }$-algebra $ (D, f) $.
In fact, we can use $ D= D_{ \omega } $ and $ f : \hat{ \funcf} ( D_{ \omega  } ) \rightarrow D_{ \omega } $ 
 the unique embedding such that 
$ f \circ  \hat{  \funcf  } (f_{n, \omega } ) = f_{n+1, \omega } $ for all $n \in \omega $.
In our notation, $f$ is $ f_{ \omega, \omega +1 }^{-} $,  an isomorphism,
thus $  h_{ \omega } $ is an $\hat{  \funcf  }  $-morphism and  unique by the initiality of $ (D_{ \omega } , f ) $.

Assume that $  h_{ \beta } : D_{ \beta } \to E $ is the unique 
$\hat{  \funcf  }  $-morphism from 
$  ( D_{ \beta } , f_{ \beta, \beta +1 }^{-} ) $
into  $ (E,g) $, and in particular that
$  ( D_{ \beta } , f_{ \beta, \beta +1 }^{-} ) $ is an initial $ \hat{  \funcf  }  $-algebra.
Then 
\[  ( D_{ \beta + 1 } , f_{ \beta +1 , \beta +2 }^{-} ) 
=  ( \hat{\funcf}  (  D_{ \beta }  ),  \hat{  \funcf  }   ( f_{ \beta, \beta +1 }^{-}  )) \]
is an initial $ \hat{  \funcf  }  $-algebra
and $  h_{ \beta }\circ  f_{ \beta, \beta +1 }^{-} = h_{ \beta  +1 } $ is the unique $\hat{  \funcf  }  $-morphism from 
$  ( D_{ \beta + 1 } , f_{ \beta +1 , \beta +2 }^{-} )  $ into $ (E,g) $.

Assume that $  h_{ \alpha } : D_{ \alpha } \to E $ is the unique $\hat{  \funcf  }  $-morphism from  
$  ( D_{ \alpha } , f_{ \alpha, \alpha +1 }^{-} ) $ into  $ (E,g) $ for every $ \alpha \in \beta $, 
for some limit ordinal  $\beta \leq \gamma_0$. Then $ f_{ \beta , \beta +1 } $ is an isomorphism, since every 
$ f_{ \alpha, \alpha +1 } $ is an isomorphism. This shows that $ h_{ \beta } $ is an $\hat{  \funcf  
}  $-morphism. Now, assume that $ h' : D_{ \beta } \rightarrow E $ is another $\hat{  \funcf  }  
$-morphism from  $ ( D_{ \beta }, f_{ \beta , \beta +1 }^{-} ) $ into $ (E,g) $. If $ \alpha \in 
\beta $, we have $ h' \circ f_{ \alpha , \beta } = 
g \circ \hat{  \funcf  } (h') \circ  f_{ \beta , \beta +1 }  \circ  f_{ \alpha , \beta } 
= g \circ \hat{  \funcf  } (h' \circ f_{ \alpha , \beta } ) \circ  f_{ \alpha, \alpha +1 }   $.
This shows that $ h' \circ f_{ \alpha , \beta } $ is an $ \hat{  \funcf  }  $-morphism from 
$  ( D_{ \alpha } , f_{ \alpha, \alpha +1 }^{-} ) $ into  $ (E,g) $ and that
 $ h_{ \alpha }   = h' \circ f_{ \alpha , \beta }  $. By definition of  $ h_{ \beta } $,
we have $ h_{ \beta } = h' $, so $ h_{ \beta } $ is unique.
\qed

\subsection*{Proposition~\ref{p_dense3}}

\begin{claimapp}
There exists a family $ \{ \Delta_{n}\}_{n\in\omega} $ of closed subsets  of $ D$, closed under binary lubs, such that 
\begin{enumerate}[(1)]
\item $ \Delta_{n } \subseteq \Delta_{n+1} $; and 
\item $ \Delta_{n} \cap  \mathcal{D}_{ \alpha}^{R} = \mathcal{D}_{n}^{R} $ for every  $ \alpha \geq 
\omega $. \end{enumerate}
\end{claimapp}

\proof
We define $  \Delta_{n} \subseteq D$ by induction on $ n \in \omega $.
We let $ \Delta_{0} := \emptyset $.
Take as induction hypothesis that $ \Delta_{n} \subseteq D$ is a closed subset, closed under binary lubs, such that
$ \Delta_{n } \subseteq \Delta_{n+1} $ and 
$ \Delta_{n} \cap  \mathcal{D}_{ \alpha}^{R} = \mathcal{D}_{n}^{R} $ for every  $ \alpha \geq 
\omega $. By induction on the structure of a strictly positive functor $ \funcf' $ with dense non-positive
parameters, we define 
a closed  subset $ \Delta_{n}^{ \funcf ' }  $ of $ \hat{ \funcf' } (D)  $, closed under binary lubs, 
satisfying
$ \Delta_{n}^{ \funcf ' } \cap  \funcf' (  \mathcal{D}_{ \alpha } ) ^{R} =   \funcf' (  \mathcal{D}_{ n} ) ^{R} $
for every $ \alpha \geq \omega $: \begin{enumerate}[$\bullet$]
\item{} 
If we have a constant functor, let $ \Delta_{n}^{ \mathcal{A}  }  := A $.
Clearly, $ A\cap \mathcal{A}^R =  \mathcal{A}^R$.
\item{} 
For the identity functor, let $ \Delta_{n}^{ \id }  :=  \Delta_{n}$.
Then $\Delta_n^\id\cap \mathcal{D}_{ \alpha }^{R}=  \mathcal{D}_{n}^{R} $ by the induction hypothesis on $n$.
\item{}
Let $ \Delta_{n}^{ \funcf_{0} + \funcf_{1} } :=  \bigcup_{i=0,1}\{ (i,x) : x \in   \Delta_{n}^{ \funcf_{i} } \} \cup \{ \bot \} $.
The verification is straight-forward.
\item{}
Let $ \Delta_{n}^{ \funcf_{0} \times \funcf_{1} } :=\Delta_{n}^{ \funcf_{0} }  \times \Delta_{n}^{ \funcf_{1} } $.
The verification is straight-forward.
\item{}
Let $ \Delta_{n}^{   [ \mathcal{B} \rightarrow  \funcf_{1} ] } 
:= \{  x \in [ B \rightarrow \hat{  \funcf_{1} } (D) ] : x  [ B  ] \subseteq \Delta_{n}^{ \funcf_{1} }  \} $.
We verify the set equality:
If for some  $ \alpha \geq \omega $, we have $ x \in  [  \mathcal{B} \rightarrow  \funcf_{1} (  \mathcal{D}_{ 
\alpha }  ) ] ^{R}  \cap \Delta_{n}^{ [ \mathcal{B} \rightarrow  \funcf_{1} ] } $
and $ b_{1} \approx_{  \mathcal{B} } b_{2} $, 
then $ x ( b_{1} )  \approx_{ \alpha }    x( b_{2} ) $.
Furthermore $ x ( b_{1} ) ,   x( b_{2} ) \in  \Delta_{n}^{ \funcf_{1} } \cap  \funcf_{1} (  \mathcal{D}_{ \alpha  } ) ^{R} =   \funcf_{1} (  \mathcal{D}_{ n} ) ^{R} $, which implies $ x ( b_{1} )  \approx_{n }    x( b_{2} ) $.
This shows that $ x \in [ \mathcal{B} \rightarrow  \funcf_{1} ( \mathcal{D}_{n} ) ]^{R} $.

If $ x \in [ \mathcal{B} \rightarrow  \funcf_{1} ( \mathcal{D}_{n} ) ]^{R} $,
then $ x [  \mathcal{B}^{R} ] \subseteq \funcf_{1} ( \mathcal{D}_{n} )^{R}  $.
This implies  $ x [ B ] \subseteq  \textrm{cl} ( \funcf_{1} ( \mathcal{D}_{n} )^{R} )  
 \subseteq \Delta_{n}^{ \funcf_{1} }   $, since $ \mathcal{B}^{R}$ is dense in
 $\mathcal{B}$. Hence, $ x \in  \Delta_{n}^{   [ \mathcal{B} \rightarrow  \funcf_{1} ] } $.
\end{enumerate}
In particular, this holds for the functor $\funcf $.
We let  $ \Delta_{n+1} := \Delta_{n}^{ \funcf } $. By the isomorphism of $ \hat{ \funcf } (D)  $  and $ D $, this 
can be considered as a closed subset of $D $ which is also closed under binary lubs. Then, for every $ \alpha \geq \omega $, 
\[   \mathcal{D}_{ n+1}^{R}  = \funcf (\mathcal{D}_n)^R
=  \Delta_{n}^{ \funcf  } \cap  \funcf(  \mathcal{D}_{ \alpha } ) ^{R}
=  \Delta_{n+1}  \cap   \mathcal{D}_{ \alpha +1 }^{R}
= \Delta_{n+1}    \cap   \mathcal{D}_{ \alpha  }^{R} \]
since $  \mathcal{D}_{ n+1}^{R} \subseteq \mathcal{D}_{ \alpha  }^{R} $.
\qed

\begin{claimapp}
Let $ n \geq 1 $. Then there exists a continuous map $ r_{n} : D \rightarrow D $ 
such that \begin{enumerate}[(1)]
\item $ \Delta_{n} = \{ x \in D:  r_{n } (x) = x \}$; and 
\item if  $ \alpha \geq n $, then $ r_{n } : \mathcal{D}_{ \alpha } \rightarrow  
\mathcal{D}_{n}  $ is equivariant. \end{enumerate} 
\end{claimapp}

\proof We prove the claim by induction on $n$.
In the beginning of the proof of proposition~\ref{p_dense3}, we make the assumption that $\funcf $ is non-trivial.
Therefore, we begin by constructing a continuous map 
$  r_{0}^{ \funcf' } : \funcf'  ( D ) \to  \funcf' ( D )  $
satisfying
$ \Delta_{0}^{ \funcf' } = \{ x \in \hat{ \funcf' } (D) :  r_{0}^{\funcf'}  (x) = x \} $,
by induction on the structure of  a non-trivial, strictly positive  $ \funcf' $.
Non-triviality means that we only have to consider the four induction steps  of lemma~\ref{l_trivialfunctor}.
\begin{enumerate}[$\bullet$]
\item{} If $ \funcf ' $ is atomic, it is a constant functor and equal to some dense $ \mathcal{A} $.
Let  $ r_{0 }^{ \funcf' } := \id_{ A } $. Verification of the set-equality above is trivial.
\item{} If  $  \funcf_{0} +  \funcf_{1}  $ is non-trivial, then at least one of $  \funcf_{0} $ and $ \funcf_{1}$ is non-trivial.
If both are  non-trivial, let $  r_{0 }^{  \funcf_{0} +  \funcf_{1} }:=  r_{0 }^{\funcf_{0}} +  r_{0 }^{\funcf_{1}}  $. Verification is 
straight-forward.
If one of them, say $ \funcf_{1} $, is trivial, we fix  $ x_{0} \in \funcf_{0}  ( \mathcal{D}_{ 0 } )^{R} $.
Let $  r_{0 }^{  \funcf_{0} +  \funcf_{1} } $ be the strict function which maps
$ (0,x) $ to $ (0,  r_{0 }^{\funcf_{0}} (x)  ) $ and $ (1,x) $ to $ (0, x_{0}  )$.
This is a continuous function, since $ r_{0  }^{\funcf_{0}} $ is continuous by the induction hypothesis.
Then $ r_{0 }^{  \funcf_{0} +  \funcf_{1} }(0,x) =(0,x) $ if and only if  $ r_{0  }^{\funcf_{0}}(x)=x$, again by
the induction hypothesis. Keeping in mind that $\Delta_0^{\funcf_1} = \emptyset$ (because
$\funcf_1$ is trivial), we see that this verifies the set equality. 

\item{} 
If $  \funcf_{0} \times \funcf_{1}  $ is non-trivial, then both  $  \funcf_{0} $ and $ \funcf_{1}$ are non-trivial,
and we let $  r_{0  }^{  \funcf_{0} \times \funcf_{1}   }  :=   r_{0  }^{\funcf_{0}}  \times r_{0 }^{\funcf_{1}}  $.
Verification is again straight-forward.

\item{} 
Finally, if $  [\mathcal{B} \to \funcf_{1} ] $ is non-trivial  then $\funcf_{1} $ is non-trivial, so let
$ r_{0 }^{ [ \mathcal{B} \to \funcf_{1} ] }  := ( \id_{ B  } \rightarrow  r_{0 }^{ \funcf_{1} } )$.
Then $ x \in \Delta^{ [ \mathcal{B} \to \funcf_{1} ] }_{0 } $ if and only if 
$ x [  B ]  \subseteq  \Delta_{0}^{  \funcf_{1} } $, which by the induction hypothesis is equivalent to
$ r_{0}^{ \funcf_{1}  } (x (b) ) = x(b)  $ for every $ b \in B$, that is
$  r_{0 }^{ \funcf_{1} } \circ x  = x $.
\end{enumerate}
It is easily verified by structural induction on $ \funcf' $ that 
$ r_{0}^{ \funcf'} : \funcf'(  \mathcal{D}_{ \alpha} ) \rightarrow  \funcf'(  \mathcal{D}_{0 } )$ is equivariant.
The only case for which it is not immediate is  $ \funcf' =  \funcf_{0} +  \funcf_{1} $ with $ \funcf_{1} $ trivial:
If $ ( 0,x ) \approx_{  \funcf'(  \mathcal{D}_{ \alpha} ) } (0,y) $, then 
$  r_{0  }^{\funcf_{0}} (x) \approx r_{0 }^{\funcf_{0}} (y) $ by the induction hypothesis, and
if  $  x \in \funcf_{1} ( \mathcal{D} )^{R} $, then   $  r_{0 }^{  \funcf_{0} +  \funcf_{1} } (1,x) =(0, x_{0} )$.

In particular, this shows that we have a continuous map  $  r_{0 }^{ \funcf }  :D \to D $ satisfying
\[ \Delta_1 = \Delta_0^\funcf = \{ x \in D :   r_{0 }^{ \funcf } (x) = x \} .\]
For the induction start $ n= 1 $,
we let $ r_{1 }  :=  r_{0 }^{ \funcf }  $.
Then $ r_{1} :  \mathcal{D}_{ \alpha  } \rightarrow \mathcal{D}_{1} $  is equivariant by induction on $ \alpha \geq 1$:
\begin{enumerate}[$\bullet$]
\item
Let $ \alpha = \alpha' + 1 $.  Then  $ r_{0}^{ \funcf} : \funcf(  \mathcal{D}_{ \alpha'} ) \rightarrow  \funcf (  \mathcal{D}_{0 } )$ equivariant.
\item
Let $ \alpha $ be a limit ordinal. If  $ x \approx_{ \alpha } y $, then $ x \approx_{ \alpha' } y $ for some $ \alpha' \in \alpha $.
By the induction hypothesis, we have $r_1(x) \approx_1 r_1(y) $. This shows that $r_1$ is equivariant.
\end{enumerate}
For the induction step, assume that the claim holds for $n$.
Recall that strictly positive operations are functorial in $ \cdomain $ as well,
 so we may consider 
$ \hat{\funcf} $ as a functor over $ \cdomain $.
Then $ r_{n+1  } : = \funcf ( r_{n} )  $ is continuous, and
$  \Delta_{n+1} = \{  x \in D :   r_{n+1 }(x) = x \} $,
This can be verified by an induction on the structure of a non-trivial, strictly positive $ \funcf' $ as for $n=1 $ above,
since $\funcf $ is assumed to be non-trivial, and all the induction steps of lemma~\ref{l_trivialfunctor} is covered above.
Moreover, $ r_{n+1} :  \mathcal{D}_{ \alpha +1 } \rightarrow \mathcal{D}_{ n+1} $ is equivariant for every $ \alpha \geq n $, because $ \funcf $ is a functor over $ \cdomainwp $ as well,
and this extends to every limit ordinal  as it did for $ n=1 $. 
\qed

\begin{claimapp}
Let $ p \in D_{c} $ and  assume that $ \upset{p} \cap \mathcal{D}_{ \alpha }^{R} \neq \emptyset $ for some $\alpha \in \gamma$. 
Then $ p \in \bigcup_{n \in \omega } \Delta_{n} $.
\end{claimapp}

\proof
We prove this claim by transfinite induction on $ \alpha $.

\begin{enumerate}[$\bullet$]
\item
Trivially true for $ \alpha = 0 $, since $ \mathcal{D}_{ 0 }^{R} = \emptyset $ .

\item
Assume that the claim holds for $ \alpha $.
By induction on the structure of a strictly positive $ \funcf ' $ with dense non-positive
parameters, we prove that $   \upset{p} \cap  \funcf' ( \mathcal{D}_{ \alpha } )^{R} \neq \emptyset  \Rightarrow  
p \in \bigcup_{n \in \omega } \Delta_{n}^{ \funcf' } $ for all $ p \in \hat{ \funcf ' } (D)_c $: 
All cases are trivially verified, except for the exponentiation: Assume that
$  \upset{ \bigsqcup_{j \in J } [p^{j} ; q^{j} ] } \cap  [  \mathcal{B} \rightarrow  \funcf_{1} (  \mathcal{D}_{  \alpha }  ) ] ^{R}  \neq \emptyset$.
Then for each $j\in J$ , we have $ \upset{ p^{j} }\cap\mathcal{B}^R\neq\emptyset$, since $\mathcal{B}$ is dense. 
This implies that  $ \upset{ q^{j} } \cap \funcf_{1} (  \mathcal{D}_{  \alpha }  ) ^{R}   \neq \emptyset  $, 
for each $j \in J$, and by the induction hypothesis, we then have 
$ \{q^{j}\}_{j\in J} \subseteq \bigcup_{n \in \omega }   \Delta_{n}^{ \funcf_{1}} $.
Since $J $ is finite, we have $ \{q^{j}\}_{j\in J} \subseteq  \Delta_{n}^{ \funcf_{1}} $ for some $n\in\omega$, and this shows
that  $   \bigsqcup_{ j \in J } [p^{j} ; q^{j} ] \in  \Delta_{n}^{[ \mathcal{B} \rightarrow   \funcf_{1}]} $.
In particular, we have $  \mathcal{D}_{ \alpha +1  } = \funcf (   \mathcal{D}_{ \alpha  } ) $ and
$ \upset{p} \cap \mathcal{D}_{ \alpha +1  }^{R} \neq \emptyset  \Rightarrow  p \in \bigcup_{n \in \omega } \Delta_{n} $.

\item
Let $\alpha\in \gamma$ be a limit ordinal and assume that the claim holds for all $ \alpha' \in  \alpha $
If $ \upset{p} \cap \mathcal{D}_{ \alpha }^{R} \neq \emptyset $, 
then $ \upset{p} \cap \mathcal{D}_{ \alpha' }^{R} \neq \emptyset $ for some  $ \alpha' \in  \alpha $,
so $ p \in \bigcup_{n \in \omega } \Delta_{n} $.\qed
\end{enumerate}

\subsection*{ Lemma~\ref{l_eta2}}

\begin{claimapp}
Let $ p \in E_{c}^{ \funcf }  $ and $ r \in  \hat{ \funcf } (D)_{c} $.
Then $   \vartheta_{ \funcf } (p) \sqsubseteq r \Leftrightarrow p \sqsubseteq  \eta_{ \funcf } (r) $.
\end{claimapp}

\proof
Observe that since $ \eta_{ \funcf } $ is a strict map, we have
$ \vartheta_{ \funcf } (p) = \bot $ if and only if $ p = \bot $, 
so we prove the result for 
$  \bigsqcup_{j  \in J } [ p^{j}; q^{j} ] \in E_{c}^{ \funcf } $ and 
$ (i,r) \in \hat{ \funcf } (D)_{c} \setminus \{ \bot \} $.

Let  $   \vartheta_{ \funcf } ( \bigsqcup_{j  \in J } [ p^{j}; q^{j} ] ) \sqsubseteq (i,r) $.
Then $ i $ is the index determined by $  \bigsqcup_{j  \in J } [ p^{j}; q^{j} ] $ as above
and $ \vartheta_{ \funcf_{i} }  ( \bigsqcup_{ j \in J' } [ p^{j}_{i}; q^{j} ] ) \sqsubseteq r $ for every 
$ J' \subseteq J $ with $\{ p^{j} \}_{ j \in J'} $ consistent.
By the induction hypothesis
$  \bigsqcup_{ j \in J' } [ p^{j}_{i}; q^{j} ] \sqsubseteq  \eta_{ \funcf_{i}} ( r  ) $
for every such $J' \subseteq J $. In particular, if $ j \in J$, then
\[ q^{j} \sqsubseteq  \Eval ( \eta_{ \funcf_{i}} ( r )  ,  p^{j}_{i}  ) =  \Eval ( \eta_{ \funcf} ( i, r )  ,  p^{j}  )  .  \]
Hence, $ \bigsqcup_{j  \in J } [ p^{j}; q^{j} ] \sqsubseteq  \eta_{ \funcf }  (i, r  ) $.

Let  $ \bigsqcup_{j  \in J } [ p^{j}; q^{j} ] \sqsubseteq  \eta_{ \funcf }  (i, r  ) $.
Then  $  q^{j} \sqsubseteq  \Eval ( \eta_{ \funcf_{i}} ( r )  ,  p^{j}_{i}  ) $
for every $ j \in J $.
Moreover, $ i $ is the index determined by $  \bigsqcup_{j  \in J } [ p^{j}; q^{j} ] $, and
 $  \bigsqcup_{ j \in J' } [ p^{j}_{i}; q^{j} ] \sqsubseteq  \eta_{ \funcf_{i}} ( r  ) $
for $ J' \subseteq J $
on condition that $\{ p^{j} \}_{ j \in J'} $ is consistent.
This condition is essential, since it ensures that
\[ \bigsqcup_{ j \in J''} q^{j} \sqsubseteq 
\bigsqcup  \{  \Eval ( \eta_{ \funcf } (i,  r )  ,  p^{j}  )  : j \in J'' \}
\sqsubseteq   \bigsqcup \{ \Eval ( \eta_{ \funcf_{i}} ( r )  ,  p^{j}_{i}  ) : j \in J'' \} , \]
whenever $ J'' \subseteq J' $.
By the induction hypothesis,
$ \vartheta_{ \funcf_{i} }  ( \bigsqcup_{ j \in J' } [ p^{j}_{i}; q^{j} ] ) \sqsubseteq r $
for every such $ J' \subseteq J $,
so
$   \vartheta_{ \funcf } ( \bigsqcup_{j  \in J } [ p^{j}; q^{j} ] ) \sqsubseteq (i,r) $.
\qed

\begin{claimapp}
Let $ b \in \mathcal{B}^{R} $ and $p  \in \textrm{approx} (b) $.
Let $ \bigsqcup_{j  \in J } [  p^{j} ;  q^{j}  ]   \in  E_{c}^{  \funcf } $ be witnessed by
$  x \in [ \mathcal{B} \rightarrow \funcf_{1} (  \mathcal{D} ) ] ^{R} $.

Then $ \bigsqcup_{ j \in J } \{ [ p_{1}^{j} ; q^{j} ] : p_{0}^{j} \sqsubseteq p  \}  \in E_{c}^{ \funcf_{1} } $ is
witnessed by $ x(b) $.
\end{claimapp}

\proof
Let $j\in J $  and  $ t \in  \upset{p_{1}^{j}} \cap  \mathcal{T}_{ \funcf_{1}}^{R} $, and assume that $p^j \sqsubseteq p $.
Then we have $p^j =(p_0^j, p_1^j) \sqsubseteq (b,t)$.
By assumption, we then have $ q^{j} \prec [ \Eval ( \eta_{ \funcf} (x) , (b,t) ) ] $.
Since $  \Eval ( \eta_{ \funcf } (x) , (b,t) )   =  \Eval ( \eta_{ \funcf_{1}} (x (b)) ,t) $,
this implies $ q^{j} \prec [ \Eval ( \eta_{ \funcf_{1}} (x (b) ) ,t ) ] $.
\qed

\begin{claimapp}
Let $ \bigsqcup_{j  \in J } [  p^{j} ;  q^{j}  ]   \in  E_{c}^{  \funcf } $.

If $ J'$ is a subset of $ J $  such that $ \{   p_{0}^{j} \}_{ j \in J' } $ is consistent in $B$,
then $ \{  \vartheta_{ \funcf_{1} }  ( \bigsqcup_{ k \in J } \{ [ p_{1}^{k} ; q^{k} ] : p_{0}^{k} \sqsubseteq p_{0}^{j}  \}  ) \}_{ j \in J'} $
is consistent in $ \hat{ \funcf_{1} } (D)  $.
\end{claimapp}

\proof
Choose $  x \in [ \mathcal{B} \rightarrow \funcf_{1} (  \mathcal{D} ) ] ^{R} $ witnessing  that  
$ \bigsqcup_{j  \in J } [  p^{j} ;  q^{j}  ]   \in  E_{c}^{  \funcf } $.

Assume that $J' \subseteq J $ and that $ \{   p_{0}^{j} \}_{ j \in J' } $ is consistent in $B$.
Let $j \in J $ and choose $ b \in \upset{ \bigsqcup_{ j \in J'} p_{0}^{j} } \cap \mathcal{B}^{R} $. 
Then  $ \bigsqcup_{ k \in J }  \{  [ p_{1}^{k} ; q^{k} ] : p_{0}^{k} \sqsubseteq p_{0}^{j} \} \in E_{c}^{ \funcf_{1}} $ 
is witnessed by $x(b)$ by the previous claim.
Moreover, by the induction hypothesis, we have
$ \vartheta_{ \funcf_{1}} (  \bigsqcup_{ k \in J }  \{  [ p_{1}^{k} ; q^{k} ] : p_{0}^{k} \sqsubseteq p_{0}^{j} \} ) 
\prec_{ \funcf_{1} ( \mathcal{D} )} [ x(b) ] $.
This is sufficient as $ \funcf_{1} ( \mathcal{D} ) $ is local.
\qed

\begin{claimapp}
Assume $  \bigsqcup_{j  \in J } [  p^{j} ;  q^{j}  ] \sqsubseteq \bigsqcup_{ k \in K} [ p^{k} ; q^{k} ]  \in E_{c}^{ \funcf }$.
Then
\[ \bigsqcup_{ j \in J} [ p_{0}^{j} ;  \vartheta_{ \funcf_{1} }  ( \bigsqcup_{ l \in J } \{ [ p_{1}^{l} ; q^{l} ] : p_{0}^{l} \sqsubseteq p_{0}^{j} \} ) ]  \sqsubseteq
\bigsqcup_{ k \in K} [ p_{0}^{k} ;  \vartheta_{ \funcf_{1} }  ( \bigsqcup_{ m \in K } \{ [ p_{1}^{m} ; q^{m} ] : p_{0}^{m} \sqsubseteq p_{0}^{k} \} ) ]  . \]
\end{claimapp} 

\proof
Fix $ j \in J $. If $ l \in J $ with $ p_{0}^{l} \sqsubseteq p_{0}^{j} $, then
\[ q^{l} \sqsubseteq  \bigsqcup_{ k \in K } \{  q^{k}  : p^{k} \sqsubseteq p^{l} \} 
\sqsubseteq  \bigsqcup_{ k \in K } \{  q^{k}  : p^{k}_{0} \sqsubseteq p^{j}_{0} \wedge  p^{k}_{1} \sqsubseteq p^{l}_{1}   \} . \]
This shows that $ \bigsqcup_{ l \in J }  \{ [ p_{1}^{l} ; q^{l} ] : p_{0}^{l} \sqsubseteq p_{0}^{j} \} 
\sqsubseteq \bigsqcup_{ k \in K }  \{ [ p_{1}^{k} ; q^{k} ] : p_{0}^{k} \sqsubseteq p^{j}_{0}  \} $, and 
we have
\[ \vartheta_{ \funcf_{1} }  (  \bigsqcup_{ l \in J }  \{ [ p_{1}^{l} ; q^{l} ] : p_{0}^{l} \sqsubseteq p_{0}^{j} \}  ) \sqsubseteq
\bigsqcup_{ k \in K }  \{  \vartheta_{ \funcf_{1} }  ( [ p_{1}^{k} ; q^{k} ] ) : p_{0}^{k} \sqsubseteq p^{j}_{0}  \} \]
since $ \vartheta_{ \funcf_{1} }  $ is continuous.
\qed

\begin{claimapp}
Let  $ \bigsqcup_{j  \in J } [  p^{j} ;q^{j}  ] \in    E_{c}^{  \funcf } $
and let $ r \in D_c^\funcf $. Then
$ \bigsqcup_{j  \in J } [  p^{j} ;q^{j}  ]  \sqsubseteq  \eta_{ \funcf } (r)$ 
if and only if,
for every $j\in J$,
\[ \bigsqcup_{k \in J } \{ [  p_{1}^{k} ;q^{k}  ]  : p_{0}^{k} \sqsubseteq p_{0}^{j} \}
\sqsubseteq \eta_{ \funcf_{1} } (r ( p_{0}^{j}  ) ) . \]
\end{claimapp}

\proof
Assume $  \bigsqcup_{j  \in J } [  p^{j} ;q^{j}  ]  \sqsubseteq  \eta_{ \funcf } (r) $.
Fix $ j \in J $ and let $ J' := \{ k \in J :  p_{0}^{k} \sqsubseteq p_{0}^{j} \} $.
Then for each $ k \in J' $, we have 
\[ q^{k} \sqsubseteq \Eval ( \eta_{ \funcf }  (r) , p^{k} ) 
\sqsubseteq \Eval ( \eta_{ \funcf_{1} }  (r (p_{0}^{k} )) , p_{1}^{k} ) 
\sqsubseteq \Eval ( \eta_{ \funcf_{1} }  (r (p_{0}^{j} )) , p_{1}^{k} ) . \]
 Hence, $  \bigsqcup_{ k \in J' }  \{  [ p_{1}^{k};   q^{k} ] \sqsubseteq  \eta_{ \funcf_{1} }  (r (p_{0}^{j} )) $.

Let  $ j \in  J $ and assume that  $ \bigsqcup_{ k \in J }  \{  [ p_{1}^{k};   q^{k} ] :  p_{0}^{k} \sqsubseteq 
p_{0}^{j}  \} \sqsubseteq  \eta_{ \funcf_{1} }  (r (p_{0}^{j} ) )$. If $ k \in J $ and
$ p_{0}^{k} \sqsubseteq p_{0}^{j} $, then
$ q^{k} \sqsubseteq  \Eval ( \eta_{ \funcf_{1} }  (r (p_{0}^{j} )) , p_{1}^{k} ) $. In particular,
$ q^{j} \sqsubseteq  \Eval ( \eta_{ \funcf_{1} }  (r (p_{0}^{j} )) , p_{1}^{j} ) =  \Eval ( \eta_{ \funcf }  (r) , p^{j} ) $.
This holds for all $ j \in J $, so $  \bigsqcup_{j  \in J } [  p^{j} ;q^{j}  ]  \sqsubseteq  \eta_{ \funcf } (r) $.
\qed

\subsection*{Lemma~\ref{l_eta4}}

\begin{claimapp}
$\mathcal{U} $ is dense and admissible.
\end{claimapp}
\proof
$\mathcal{U} $ is obviously dense, since $\mathcal{T} $ is dense.

We can characterise $  \mathcal{U} $  as the domain-per $ [ \mathcal{N} \rightarrow_{ \bot } \mathcal{T} ] $, since strict
functions $x: \mathbb{N}_\bot \to T $ are interchangeable with countable sequences over $ T $.
We know that $\mathcal{N}$ is dense and admissible and that $\mathcal{T}$ is admissible.
By lemma~\ref{l_adm5}, this means that $\mathcal{U} $ is  admissible.
\qed

\begin{claimapp}
Let $(x,u),(x',u')\in D \times U $ and assume  $ (x,u) \sqsubseteq (x',u') $.

Then $ M_{(x,u)} \leq M_{(x',u')} $, and 
$z^{m}_{ (x,u) } \sqsubseteq z^{m}_{ (x',u') }$
for every $ m \leq M_{(x,u) } $ finite.
Moreover, if $ M_{ (x,u) } <  M_{ (x',u') } $, then  
$z^{M_{(x,u)}}_{ (x,u) } = \bot $.
\end{claimapp}

\proof
The monotonicity of $\eta$ gives $z^{m}_{ (x,u) } \sqsubseteq z^{m}_{ (x',u') }$ by an easy induction on $m$.
This implies that
\[ z^{m}_{ (x',u') } \in  \biguplus_{ n \leq N}  A_{n}  \Rightarrow  z^{m}_{ (x,u) }   \in  \biguplus_{ n \leq N}  A_{n} ,\] 
which shows that $ M_{(x,u)} \leq M_{(x',u')} $.

If $ M= M_{ (x,u) } < M_{ (x',u') } $, then $z^{M}_{ (x,u) }  \in  \biguplus_{ n \leq N}  A_{n} $ and
$z^{M}_{ (x',u') }  \notin \biguplus_{ n \leq N}  A_{n} $.
Since $z^{M}_{ (x,u) }  \sqsubseteq z^{M}_{ (x',u') } $,
this leaves $z^{M}_{ (x,u) }  = \bot $ as the only possibility.
\qed

\begin{claimapp}
Let $ \Delta $ be a non-empty directed subset of $ D \times U $.
For each $ m \in \omega $, let $ \Delta^{m} $ be the subset 
$ \{ (x,u) \in \Delta :  M_{(x,u) }\geq m \} $. 

If $ m \leq  M_{   \sqcup \Delta  }$ is finite, then
\begin{enumerate}[(1)]
\item $ \Delta^{m}  $ is  directed  with $ \bigsqcup  \Delta^{m} =  \bigsqcup \Delta $; and
\item $ \{ z^{m}_{ (x,u)}  : (x,u) \in \Delta^{m}  \} $ is directed with least upper bound $ z^{m}_{ \sqcup \Delta } $.
\end{enumerate}
\end{claimapp}

\proof
We have seen that $ (x,u ) \sqsubseteq (x',u') \Rightarrow M_{ (x,u)} \leq M_{ (x',u') } $,
which shows that $ \Delta^{m}  $ is an upwards-closed subset of $ \Delta $. Hence, $\Delta^m$ 
is directed with $ \bigsqcup  \Delta^{m} =  \bigsqcup \Delta $ whenever it is non-empty.

We prove $  \Delta^{m} \neq \emptyset $ and the second part of the claim  simultaneously by induction on $ m $:
\begin{enumerate}[$\bullet$] 
\item Let $m=0$.  Then  $ \Delta^{0} =   \Delta \neq \emptyset$.
Moreover, $z^{0}_{\sqcup \Delta} =  \bigsqcup  \{ z^{0}_{ (x,u)}  : (x,u) \in \Delta \} $, since
$\eta$ is continuous.
\item Let $ m < M_{   \sqcup \Delta  }$ and assume $  \Delta^{m} \neq \emptyset $ and
$ z^{m}_{ \sqcup \Delta } = \bigsqcup \{ z^{m}_{ (x,u)}  : (x,u) \in \Delta^{m}  \} $.
Then $ z^{m}_{ \sqcup \Delta }   \notin  \biguplus_{ n \leq N}  A_{n} $, 
and since $  \biguplus_{ n \leq N}  A_{n} $ is a closed subset,
there exists some $  (x,u) \in \Delta^{m} $ such that $ z^{m}_{ (x,u)} \notin  \biguplus_{ n \leq N}  A_{n} $
and  $ m < M_{ (x,u) }$. 
This shows that $ \Delta_{m+1} \neq \emptyset $.
Furthermore,
\begin{eqnarray*}
z^{m+1}_{ \sqcup \Delta } 
&= & 
\Eval ( \eta ( \bigsqcup \{ d^{m}_{ (x,u) } : (x,u) \in \Delta_{m} \} ) , \bigsqcup \{ u_{m+1}: (x,u ) \in   \Delta_{m} \} )  \\
&= &
\Eval ( \eta ( \bigsqcup \{ d^{m}_{ (x,u) } : (x,u) \in \Delta_{m+1} \} ) , \bigsqcup \{ u_{m+1}: (x,u ) \in   \Delta_{m+1} \} )  \\
&= &
 \bigsqcup \{ \Eval ( \eta (  d^{m}_{ (x,u) }) ,  u_{m+1}   ) : (x,u ) \in   \Delta_{m+1} \}  \\
& = & 
\bigsqcup \{ z^{m+1}_{ (x,u ) }  : (x,u) \in \Delta_{m+1}
\}.\rlap{\hbox to 212 pt{\hfill\qEd}} 
\end{eqnarray*}
\end{enumerate}

\begin{claimapp}
Let $ x \in  \mathcal{D}^{R} $ with $ \rank (x) = \alpha + 1 $. Then
$  \eta (x) \in [  \mathcal{T}  \rightarrow \biguplus_{ k \in K  } \funcf^{k} ( \mathcal{D}_{ \alpha } ) ]^{R} $.
\end{claimapp}

\proof
First note that the operation
$ [  \mathcal{T}  \rightarrow \biguplus_{ k \in K } \funcf^{k} ( \cdot  ) ] $
is strictly positive and functorial over $ \cclcdomwp $.
The equiembedding $ f_{\alpha } :  \mathcal{D}_{ \alpha }  \rightarrow  \mathcal{D} $ 
 is the identity map on $D$, so this gives us
$ [  \mathcal{T}  \rightarrow \biguplus_{ k \in K_{ \funcf} } \funcf^{k} ( \mathcal{D}_{ \alpha } ) ]^{R} 
\subseteq  [  \mathcal{T}   \rightarrow \biguplus_{ k \in K_{ \funcf} } \funcf^{k} ( \mathcal{D}  ) ]^{R} $.

Moreover,  the domain function $ \eta $ is independent of the per $\approx_\alpha$. Thus, if $  \rank (x) 
= \alpha + 1 $, then $ x \in  \mathcal{D}_{ \alpha + 1}^R  = \funcf (  \mathcal{D}_{ \alpha }  )^R $. This shows that
$ \eta (x) \in  [  \mathcal{T}  \rightarrow \biguplus_{ k \in K } \funcf^{k} ( \mathcal{D}_{ \alpha } ) ]^{R} $.
\qed

\begin{claimapp}
Let $(x,u),(x',u') \in D \times U $ and  assume that $ (x,u) \approx_{\mathcal{D}\times 
\mathcal{U}} (x',u') $.
Then $ M_{ (x,u)}  =M_{ (x',u')} < \omega $, and 
$ z^{m}_{ (x,u) } \approx z^{m}_{ (x',u') }$
for every $ m \leq M_{ (x,u)} $.
\end{claimapp}

\proof
By induction on $ m $, we have $ z^{m}_{ (x,u)} \approx z^{m}_{ (x',u')} $ since $\eta$ is
equivariant. This means that the respective evaluation sequences  simultaneously reach 
$  \biguplus_{ n \leq N}  A_{n} $, and that $ M_{ (x,u) }  =M_{ (x',u')} $.

Moreover, if  $ m+1 < M_{ (x,u)} $ and $ \rank ( d^{m}_{ (x,u)}  )= \alpha + 1 $, then $ \rank (d^{m+1}_{ (x,u)} ) \leq  \alpha$.
If we assume $ M_{ (x,u)} = \omega $, we obtain an infinite evaluation sequence over $ \mathcal{D}^{R} $ and
the rank operation gives an infinite and strictly decreasing sequence of ordinals, which leads to a 
contradiction. Hence, $M_{ (x,u)} < \omega $.
\qed

\subsection*{Lemma~\ref{l_eta5}}

\begin{claimapp}
Let 
$  \bigsqcup_{ j \in J } [ p^{j}; ( q^{j}, n^{j} ) ] \sqsubseteq  \bigsqcup_{ k \in K } [ p^{k}; ( q^{k}, n^{k} ) ] $
(with $J \cap K = \emptyset$).
Then there exists a function
$ f: T ( \{ (p^{j} , n^{j}) \}_{j \in J} )  \rightarrow T ( \{ (p^{k} , n^{k}) \}_{k \in K} )  $ such that 
\begin{enumerate}[$\bullet$]
\item{}  $ p^{f( \varsigma)} \sqsubseteq p^{ \varsigma} $;
\item{}
$  \varsigma  \subseteq  \tau \Rightarrow  f( \varsigma )  \subseteq  f( \tau  )$;
\item{}
$  \vert f (  \varsigma ) \vert = \vert \varsigma \vert $; and
\item{} 
$ \varsigma $ is maximal in $ T ( \{ (p^{j} , n^{j}) \}_{j \in J })  $ 
$\Leftrightarrow $
$ f ( \varsigma ) $ is maximal in $ T ( \{ (p^{k} , n^{k}) \}_{k \in K} )  $.
\end{enumerate}
\end{claimapp}

\proof
Let $ \varsigma \in T ( \{ (p^{j} , n^{j}) \}_{j \in J} ) $.
We define a sequence $\varsigma'  $ over $K$ of length $\vert \varsigma \vert $
inductively:
Let $ \varsigma'_{0} := \{ k \in K : p_{0}^{k} \sqsubseteq p_{0}^{ \varsigma }  \} $ and let
$ \varsigma'_{m+1} := \{ k \in \varsigma'_{m} : p_{m+1}^{k} \sqsubseteq p_{m+1}^{ \varsigma }  \} $
for $m+1 < \vert \varsigma \vert $.

We show that  $ \varsigma'_{m}  \neq \emptyset $:
For each $ j \in J$,  there is some $k \in K $ such that $  p^{k} \sqsubseteq p^{j}$, since 
$ (q^{j} , n^{j} ) \neq \bot $ by assumption.
In particular, for each  $ j \in  \varsigma_{m} $  there is some $k \in K $ such that 
$  p^{k}_{m} \sqsubseteq  p^{j}_{m} \sqsubseteq p^{\varsigma}_m$ and $k \in  \varsigma'_{m}$. By  induction on $ m < \vert \varsigma \vert  $, this shows that
\[ \emptyset \neq \bigcap_{  n \leq  m  }  \{ k \in K :    p^{k}_{n}  \sqsubseteq p^{j}_{n} \} \subseteq \varsigma'_{m} .\]
Clearly, $\{  p_{m}^{k}\}_{ k \in \varsigma'_{m}} $ is bounded by $ p_{m}^{ \varsigma }  $ by definition of $  \varsigma'_{m}$.
Moreover, $ n^{j} = n^{k} $ whenever $  p^{k} \sqsubseteq p^{j} $, so if we take
$ M_{  \varsigma'  } = M_{ \varsigma } $ and $ n_{ \varsigma' } = n_{ \varsigma } $, we see that
$  \varsigma'  \in T ( \{ (p^{k} , n^{k}) \}_{k \in K} )  $.

Let $ f ( \varsigma ) :=  \varsigma' $. 
The inductive definition of $\varsigma'$ shows that $ f$ is monotone and that $  \vert f (  \varsigma ) \vert = \vert \varsigma \vert $.
Since $ M_{ f( \varsigma)  } = M_{ \varsigma } $, the function also preserves maximality. \qed

\begin{claimapp}
Let $ \varsigma $ be maximal 
with $ \vert \varsigma \vert  = M+1  $, and let $ j \in \varsigma_{ M } $.
If $ u \in \upset{ p^{ \varsigma } } \cap \mathcal{U}^{R}  $, 
then $ q^{ j } \prec_{(\biguplus_{ n \leq N}  \mathcal{A}_{n} )  } [ z^{ M }_{(x,u)} ]$, where
$z^{ M }_{(x,u)}$ is the evaluation result of $(x,u)$.
\end{claimapp}

\proof
Let $ u \in \upset{ p^{ \varsigma } } \cap \mathcal{U}^{R}  $.
Choose some $ u' \in  \upset{ p^{ j } } \cap \mathcal{U}^{R} $ such that $ u'_{m} = u_{m} $ for every
$ m \leq M$, which is possible since $ p^{j}_{m} \sqsubseteq  p^{ \varsigma }_{m} $.
By assumption, $   (q^{j} , n^{j} )  \prec_{\mathcal{E}} [ \Eval ( \bar{ \eta } (x) , u') ]  $, and
since $ n^{j} = n_{ \varsigma }$ codes an evaluation path of length $ M $,
this implies that $ q^{j} \prec_{( \biguplus_{ n \leq N}  \mathcal{A}_{n}) } [ z^{M}_{ (x,u' ) }] $.
However, the evaluation result depends only on the first $M+1 $ entries of $u' $, so
$ z^{M}_{ (x,u' ) }  = z^{M}_{ (x,u) } $.
\qed

\begin{claimapp}
Let  $ \varsigma $ be non-maximal and non-empty,
and assume that
\[ \forall \tau \in S( \varsigma ) \; \forall v \in   \mathcal{U}^{R}  \; 
(  p^{ \tau } \sqsubseteq  v   \Rightarrow q^{ \tau } \prec_{ (\biguplus_{k \in K_\funcf}  \funcf^k (\mathcal{D})) } 
[ z^{\vert \tau \vert -1}_{(x,v)} ]) .\]
Then $  \{ [p_{\vert \varsigma \vert}^{ \tau } ; q^{ \tau } ] : \tau \in S ( \varsigma ) \}  $ is consistent and 
if  $ u \in \upset{ p^{ \varsigma } } \cap  \mathcal{U}^{R} $, then 
\[ \bigsqcup \{ [p_{\vert \varsigma \vert}^{ \tau } ; q^{ \tau } ] : \tau \in S ( \varsigma ) \}  
\prec_{\eta_\funcf[ \funcf (\mathcal{D})]^d}
[\eta_\funcf( d^{\vert \varsigma \vert-1}_{(x,u) })].\] 
\end{claimapp}

\proof
Let $ u \in \upset{ p^{ \varsigma } } \cap  \mathcal{U}^{R} $.
We must show that for each  $ t \in \mathcal{T}_{\funcf}^{R} $ and each $\tau \in S ( \varsigma )$, we have
\[ \bigsqcup \{ q^{ \tau} : p_{\vert \varsigma \vert}^{ \tau } \sqsubseteq t  \} \prec_{( \biguplus_{k \in K_\funcf}  \funcf^k (\mathcal{D}))}
 [  \Eval ( \eta_\funcf ( d^{\vert \varsigma \vert-1}_{(x,u) } ) , t )   ] .\]
Let  $u' \in \mathcal{U}^{R} $ be $u$ with $u_{\vert \varsigma \vert}$ replaced by $t$. 
If $ t \in \upset{ p_{\vert \varsigma \vert}^{ \tau } } $, then $ u' \in \upset{ p^{ \tau } } $ and 
$ q^{ \tau } \prec_{( \biguplus_{k \in K_\funcf}  \funcf^k (\mathcal{D})) } [ z^{\vert \tau \vert -1}_{(x,u')} ] $ by assumption.
On the other hand,
$ z^{\vert \varsigma \vert}_{(x,u)} =  z^{\vert \varsigma \vert}_{(x,u')} $ 
and  $ \Eval ( \eta_\funcf ( d^{\vert \varsigma \vert-1}_{(x,u) } ) , t ) =  z^{\vert \varsigma \vert}_{(x,u')} 
 =  z^{\vert \tau \vert -1}_{(x,u')}$. 
\qed

\begin{claimapp}
Let $  \bigsqcup_{ j \in J } [ p^{j}; ( q^{j}, n^{j} ) ] \sqsubseteq  
\bigsqcup_{ k \in K } [ p^{k}; ( q^{k}, n^{k} ) ] \in  F_{c} $, and 
let $ f: T ( \{ (p^{j} , n^{j}) \}_{j \in J} )  \rightarrow T ( \{ (p^{k} , n^{k}) \}_{k \in K} )  $
be as in the claim above.
If $ \varsigma \in  T ( \{ (p^{j} , n^{j}) \}_{j \in J} ) $ is non-empty, then
$q^\varsigma \sqsubseteq q^{f(\varsigma)} $.
\end{claimapp}

\proof
Let $ \varsigma $ be maximal with $ \vert \varsigma \vert = M+1 $. If $ j \in \varsigma_{M}$, then
\[q^{j} \sqsubseteq
\bigsqcup_{ k \in K } \{ q^{k} : p^{k} \sqsubseteq p^{j} \} 
 \sqsubseteq  
\bigsqcup_{ k \in K } \{ q^{k} : \forall m \leq M \:  ( p^{k}_{m} \sqsubseteq p^{j}_{m}) \} 
 \sqsubseteq 
\bigsqcup_{ k \in f( \varsigma )_{M}}   q^{k}\ .
\]
This shows that $ q^{ \varsigma } = \bigsqcup_{ j \in \varsigma_M} q^j \sqsubseteq q^{ f( \varsigma ) } $. 

Let $ \varsigma $ be non-maximal 
and take as induction hypothesis that $ q^{   \tau } \sqsubseteq q^{ f( \tau ) } $ for all $ \tau \in S ( \varsigma ) $.
Moreover,  if $ \tau \in S ( \varsigma ) $, then $ f( \tau ) \in S ( f( \varsigma ) ) $ and
$ q^{ \tau } \sqsubseteq \bigsqcup_{ \upsilon \in  S( f ( \varsigma )) } \{ q^{ \upsilon } : 
p_{ \vert \varsigma \vert}^{ \upsilon } \sqsubseteq p_{ \vert \varsigma \vert}^{ \tau}  \} $.
This proves that
\[   \bigsqcup \{ [p^{ \tau }_{ \vert \varsigma \vert} ; q^{ \tau } ] : \tau \in S ( \varsigma ) \} \sqsubseteq  
\bigsqcup \{ [p^{ \upsilon}_{ \vert \varsigma \vert} ; q^{ \upsilon } ] : \upsilon \in S ( f( \varsigma)  ) \} .\]
Furthermore, $ \vartheta_\funcf $ is monotone, and 
$ n_{ \varsigma } = n_{ f ( \varsigma ) } $ so the respective evaluation paths are identical. 
This shows that
\[ (k^{ \vert \varsigma \vert -1}_{\varsigma } , \vartheta_\funcf ( \bigsqcup_{ 
\tau \in S ( \varsigma )}  [p^{ \tau }_{ \vert \varsigma \vert} ; q^{ \tau } ] ) )
\sqsubseteq (k^{ \vert \varsigma \vert -1}_{\varsigma } , \vartheta_\funcf ( \bigsqcup_{ 
\upsilon \in S ( f( \varsigma) )}  [p^{ \tau }_{ \vert \varsigma \vert} ; q^{ \tau } ] ) )\ .\eqno{\qEd}\]

\begin{claimapp}
Let $ \bigsqcup_{ j \in J } [ p^{j}; ( q^{j}, n^{j} ) ]  \in F_c$ and let $ r \in D_c$.
Then $ \bigsqcup_{ j \in J } [ p^{j}; ( q^{j}, n^{j} ) ]  \sqsubseteq \bar{ \eta } (r)  $ if and only if
$ q^{\varsigma } \sqsubseteq z^{ \vert \varsigma \vert-1}_{ (r, p^{ \varsigma } )}  $ for every non-empty $ \varsigma \in T $.
\end{claimapp}

\proof 
Assume that  $ \bigsqcup_{ j \in J } [ p^{j}; ( q^{j}, n^{j} ) ]  \sqsubseteq \bar{ \eta } (r)  $ 
and let $ \varsigma \in T $ be maximal.
For each $j \in J$, we have
$ (q^{j} , n^{j} ) \sqsubseteq \Eval (  \bar{ \eta } (r)  , p^{j} ) =  ( z^{M}_{(r,p^{j})} ,  n_{(r,p^{j})} ) $, with $M$ the length of
the evaluation path of  $(r,p^{j})  $, and this implies that $ q^j \sqsubseteq z^M_{(r,p^{j})} $ and $n^j =n_{(r,p^{j})} $.
Moreover, if $ j \in\varsigma_M$, then $ z^{M}_{(r,p^{j})} \sqsubseteq z^{M}_{ (r, p^{ \varsigma } )} $ since 
$p^j \sqsubseteq p^\varsigma $, and $  \vert \varsigma \vert = M+1 $ since $ n_\varsigma= n^j  $ which codes the evaluation path of  $(r,p^{j})  $.
This shows that $ q^\varsigma  \sqsubseteq z^{M}_{ (r, p^{ \varsigma } )} $ .

Assume that $ q^{\varsigma } \sqsubseteq z^{ \vert \varsigma \vert-1}_{ (r, p^{ \varsigma } )}  $  for every maximal $ \varsigma \in T $.
If  $ j \in J$, there is some maximal $ \varsigma $ with  $\varsigma_{M} = \{ j \} $. Then
$ q^{j}  =  q^{ \varsigma}  \sqsubseteq z^{M}_{ (r, p^{ \varsigma } )}  $ and $  n^{j}  = n_{ \varsigma } =  n_{(r,p^{j})} $;
in sum $(q^j, n^j) \sqsubseteq \Eval (\bar{\eta}(r), p^j) $.  This shows that 
$ \bigsqcup_{ j \in J } [ p^{j}; ( q^{j}, n^{j} ) ]  \sqsubseteq \bar{ \eta } (r)  $.

Let  $ \varsigma \in T $ be non-maximal and assume that
$ q^{\varsigma } \sqsubseteq z^{ \vert \varsigma \vert-1}_{ (r, p^{ \varsigma } )}  $.
Then 
\[ \vartheta_\funcf ( \bigsqcup_{ \tau \in S ( \varsigma )}  [p^{ \tau }_{ \vert \varsigma \vert} ; q^{ \tau } ] )  
\sqsubseteq d^{ \vert \varsigma \vert-1}_{ (r, p^{ \varsigma } )} \]
which, since $\eta_\funcf$ is the upper adjoint of $\vartheta_\funcf$ implies that
\[ \bigsqcup_{ \tau \in S ( \varsigma )}  [p^{ \tau }_{ \vert \varsigma \vert} ; q^{ \tau } ]  
\sqsubseteq \eta_\funcf ( d^{ \vert \varsigma \vert-1}_{ (r, p^{ \varsigma } )}) .\]
In particular, for each $  \tau \in S ( \varsigma )$, we have 
$ q^\tau \sqsubseteq \Eval ( \eta_\funcf ( d^{ \vert \varsigma \vert-1}_{ (r, p^{ \varsigma } )}), p^\tau_{ \vert \varsigma \vert}) =
  z^{\vert \varsigma \vert}_{ (r, p^{ \tau } )}  = z^{\vert \tau \vert-1}_{ (r, p^{ \tau } )}  $ .

Let  $ \varsigma \in T $ be non-maximal and assume that $  q^{ \tau } \sqsubseteq z^{\vert \tau \vert-1}_{ (r, p^{ \tau } )}  $
for all extensions $\tau $ of $\varsigma $ and that if $ \tau $ is a maximal extension of $ \varsigma $, then
 $n_\tau $ codes the evaluation path of $ (r, p^{ \tau } ) $.
Choose some  maximal extension $ \tau $ of $ \varsigma  $.
Then $ n_\varsigma = n_\tau$ and $ k^{ \vert \varsigma \vert-1}_{\varsigma }= k^{ \vert \varsigma \vert-1}_{ (r, p^{ \varsigma } )}   $
since $p^\varsigma_m=p^\tau_m $ for $m <  \vert \varsigma \vert$.
Moreover, for each  $  \tau \in S ( \varsigma )$, we have
 $  q^{ \tau } \sqsubseteq z^{\vert \tau \vert-1}_{ (r, p^{ \tau } )} =  \Eval ( \eta_\funcf ( d^{ \vert \varsigma \vert-1}_{ (r, p^{ \varsigma } )}), p^\tau_{ \vert \varsigma \vert})   $. Thus,
\[ \bigsqcup_{ \tau \in S ( \varsigma )}  [p^{ \tau }_{ \vert \varsigma \vert} ; q^{ \tau } ]  
\sqsubseteq \eta_\funcf ( d^{ \vert \varsigma \vert-1}_{ (r, p^{ \varsigma } )}) ,\]
which since $\vartheta_\funcf$ is the lower adjoint of $\eta_\funcf$ shows that
\[ \vartheta_\funcf ( \bigsqcup_{ \tau \in S ( \varsigma )}  [p^{ \tau }_{ \vert \varsigma \vert} ; q^{ \tau } ] )  
\sqsubseteq d^{ \vert \varsigma \vert-1}_{ (r, p^{ \varsigma } )} .\]
Hence, $ q^{\varsigma } \sqsubseteq z^{ \vert \varsigma \vert-1}_{ (r, p^{ \varsigma } )}  $.
\qed

\begin{claimapp} 
Let $ y \in \bar{\eta}[ \mathcal{D}]^R $. Then  $  \bar{ \eta } (  \bar{ \vartheta } (y)) \approx_{ [ \mathcal{U} \rightarrow \mathcal{E} ]} y $.
\end{claimapp}

\proof
Fix some $ x \in \mathcal{D}^R $ such that $ \bar{ \eta } ( x) \approx_{ [ \mathcal{U} \rightarrow \mathcal{E} ] } y $.
Since $  \mathcal{D}$ is local and complete and $\bar{\eta} $ is equivariant, we can assume that $ x = \bigsqcup [x]$.
If $ q \in \textrm{approx} (y) $, then $  \bar{ \vartheta } (q)   \prec_{\mathcal{D}} [x] $ which implies 
 $  \bar{ \vartheta } (q)  \sqsubseteq x $. This shows that  $  \bar{ \vartheta } (y) =  \bar{ \vartheta } (\bigsqcup
\textrm{approx} (y)) \sqsubseteq x $.

Let $ u \in  \mathcal{U}^{R} $. Then  $ \Eval ( \bar{ \eta } (x ) , u ) \approx_{ \mathcal{E}} y(u) $ and since
  $ ( \bar{ \vartheta } , \bar{ \eta } ) $ is an adjunction pair,
\[ y(u) \sqsubseteq   \Eval ( \bar{ \eta } ( \bar{ \vartheta } (y ) ) , u )
 \sqsubseteq \Eval ( \bar{ \eta } (x ) , u ) .\]
This implies that $   \Eval ( \bar{ \eta } ( \bar{ \vartheta } (y ) ) , u )  \approx_{ \mathcal{E}} y(u) $ since $  \mathcal{E}$ is convex.
\qed

\subsection*{Proposition~\ref{p_qcb7}}

\begin{claimapp}
Let $\funcf $ and $\funcg$ be weakly equivalent strictly positive endofunctors over $ \cclcdomwp $.
Then there exist assignments $ \varphi  \mapsto  \varphi^{ \funcf , \funcg } $ 
and $ \varphi  \mapsto  \varphi^{ \funcg , \funcf } $ 
from the class of equivariant maps into itself with the following properties:
\begin{enumerate}[(1)]
\item if $ \varphi :  \mathcal{D}  \rightarrow   \mathcal{E}  $, 
then $ \varphi^{ \funcf , \funcg } : \funcf (  \mathcal{D} ) \rightarrow  \funcg (  \mathcal{E} ) $ and
 $ \varphi^{ \funcg , \funcf } : \funcg (  \mathcal{D} ) \rightarrow  \funcf (  \mathcal{E} ) $;
\item if $ ( \varphi , \chi ) $ is a weak isomorphism of domain-pers $  \mathcal{D} $ and $   \mathcal{E}  $, 
then $ (  \varphi^{ \funcf , \funcg } ,  \chi^{ \funcg , \funcf } ) $ is a weak isomorphism of $ \funcf (  \mathcal{D} )  $ and $ \funcg (  \mathcal{E}  ) $; and
\item if $ \varphi $, $ \chi$ are equivariant maps and $f ,g $ are equiembeddings which satisfy
$ g \circ \varphi = \chi \circ f $, then
$ \funcg ( g)  \circ \varphi^{ \funcf , \funcg } = \chi^{ \funcf , \funcg } \circ \funcf ( f ) $.
\end{enumerate}
\end{claimapp}

\proof
We define the assignments $ \varphi  \mapsto  \varphi^{ \funcf , \funcg } $  and $ \varphi  \mapsto  \varphi^{ \funcg , \funcf } $  
by simultaneous induction on the structure of $ \funcf $ and $ \funcg $.
The tedious but straight-forward verifications of properties 1-3 at each induction step are left for the reader.
\begin{enumerate}[$\bullet$]
\item
If $ \funcf = \funcg = \id $, let 
$ \varphi^{ \funcf , \funcg } :=   \varphi^{ \funcg , \funcf } := \varphi $.
\item
If $ \funcf $ and $ \funcg $ are constant functors, equal to $ \mathcal{A} $ and $ \mathcal{A}' $, respectively,
choose some weak isomorphism pair $ ( \psi_{ \mathcal{A},  \mathcal{A}' }  , \psi_{  \mathcal{A}',  \mathcal{A} } ) :  \mathcal{A} \rightarrow  \mathcal{A}' $,
and let $  \varphi^{ \funcf , \funcg } :=   \psi_{ \mathcal{A},  \mathcal{A}' }   $ and
 $  \varphi^{ \funcg , \funcf }  :=  \psi_{ \mathcal{A}',  \mathcal{A} }   $ for all equivariant $ \varphi $.

\item
If $ \funcf = \funcf_{0} + \funcf_{1} $ and  $ \funcg = \funcg_{0} + \funcg_{1}$
with $ \funcf_{i} $ and $ \funcg_{i} $ weakly equivalent for $ i = 0 $ and for $ i=1 $,
let $  \varphi^{ \funcf , \funcg } :=  \varphi^{ \funcf_{0} , \funcg_{0} } +  \varphi^{ \funcf_{1} , \funcg_{1} } $
and   $ \varphi^{ \funcg , \funcf }:=  \varphi^{ \funcg_0 , \funcf_0 } +  \varphi^{ \funcg_1 , \funcf_1 }$.

\item
If $ \funcf = \funcf_{0} \times \funcf_{1} $ and $ \funcg = \funcg_{0} \times \funcg_{1}$
with $ \funcf_{i} $ and $ \funcg_{i} $ weakly equivalent for $ i = 0 $ and for $ i=1 $,
let $  \varphi^{ \funcf , \funcg } :=  \varphi^{ \funcf_{0} , \funcg_{0} } \times  \varphi^{ \funcf_{1} , \funcg_{1} } $ 
and $ \varphi^{ \funcg , \funcf } :=  \varphi^{ \funcg_0 , \funcf_0 } \times  \varphi^{ \funcg_1 , \funcf_1 } $.

\item
If $ \funcf = [ \mathcal{B}  \rightarrow  \funcf_{1} ] $ and $ \funcg =  [ \mathcal{B}'  \rightarrow  \funcg_{1} ] $
with  $  \funcf_{1} $  and $ \funcg_{1} $ weakly equivalent,
choose some weak isomorphism pair $ ( \psi_{ \mathcal{B},  \mathcal{B}' }  , \psi_{  \mathcal{B}',  \mathcal{B} } ) :  \mathcal{B} \rightarrow  \mathcal{B}' $.
Let $  \varphi^{ \funcf , \funcg } := (  \psi_{  \mathcal{B}',  \mathcal{B} } \rightarrow   \varphi^{ \funcf_{1} , \funcg_{1} } ) $
and let  $  \varphi^{ \funcg , \funcf } := (  \psi_{  \mathcal{B},  \mathcal{B}' } \rightarrow   \varphi^{ \funcg_{1} , \funcf_{1} } ) $.\qed
\end{enumerate}

\begin{claimapp}
There are families $\{\varphi_{ \beta }: \mathcal{D}_{\beta} \to \mathcal{E}_{\beta} \}_{\beta\leq\gamma} $ 
and  $\{\chi_{ \beta }: \mathcal{E}_{\beta} \to \mathcal{D}_{\beta} \}_{\beta\leq\gamma} $ of
equivariant maps such that each $ ( \varphi_{ \beta } , \chi_{ \beta } ) $ is  a weak isomorphism pair.
\end{claimapp}

\proof 
We define families $\{\varphi_{ \beta }: \mathcal{D}_{\beta} \to \mathcal{E}_{\beta} \}_{\beta\leq\gamma} $ 
and  $\{\chi_{ \beta }: \mathcal{E}_{\beta} \to \mathcal{D}_{\beta} \}_{\beta\leq\gamma} $ of
equivariant maps such that
\begin{enumerate}[(1)]
\item each $ ( \varphi_{ \beta } , \chi_{ \beta } ) $ is  a weak isomorphism pair; and
\item  
$ \{ \varphi_{ \alpha } \}_{ \alpha \in \beta } $ is a uniform mapping (see definition~\ref{d_uniformmapping}) from 
$   ( \{ \mathcal{D}_{\alpha}  \}_{ \alpha \in  \beta }, \{ f_{\alpha , \alpha' } \}_{ \alpha \leq \alpha'  \in \beta }) $ 
to $   ( \{ \mathcal{E}_{\alpha}  \}_{ \alpha \in  \beta }, \{ g_{\alpha , \alpha' } \}_{ \alpha \leq \alpha'  \in \beta }) $ 
 if $ \beta \leq \gamma$ is a limit ordinal.
\end{enumerate}
The second point is necessary as part of the induction hypothesis for the definition of $\varphi_\beta$ when $\beta$ is a limit ordinal.

Since $\funcf$ and $\funcg $ are weakly equivalent, we can choose 
 assignments $ \varphi  \mapsto  \varphi^{ \funcf , \funcg } $ and $ \varphi  \mapsto  \varphi^{ \funcg , \funcf } $ as in the claim above.

First, we define $ \varphi_{ \beta } $ by transfinite induction on $ \beta $:
\begin{enumerate}[$\bullet$]
\item Let $ \varphi_{ 0}  := \id_{ D_{0} }  = \id_{E_0}$.
\item If $ \beta$ is a successor ordinal, let  $ \varphi_{ \beta } := \varphi_{ \beta-1 }^{ \funcf , \funcg } $.
\item If $\beta $ is a limit ordinal, then $\{\varphi_{ \alpha } \}_{ \alpha \in \beta } $ is a uniform mapping 
from 
$   ( \{ \mathcal{D}_{\alpha}  \}_{ \alpha \in  \beta }, \{ f_{\alpha , \alpha' } \}_{ \alpha \leq \alpha'  \in \beta }) $ 
to $   ( \{ \mathcal{E}_{\alpha}  \}_{ \alpha \in  \beta }, \{ g_{\alpha , \alpha' } \}_{ \alpha \leq \alpha'  \in \beta }) $
by the induction hypothesis.
Let  $ \varphi_{ \beta } $ be the unique equivariant map from  $ \mathcal{D}_{\beta} $ into $\mathcal{E}_{\beta}$ such that
$ \varphi_{ \beta }\circ f_{\alpha , \beta} = g_{\alpha, \beta} \circ \varphi_\alpha $ for all $\alpha\in\beta$. This map exists by
 lemma~\ref{l_equid3}.
\end{enumerate}
We define $\chi_\beta$ symmetrically (just by swapping $\funcf$ and $\funcg$). 

\begin{enumerate}[(1)]
\item
By a transfinite induction on $\beta$, we show that
$ ( \varphi_{ \beta  } , \chi_{ \beta  } )$ is a weak  isomorphism pair:
\begin{enumerate}[$\bullet$]
\item For a successor ordinal $\beta $, this follows from the claim above since
$ ( \varphi_{ \beta  } , \chi_{ \beta  } )=  ( \varphi_{ \beta-1  }^{ \funcf , \funcg } , 
\chi_{ \beta-1  }^{ \funcg , \funcf } )$.
\item If $\beta$ is a limit ordinal, then it follows from the induction hypothesis by symmetry that $\{\chi_\alpha \}_{\alpha \in \beta} $ 
is a uniform mapping. Thus,  $ ( \varphi_{ \beta  } , \chi_{ \beta  } )$ is a weak  isomorphism pair 
by lemma~\ref{l_equid3}. \end{enumerate}

\item
Let  $\beta$ be a limit ordinal. 
It is sufficient to prove that
$g_{\alpha,\alpha'}\circ \varphi_\alpha = \varphi_{\alpha'} \circ f_{\alpha, \alpha'} $ for all $\alpha\leq\alpha'\in\beta$,
since this implies that $\{\varphi_{ \alpha } \}_{ \alpha \in \beta } $ is a uniform mapping.
We prove this by a transfinite induction on $ \alpha $ and $\alpha'$:
\begin{enumerate}[$\bullet$]
\item $g_{0,\alpha'}\circ \varphi_0 = \varphi_{\alpha'} \circ f_{0, \alpha'} $ trivially for all $\alpha'$.
\item If $\alpha,\alpha'$ both are successor ordinals, then
\[ g_{\alpha, \alpha'  } \circ \varphi_{\alpha }  
= \funcg (  g_{\alpha-1, \alpha'-1 } ) \circ \varphi_{\alpha-1 }^{ \funcf , \funcg } 
 =  \varphi_{\alpha'-1  }^{ \funcf , \funcg } \circ \funcf (  f_{\alpha-1, \alpha'-1 }  ) 
 =  \varphi_{\alpha'  } \circ f_{\alpha , \alpha'  } .\]
\item If $\alpha'$ is a limit ordinal, then $  \varphi_{\alpha'} \circ f_{\alpha, \alpha'} = g_{\alpha,\alpha'}\circ \varphi_\alpha$
by definition of $\varphi_{\alpha'}$.
\item If $\alpha$ is a limit ordinal, choose arbitrary $\alpha''\in\alpha$. By the induction hypothesis, we have 
 $  \varphi_{\alpha'} \circ f_{\alpha'', \alpha'} = g_{\alpha'',\alpha'}\circ \varphi_{\alpha''}$. Then
\begin{eqnarray*}
(g_{\alpha,\alpha'} \circ \varphi_{\alpha})\circ f_{\alpha'', \alpha} &
= g_{\alpha,\alpha'} \circ ( g_{\alpha'',\alpha} \circ \varphi_{\alpha''} )\\
& = g_{\alpha'',\alpha'} \circ  \varphi_{\alpha''} \\
& =  \varphi_{\alpha'} \circ f_{\alpha'',\alpha'} \\
&= ( \varphi_{\alpha'} \circ f_{\alpha,\alpha'})  \circ f_{\alpha'',\alpha}
\end{eqnarray*}
which shows that $g_{\alpha,\alpha'}\circ \varphi_\alpha = \varphi_{\alpha'} \circ f_{\alpha, \alpha'} $ since $\alpha''$ was arbitrary.\qed
\end{enumerate}
\end{enumerate}

\section{Notation list} \label{app.notation}

\noindent Notation which is repeated in different proofs is tentatively put together in this list.

\begin{enumerate}[$\bullet$]

\item  $D_0$ is the trivial domain. 
\item $\mathcal{D}_0$ is the trivial domain-per.

\item $\mathbb{N}_\bot $ is the flat domain of  natural numbers.
\item $\mathcal{N} $ is the domain-per with $\mathbb{N}_\bot $ as domain and equality restricted to $\mathbb{N} $ as partial equivalence relation.

\item Let $\mathcal{D}= (D, \approx) $ be a domain-per.
\begin{enumerate}[$-$]
\item  $\mathcal{D}^R $ is the subspace $ \{ x \in D : x \approx x \}$ of $D$.
\item If $ x \in  \mathcal{D}^R $, then 
$[x] = \{ x' \in D : x \approx x' \} $.
\item $ \mathcal{Q} \mathcal{D}  $ is the quotient space of  $ \mathcal{D}^{R}$ under
the equivalence relation $\approx\vert_{ \mathcal{D}^R } $.
\item $(D, \mathcal{D}^R , \delta_{\mathcal{D}})$ is a quotient domain representation of $ \mathcal{Q} \mathcal{D}  $.
\item If $p \in D_c $ and $x \in  \mathcal{D}^R $, then $ p \prec_{  \mathcal{D} }  [x] $, if  there exists some $ x' \in  [x]  $ with $ p \sqsubseteq x' $.
\item $ \mathcal{D}^{d} $ is  the dense part of  $ \mathcal{D} $, see definition~\ref{d_densepart}.
\end{enumerate}

\item Let $f: \mathcal{D} \to \mathcal{E} $ be an equivariant map. 
\begin{enumerate}[$-$]
\item  $f^{\mathcal{Q}}:  \mathcal{Q}\mathcal{D} \to  \mathcal{Q}\mathcal{E} $ is the continuous map $ f^{\mathcal{Q}} ([x]) = [f(x)] $.
\item $ f [  \mathcal{D} ] $ is the image of $\mathcal{D} $ under $f$, see definition~\ref{d_image}.
\end{enumerate}

\item  Let $ ( \{ \mathcal{D}_{i}  \}_{ i \in I }, \{ f_{i,j} \}_{ i \leq j \in I })  $ be a directed system over $(I,\leq)$ in $ \cclcdomwp $.
\begin{enumerate}[$-$]
\item $D_i$ is the underlying domain of $\mathcal{D}_{i}$.
\item $D_I $ is $     \{ x \in \prod_{i \in I} D_{i} : \forall i,j \in I  ( i\leq j \to f^{-}_{i,j} ( x_{j} ) = x_{i}) \}  $ with the product order.
\item $ f^{-}_{i}: D_I \to D_i $ is the projection $ f^{-}_{i} (x) = x_{i} $.
\item 
 $ x \approx_I x'$, and this is witnessed by $i$, if
\[ \exists  i \in I \ 
(x_i \approx_i x'_i \wedge  \forall k\geq i \: (f_{i,k}(x_i)\approx_k x_k \wedge
f_{i,k}(x'_i)\approx_k x'_k)).\]
\item  $ \mathcal{D}_{I} $ is  $D_I  $ with per $ \approx_I $.
\end{enumerate}

\item $\hat{\funcf} : \cdom \to \cdom $ is the underlying functor of $\funcf : \cclcdomwp \to \cclcdomwp $, as described in remark~\ref{r_undfunctor}.

\item Let $\funcf : \cclcdomwp \to \cclcdomwp $ be a strictly positive functor. Assume that all non-positive parameters are dense.

Let $\mathcal{D} $ be a convex, local and complete domain-per.
\begin{enumerate}[$-$]
\item $ \funcf (  \mathcal{D} )^{d} $ is the dense part of $ \funcf (  \mathcal{D} ) $ with underlying domain $D^\funcf $, see definition~\ref{d_densepart}.
\item $ \mathcal{T}_\funcf  $ is the dense domain-per defined in the proof of lemma~\ref{l_eta1}. The underlying domain is $T_\funcf$.
\item $ \{ \funcf^{k} \}_{ k \in K_{\funcf} } $ is the set of  atomic subfunctors  of $ \funcf $ with
repetition allowed.
\item 
$ \eta_\funcf : \funcf (  \mathcal{D} )^{d}  \rightarrow  [   \mathcal{T}_\funcf  \rightarrow  \biguplus_{ 
k \in K_\funcf }  \funcf^{k} ( \mathcal{D}  ) ] $  is the equivariant and equi-injective map
defined in the proof of lemma~\ref{l_eta1}.
\item  $ E^{ \funcf } $ is the underlying domain of $  \eta_\funcf [  \funcf  ( \mathcal{D} ) ]^{d} $, the dense part of the image of  $ \funcf (  \mathcal{D} )^{d} $ under $ \eta_\funcf $, see definition~\ref{d_image}.
\item  $ \vartheta_\funcf :   E^{ \funcf } \rightarrow D^{ \funcf }  $ is the lower adjoint of $ \eta_\funcf $, defined in the proof of lemma~\ref{l_eta2}.
\end{enumerate}

\item Let $ \funcf : \cclcdomwp \to \cclcdomwp $ be a strictly positive functor. Assume that all  parameters are dense and admissible.
\begin{enumerate}[$-$]
\item  $  \mathcal{D} $ is the dense  least fixed point of $ \funcf $ defined in subsection~\ref{ss_density}. The underlying domain is $D$.
\item  $  \mathcal{T} $ is  the dense and admissible  domain-per $ \mathcal{T}_\funcf  $ defined in the proof of lemma~\ref{l_eta1}.  The underlying domain is $T $.
\item  $  \eta : \mathcal{D}  \rightarrow^{} [  \mathcal{T}  \rightarrow \biguplus_{ k \in K } \funcf^{k} ( \mathcal{D} ) ]$
is the equivariant and equi-injective map $\eta_\funcf $ as defined in the proof of lemma~\ref{l_eta1}.
\item  $  \mathcal{U} $ is  the dense and admissible  domain-per defined in the proof of lemma~\ref{l_eta4}. The underlying domain is $U$.
\item $  \mathcal{A}_{0}, \ldots ,  \mathcal{A}_{N} $ are the positive parameters of $ \funcf $. The underlying domains are $A_0, \dots , A_N $.
\item $ \mathcal{E} $ is the domain-per $ (  \biguplus_{ n \leq N}  \mathcal{A}_{n}   ) \otimes \mathcal{N} $. The underlying domain is $E$.
\item $ \bar{\eta } :  \mathcal{D} \rightarrow    [ \mathcal{U} \rightarrow \mathcal{E}  ]  $  is the equivariant and equi-injective map
defined in the proof of lemma~\ref{l_eta4}.
\item If $(x,u ) \in D \times U $,
\begin{enumerate}[$*$]
\item the evaluation sequence  $ \{ d^{m}_{ (x,u) }  \}_{ m <  M } $ over $ D$,
\item the evaluation path  $ \{ k^{m}_{ (x,u) }  \}_{ m <  M } $ over $ K_\funcf$ and
\item the evaluation result $ z_{(x,u)}^M   \in  \biguplus_{ n \leq N}  A_{n} $ 
\end{enumerate} are  defined in  the proof of lemma~\ref{l_eta4}.
\item $  \bar{\eta }[ \mathcal{D}]^{d}$ is the dense part of the image of  $  \mathcal{D} $ under $ \bar{\eta } $, see definition~\ref{d_image}.
\end{enumerate}

\end{enumerate}	

\section*{Acknowledgements}

\noindent The research is funded by  the Norwegian Research Council as part of the
project 'Computability and Complexity in Type Theory'. The work is
carried out mainly at the Department of Mathematics, University of Oslo, but in
part also at the Department of Computer Science, University of Wales Swansea.

I am grateful to  Dag Normann for bringing up the main problem treated in 
this paper and for his continuous encouragement and guidance.
I would also like to thank Jens Blanck for useful comments and suggestions 
on an early draft of this paper, and Ulrich Berger, Fredrik Dahlgren and 
John V. Tucker for enlightening discussions within the field. 

Finally, I am thankful to an anonymous referee for spotting a crucial mistake in the first 
submission of this paper, thus giving inspiration for this improved version.
I also appreciate the  referees' suggestions and comments  which have  helped me to increase the readability of the paper.

\end{document}